%% file: after_QST_proof_Jun2024.tex
\documentclass[pra,aps,floatfix,amsmath,superscriptaddress,twocolumn,longbibliography,nofootinbib]{revtex4-2}%\documentclass[prx,aps,floatfix,amsmath,superscriptaddress,twocolumn]{revtex4}
\usepackage{amssymb,enumerate}
\usepackage{graphicx}
\usepackage{graphics}
\usepackage{amsmath}
\usepackage{amsthm,bbm}
\usepackage{color}
\usepackage{dsfont}
\usepackage{hyperref}

\usepackage[capitalise,nameinlink]{cleveref}

\usepackage{qcircuit}

\usepackage{graphicx}

\usepackage{xcolor}
\hypersetup{
    colorlinks,
    linkcolor={red!50!black},
    citecolor={blue!50!black},
    urlcolor={blue!80!black}
}

\usepackage{xfrac}

\usepackage{amsmath}
\usepackage{tikz-cd}
\usepackage{empheq}

\usepackage{enumitem}

\usepackage{graphicx}% http://ctan.org/pkg/graphicx

\usepackage{hyperref}
\usepackage[capitalise]{cleveref}
\usepackage{youngtab}

\usepackage{tikz}
\usepackage{mathtools}

\newcommand{\SU}{\operatorname{SU}}

\newcommand{\e}{\operatorname{e}}

\renewcommand{\i}{\mathrm{i}}
\renewcommand{\j}{\mathrm{j}}

\newcommand{\R}{\mathbb{R}}

\newcommand{\Tr}{\operatorname{Tr}}

\newcommand{\<}{\langle}
\renewcommand{\>}{\rangle}

\renewcommand{\i}{\mathrm{i}}

\renewcommand{\i}{\mathrm{i}}

\newtheorem{thm}{Theorem}

\usepackage{lipsum}
\usepackage{lmodern}
\usepackage{tcolorbox}

\usepackage{amsfonts}
\usepackage{graphicx,graphics,epsfig,times,bm,bbm,amssymb,amsmath,amsfonts,mathrsfs}
\usepackage[normalem]{ulem}
\usepackage{setspace}
\usepackage{subfigure}
\usepackage{dsfont}
\usepackage{braket}
\usepackage{upgreek }
\usepackage{tikz}
\usepackage{natbib}
\usepackage{chngcntr}

\newtheorem{theorem}{Theorem}

\newtheorem{corollary}[theorem]{Corollary}

\newtheorem{lemma}[theorem]{Lemma}

\newtheorem{proposition}[theorem]{Proposition}

\makeatletter 
\renewcommand\onecolumngrid{% 
  \do@columngrid{one}{\@ne}%
  \def\set@footnotewidth{\onecolumngrid}%
  \def\footnoterule{\kern-6pt\hrule width 1.5in\kern6pt}%
}
\makeatother

\newcommand{\bes} {\begin{subequations}}
\newcommand{\ees} {\end{subequations}}
\newcommand{\bea} {\begin{eqnarray}}
\newcommand{\eea} {\end{eqnarray}}
\newcommand{\be} {\begin{equation}}
\newcommand{\ee} {\end{equation}}

\def\b{\lambda}

\def\>{\rangle}
\def\<{\langle}
\def\Tr{\textrm{Tr}}

\newcommand{\ketbra}[2]{|{#1}\rangle\langle{#2}|}
\newcommand{\ketbrasame}[1]{\ketbra{#1}{#1}}

\newcommand{\ignore}[1]{}

\newcommand{\SKexp}{\log_{(1+\sqrt{5})/2}2}

\newcommand{\swap}{\text{Sw}}
\newcommand{\iswap}{\i\swap}
\newcommand{\siswap}{\sqrt{\i\swap}}
\newcommand{\iSWAP}{\i\text{SWAP}}
\newcommand{\siSWAP}{\sqrt{\iSWAP}}

\newcommand{\CZ}{\text{C}Z}

\newcommand{\sct}[1]{^{\color{black}(#1)}}

\input{tikzsetting}

\newcounter{step}
\crefformat{step}{#2Step #1#3}
\crefmultiformat{step}{#2Steps~#1#3}%
{ and~#2#1#3}{, #2#1#3}{ and~#2#1#3}

\crefformat{section}{#2Sec.~#1#3}

\usepackage{chngcntr}

\begin{document}
\title{Synthesis of Energy-Conserving Quantum Circuits with  XY interaction}
	  
\author{Ge Bai}
\affiliation{Centre for Quantum Technologies, National University of Singapore, 3 Science Drive 2, Singapore 117543}
\affiliation{QICI Quantum Information and Computation Initiative, Department of Computer Science,
The University of Hong Kong, Pokfulam Road, Hong Kong
}
\affiliation{HKU-Oxford Joint Laboratory for Quantum Information and Computation}

 \author{Iman Marvian}
\affiliation{Duke Quantum Center, Departments of Electrical and Computer Engineering and Physics, Duke University, Durham, North Carolina 27708, USA}
\begin{abstract} 
We study quantum circuits constructed from $\sqrt{\i\text{SWAP}}$ gates and, more generally, from the entangling gates that can be realized with the XX+YY interaction alone.
Such gates preserve the Hamming weight of states in the computational basis, which means they  respect the global U(1) symmetry corresponding to rotations around the z axis. Equivalently,  assuming that the intrinsic Hamiltonian of each qubit in the system is  the Pauli Z operator, they conserve the total energy of the system.
% Assuming that the intrinsic Hamiltonian of each qubit in the system is the Pauli Z operator, such gates conserve the total energy of the system. Equivalently, they respect a global U(1) symmetry. 
We develop efficient methods for synthesizing circuits realizing any desired energy-conserving unitary using XX+YY  interaction with or without single-qubit rotations around the z-axis. Interestingly, implementing generic energy-conserving unitaries, such as CCZ and Fredkin gates, with 2-local energy-conserving gates requires the use of ancilla qubits.     
  When single-qubit rotations around the z-axis are  permitted, our scheme requires only a single ancilla qubit, whereas with the XX+YY interaction alone, it requires 2 ancilla qubits. In addition to exact realizations, we also consider approximate realizations and show how a general energy-conserving unitary can be synthesized using only a sequence of $\sqrt{\i\text{SWAP}}$ gates and 2 ancillary qubits, with arbitrarily small error, which can be bounded via the Solovay-Kitaev theorem. 
 Our methods are also applicable for synthesizing energy-conserving unitaries when, 
  rather than the XX+YY interaction,
  one has access to 
 any other energy-conserving 2-body interaction that is not diagonal in the computational basis, such as the Heisenberg exchange interaction. We briefly discuss the applications of these circuits in the context of quantum computing, quantum thermodynamics, and quantum clocks. 

  \end{abstract}

\maketitle
\section{Introduction}

In the field of quantum computing and other related areas, such as quantum control and quantum thermodynamics, one is often interested in implementing desired unitary transformations on a quantum system, e.g., on a finite number of qubits. Inspired by the classical circuit model, researchers in this field have developed circuit synthesis techniques to implement any desired unitary using elementary gate sets acting on a few qubits in the system \cite{kitaev2002classical, NielsenAndChuang, barenco1995elementary, deutsch1985quantum, deutsch1995universality, divincenzo1995two, Lloyd:95}. For instance, it has been shown that any unitary transformation on $n$ qubits can be implemented exactly with $\mathcal{O}(4^n)$ single-qubit and 2-qubit gates, such as Controlled-NOT (CNOT) gates \cite{mottonen2004quantum}. 
 However, these general circuit synthesis techniques do not take into account the specific properties of the desired unitaries, such as their symmetries. Such considerations can significantly reduce the number of required gates and also enable circuit realizations that are more noise-resilient. Additionally, the general circuit synthesis techniques do not distinguish between generic gates and the gates that can be realized with native interactions on a particular platform---an important property that makes the gates easier to implement and more robust against noise.
  
In this work, we study energy-conserving quantum circuits, which are circuits formed from single-qubit rotations around the z-axis and 2-qubit unitary gates $U$ that conserve the sum of Pauli $Z$ operators, such that 
\be
[U, Z\otimes \mathbb{I}+\mathbb{I}\otimes Z]=0\ .
\ee
Assuming the qubits have identical intrinsic Hamiltonian, with eigenstates $|0\rangle$ and  $|1\rangle$, such gates conserve the total intrinsic energy of the system. Hence, in this paper,  
 we refer to such unitary transformations  as energy-conserving unitaries   (See \cref{Sec:prem} for further details and definitions). 
Note that this condition can be equivalently  understood as a global U(1) symmetry, where the representation of symmetry on each qubit is $\e^{\i\theta Z}: \theta\in (-\pi,\pi]$.

\begin{figure}
\begin{minipage}{\linewidth}
  \Qcircuit @C=1em @R=.7em { 
   &   \multigate{1}{\i\text{Sw}}& \qw &   \multigate{1}{\i\text{Sw}}& \gate{\sqrt{X}} & \multigate{1}{\i\text{Sw}} & \qw &\qw \\
 & \ghost{i\text{Sw}} &  \gate{\sqrt{X}}& \ghost{i\text{Sw}} &  \qw& \ghost{i\text{Sw}}&\gate{\sqrt{X}}&\qw    }
\end{minipage}
\vspace{10mm}\\

\begin{minipage}{\linewidth}
\Qcircuit @C=1em @R=.7em {
\lstick{\ket{0}} 
   &   \multigate{1}{\i\text{Sw}}& \qw &   \multigate{2}{\i\text{Sw}}& \qw & \qw & \ket{0} \\
 & \ghost{i\text{Sw}} &  \multigate{1}{\i\text{Sw}} & \qw & \gate{S^\dag} & \qw \\ 
  & \qw &  \ghost{\i\text{Sw}}&  \ghost{\i\text{Sw}} & \gate{Z}&\qw  }   
\end{minipage}
   \caption{\textbf{SWAP from $\i\text{SWAP}$.} The top circuit is the standard way of implementing the SWAP gate using three   $\i\text{SWAP}=\exp(\i R \pi/2)$ gates, also denoted as $\i\text{Sw}$, where $R=(X\otimes X+Y\otimes Y)/2$ is the XY interaction \cite{schuch2003natural}.   (Note that SWAP and $\i\text{SWAP}$ are indeed equal, up to a global phase in the sector with Hamming weight 1.)   This circuit requires three $\sqrt{X}$ gates, which are not energy-conserving. 
   Interestingly, it turns out that even though SWAP is energy-conserving, unless one uses ancilla qubits,  such non-energy-conserving gates are unavoidable (See \cite{marvian2022restrictions}  and   \cref{Thm0}).  The bottom circuit realizes SWAP using   three  $\i\text{SWAP}$ gates together with two energy-conserving single-qubit gates, namely $S^\dag$ and  $Z$, and with one ancillary qubit.   
 The lack of non-energy-conserving unitaries makes the bottom circuit more resilient against certain types of noise, such as the fluctuations of the master clock that controls the timing of pulses \cite{ball2016role, marvian2022restrictions}.   }\label{Fig1}
\end{figure}
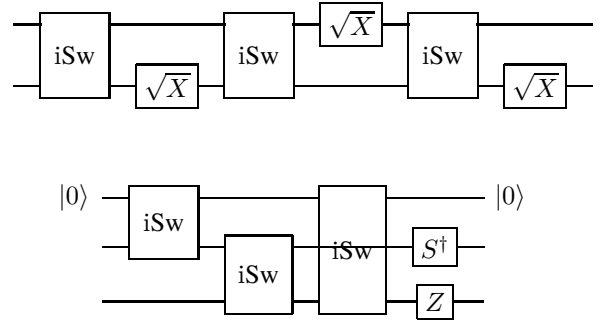

A canonical example of such an energy-conserving gate family, extensively discussed in this paper, is one that can be realized with the Hamiltonian $X\otimes X+Y\otimes Y$, also known as the XY interaction.  Besides its fundamental importance in the condensed matter theory, this interaction also plays a crucial role in quantum computing. Namely, it is the native interaction in various solid-state qubits, including  quantum 
dot spins  
\cite{imamog1999quantum,warren2019long,burkard2023semiconductor}, as well as some 
superconducting qubits 
\cite{bialczak2010quantum,abrams2020implementation,sung2021realization,houck2007generating,johansson2006vacuum, abrams2019implementation}.
We also consider quantum circuits that contain other 2-qubit energy-conserving interactions, such as the Heisenberg exchange interaction $X\otimes X+Y\otimes Y+Z\otimes Z$.

It turns out that many useful gates and subroutines in quantum computing are energy-conserving unitaries. This includes the SWAP and controlled-SWAP gates (also know as the Fredkin gate), the controlled-Z (CZ) gate with arbitrary number of control qubits, and the family of   unitaries generated by the multi-qubit swap Hamiltonian 
\cite{lloyd2014quantum,  marvian2016universal, kimmel2017hamiltonian, pichler2016measurement}.  The standard approaches \cite{NielsenAndChuang} for implementing such unitaries ignore this conservation law and decomposes the desired unitary to elementary gates that do not conserve the intrinsic energy of qubits, such as CNOT and Hadamard.

In this work, on the other hand,  we are interested in the synthesis of energy-conserving unitaries with energy-conserving gates alone. In such circuits, the total energy of the qubits in the system remains conserved throughout the execution of the circuit.   As argued in \cite{marvian2022restrictions}, this makes the circuit more resilient against certain types of noise, such as those induced by the instability of the master clock that controls the timing of pulses   \cite{ball2016role}.  As a simple example, in \cref{Fig1} we compare two different realizations of a useful and common energy-conserving gate in quantum computing, namely the SWAP gate.

Indeed, even in the cases when the target unitary is not energy-conserving, still, it might be desirable to minimize the use of non-energy-conserving gates.  For instance, instead of the standard   implementation of the 3-qubit Toffoli  gate that requires multiple single and two-qubit non-energy-conserving gates  (namely, 6 CNOTs and 2 Hadamards \cite{NielsenAndChuang,divincenzo1998quantum,shende2008cnot}), one can implement this useful gate by sandwiching the energy-conserving controlled-controlled-Z (CCZ) gate between two Hadamards, as shown in \cref{fig:Toffoli}. Then, by realizing  CCZ gate with energy-conserving gates, one obtains an implementation of Toffoli, which is more robust against certain fluctuations of the master clock.

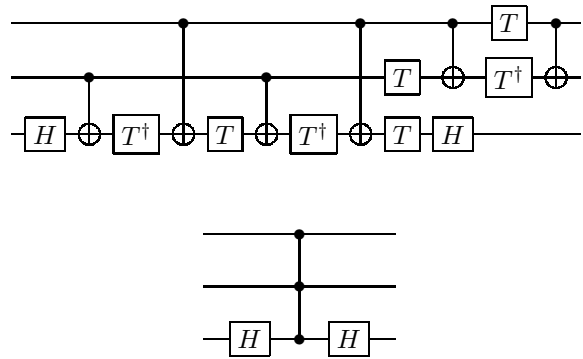
\begin{figure}
\begin{minipage}{\linewidth}
\Qcircuit @C=0.5em @R=.7em {
& \qw     & \qw     & \qw          & \ctrl{2} & \qw     & \qw     & \qw          & \ctrl{2} & \qw      & \ctrl{1} & \gate{T}     & \ctrl{1} & \qw \\
& \qw     & \ctrl{1}& \qw          & \qw      & \qw     & \ctrl{1}& \qw          & \qw      & \gate{T} & \targ    & \gate{T^\dag}& \targ    & \qw \\ 
& \gate{H}& \targ   &\gate{T^\dag} & \targ    & \gate{T}& \targ   &\gate{T^\dag} & \targ    & \gate{T} & \gate{H} & \qw          & \qw      & \qw    
}
\end{minipage}
\vspace{10mm}\\

\begin{minipage}{\linewidth}
\Qcircuit @C=1em @!R @R=.7em {
& \qw      & \ctrl{1} & \qw & \qw \\
& \qw      & \ctrl{1} & \qw & \qw  \\ 
& \gate{H} & \ctrl{-1} & \gate{H} & \qw  }   
\end{minipage}
   \caption{\textbf{Two different implementations of the Toffoli gate.} The top circuit is the standard way of implementing the Toffoli gate \cite{NielsenAndChuang,divincenzo1998quantum,shende2008cnot}, which is one of the most useful gates in quantum computing. This circuit contains eight non-energy-conserving gates, namely six CNOTs and two Hadamards. (Indeed, this can be reduced to seven single-qubit  non-energy-conserving gates, if one replace each CNOT with a CZ gate sanwhiched with two Hadamards.)  The bottom circuit requires a CCZ gate, which is energy-conserving, and two non-energy-conserving gates, namely two Hadamard. Using the techniques developed in this paper, CCZ gate can be realized with XY interaction and local Z with one ancilla qubit, or XY interaction alone with two ancilla qubits. 
   Therefore, in total, it requires only two single-qubit non-energy-conserving gates.
      }\label{fig:Toffoli}
\end{figure}

We note that using energy-conserving quantum circuits for suppressing noise has been previously considered in the context of  decoherence-free subspaces (DFS) \cite{lidar1998decoherence, Bacon:2000qf, divincenzo2000universal, Zanardi:97c, kempe2002exact, brod2013computational}. For instance, using the dual-rail encoding, a logical qubit can be encoded in the 2D subspace of a pair of qubits spanned by $|01\rangle$ and $|10\rangle$ states. Then, an encoded version of any desired quantum circuit can be performed in a DFS, i.e., an eigen-subspace  of the total intrinsic Hamiltonian $\sum_{j=1}^n Z_j$.  Indeed, it has been shown that using XY interaction, it is possible to achieve such \emph{encoded universality}  \cite{kempe2001encoded, lidar2001reducing, kempe2002exact}. However, a clear downside of this approach is that the encoding reduces the number of logical qubits (e.g., by a factor of 2 in the case of dual-rail encoding), which is undesirable, especially in the NISQ era \cite{preskill2018quantum}.  

On the other hand, in this paper, we do not assume that the global state is restricted to a DFS. Instead, the state of the qubits can be arbitrary. 
 Furthermore, we are interested in a stronger notion of universality, which requires implementing all energy-conserving unitaries.  Of course, this notion of universality also implies the encoded universality of XY  interaction. In particular, one can use the synthesis techniques developed in this work to find an efficient implementation of the exact gate sequences that are needed to perform an encoded version of a desired circuit in DFS.

Besides quantum computing, the notion of energy-conserving unitaries appears in many broad areas of quantum information science and, therefore,  it is crucial  to understand how they can be realized with energy-conserving circuits. For instance, in quantum thermodynamics one often assumes  energy-conserving unitaries are ``free", i.e.,  can be realized with negligible cost  \cite{FundLimitsNature, brandao2013resource, janzing2000thermodynamic,  lostaglio2015quantumPRX, chitambar2019quantum}. However, prior to this work, it was not known how a general energy-conserving unitary can be realized, and in particular, how it can be decomposed to a finite sequence of  local energy-conserving unitaries.  Other areas of applications include quantum clocks, quantum reference frames and the resource theory of asymmetry  (See \cref{Sec:discus} and \cite{marvian2022restrictions} for further discussion).

\subsection*{Summary of the main results: Circuit synthesis techniques}

It was shown in \cite{marvian2022restrictions} that 
it is impossible to implement a generic energy-conserving unitary by applying local energy-conserving gates on the subsystems that form the system (In \cref{Sec:summary} and \cref{Thm0} we briefly review this result).  This is in sharp contrast with the standard universality of 2-qubit gates in the absence of energy conservation \cite{divincenzo1995two, lloyd1995almost}.    In addition to this result,  Ref. \cite{marvian2022restrictions} also shows that this no-go theorem 
can be circumvented using a single ancilla qubit. In particular, Ref. \cite{marvian2022restrictions} proves that for any desired energy-conserving unitary $V$ on $n$ qubits, there exists an energy-conserving unitary $\widetilde{V}$ on $n+1$ qubits, such that $\widetilde{V}$ can be realized by Hamiltonians $X\otimes X+Y\otimes Y$ and single-qubit $Z$ and 
\be
\widetilde{V} (|\psi\rangle\otimes |0\rangle_{\text{anc}})=(V|\psi\rangle)\otimes |0\rangle_{\text{anc}}\ ,
\ee 
for all states $|\psi\rangle$ of $n$ qubit system (See the example in \cref{Fig1}). Subsequently,  Ref.   \cite{marvian2023non}  simplified and generalized this result, to symmetric 
quantum circuits with arbitrary Abelian symmetries.

Building on the ideas developed in  \cite{marvian2022restrictions} and \cite{marvian2023non}, in this work we go beyond these results   and develop various circuit synthesis  techniques for  constructing explicit  circuits with energy-conserving gates. 
 (This is analogous  to the development of the theory of quantum circuits: first, the universality of 2-local gates was established mostly using  Lie-algebraic  arguments \cite{divincenzo1995two, lloyd1995almost}, and then, building on those results, explicit and efficient circuit synthesis techniques were developed. See, e.g., \cite{barenco1995elementary, vatan2004optimal}.)  

Using these circuit synthesis methods we show that

\begin{theorem}\label{Thm1}\textbf{(Exact implementation)}
 Any energy-conserving unitary transformation on $n$ qubits,  i.e., a unitary transformation  
 that commutes with $\sum_{j=1}^n Z_j$,   
 can be realized exactly, up to a possible global phase,  using $\mathcal{O}(4^n n^{3/2})$  gates from any one of the following universal gate sets
\begin{enumerate}
\item  \label{gates_phaseZ}
$\sqrt{\i\text{SWAP}}=\exp( \i\frac{ \pi}{4} R)$\ , and  $\exp(\i\phi Z): \phi\in(-\pi,\pi]$ gates, with one ancilla qubit, where 
\be
R=\frac{1}{2}(X\otimes X+Y\otimes Y)\ .
\ee
    \item \label{gates_phaseH}  $\exp(\i\alpha H_\text{int}): \alpha \in\mathbb{R}$\ , and $S=e^{\i\pi/4}\exp({-\i \pi Z/4})$ gates,  {with} one ancilla qubit, where $H_\text{int}$ is any energy-conserving 2-qubit Hamiltonian that is not diagonal in the computational basis, such that 
    \be
    [H_\text{int}, Z\otimes \mathbb{I}]=-[H_\text{int}, \mathbb{I}\otimes Z]\neq 0\ ,
    \ee
    \item \label{gates_XY_only} $\exp(\i\alpha R): \alpha \in (-\pi,\pi]$ gates, {with} two ancilla qubits.
\end{enumerate}

\end{theorem}

Two canonical examples of non-diagonal energy-conserving  Hamiltonians are the XY interaction and the Heisenberg interaction $R_{\text{Heis}}:=(X\otimes X + Y\otimes Y + Z \otimes Z)/2$. Then, part \ref{gates_phaseH} of the theorem implies that the following sets are universal:
\begin{itemize}
    \item $\exp(\i\alpha R): \alpha \in(-\pi,\pi]$\ , and $S$ gates, with one ancilla qubit,
    \item $\exp(\i\alpha R_{\text{Heis}}): \alpha \in(-\pi,\pi]$\ , and $S$ gates, with one ancilla qubit.
\end{itemize}

A few remarks are in order: First, as we show in \cref{ss:energy_conserving},   general  energy-conserving unitary transformations on $n$ qubits are smoothly parameterized by  $\approx 4^n /\sqrt{\pi n}$  real parameters.  Therefore, a simple parameter counting implies that for generic energy-conserving unitaries, the above constructions are close to optimal, by a factor of $n^2$.\footnote{ 
 More generally, we will show in \cref{sec:nqubit} \cref{step:semi-universal} that if an energy-conserving unitary acts non-trivially only on a subspace spanned by $D$ elements of the computational basis then
it can be realized using $\mathcal{O}(n^2\times D^2)$ gates from any one of the above gate sets.}

Second, note that in part \ref{gates_phaseH} of the theorem the condition that $H_\text{int}$ is not diagonal is necessary, because otherwise qubits cannot exchange energy with each other, i.e., the overall $n$-qubit unitary will be also diagonal in the computational basis. 

Finally, note that in all the 3 cases in \cref{Thm1} one needs, at least, one  ancilla qubit, which by the argument \cite{marvian2022restrictions} is unavoidable.  However, for the last gate set an extra ancilla is used. This is because $XY$ Hamiltonian has an additional $\mathbb{Z}_2$ symmetry, namely 
\be\label{Z22}
X^{\otimes 2} R X^{\otimes 2}=R \ .
\ee 
   As we further explain in \cref{sec:XY_only}, this  $\mathbb{Z}_2$ symmetry can be broken with the extra ancilla qubit. In this case, the ancilla can be interpreted as a quantum reference frame, or asymmetry catalyst  \cite{gour2008resource,marvian2013theory}.\footnote{It is worth noting that with any finite number of ancilla qubits, it is impossible to implement non-energy-conserving unitaries, using energy-conserving interactions alone. This can be understood as a consequence of the Wigner-Araki-Yanase theorem \cite{Wigner52,araki1960measurementof}, or equivalently, a consequence of the no-programming theorem \cite{nielsen1997programmable}  (See \cite{marvian2012informationtheoretic} for the connection).}

In addition to exact implementation, we also investigate the approximate implementation of energy-conserving unitaries
using finite gate sets, e.g., 
$\sqrt{\i\text{SWAP}}$ and $S$ gates. Applying the Solovay-Kitaev theorem \cite{kitaev1997quantum,kitaev2002classical,kuperberg2023breaking}, one can  bound the number of such gates that are needed to implement a desired energy-conserving unitary, with any error $\epsilon>0$ as quantified in terms of the operator norm distance between the desired unitary and the realized unitary (See \cref{sec:approx} for the formal statements).  In particular, we show that
\begin{corollary}\label{cor}\textbf{(Approximate implementation)} 
    Any energy-conserving unitary on {$n\geq2$} qubits can be realized with an error bounded by $\epsilon>0$, using  $\mathcal{O}(4^n n^{3/2}(n+\log \epsilon^{-1})^{\nu})$  number of gates from either of the following gate sets:
\begin{enumerate}
        \item $\sqrt{\i\text{SWAP}}$ gate alone {with} 2 ancilla qubits.
        \item $\sqrt{\i\text{SWAP}}$ and $S$ gates {with} 1 ancilla qubit. \label{gates_finite_iswap_T}
%        \item $\exp( \i\frac{ \pi}{8}R_{\text{Heis}})$ and $S$ gates and 1  ancilla qubit.\label{gates_finite_XYZ_T}
\end{enumerate}
    Here, $\nu$ is the exponent for the complexity of Solovay-Kitaev algorithm, and can be chosen as any number greater than $\SKexp\approx 1.44042$.
\end{corollary}

In the above discussions, we  did not specify any geometry for the system of qubits and 
assumed gates can be applied between any pair of qubits. What happens if we assume the qubits form an open chain and gates are allowed only between nearest-neighbor qubits? 

As it was noted in \cite{marvian2022restrictions}, this additional restriction does not affect the set of realizable unitaries, provided that the ancilla can be coupled to all the qubits in the chain.  This is because interacting two qubits with the ancilla  using $X\otimes X+Y\otimes Y$ and $Z$ Hamiltonians, allows us to perform the SWAP gate that exchanges the state of two qubits (See \cref{Fig1}). Then, applying SWAPs between nearest-neighbor qubits, one can arbitrarily change the order of qubits.  On the other hand, if one requires that the ancilla should also be part of the chain and only interact with its nearest neighbors, then the above universality result does not hold anymore. Indeed, in this case, using the Jordan-Wigner transformation  \cite{jordan1928pauli, fradkin1989jordan, nielsen2005fermionic}  this system can be mapped to a 
free fermionic system, which implies that the group of realizable unitaries is significantly smaller than the group of energy-conserving unitaries (namely, it has dimension $n^2$).  Of course, if one is allowed to use the SWAP gate, this restriction can be avoided.
That is
\begin{corollary}\label{cor3}\textbf{(Approximate implementation with nearest-neighbor gates)} 
Consider an open chain of $n+1$ qubits, where one of the qubits is designated as the ancilla qubit and is initially prepared in state $|0\rangle$.
Any energy-conserving unitary on the rest of qubits can be realized using 
  $\mathcal{O}(4^n n^{3/2}(n+\log \epsilon^{-1})^{\nu})$ number of gates $S$, SWAP, and  $\sqrt{\i\text{SWAP}}$ with two-qubit gates restricted to nearest-neighbor qubits, for any $\nu>\SKexp$.
\end{corollary}

%Suppose energy-conserving unitary $V$ on $n\ge 3$ qubits has the additional symmetry $X^{\otimes n} V X^{\otimes n}= V$. Then,  unitary $V$  can be approximately realized using $\sqrt{ \i \text{Sw}}$ alone.  \\ 

Finally, in \cref{sec:XY_only} we  
develop circuit synthesis techniques using XY interaction alone and, in particular,  without single-qubit $Z$ Hamiltonian and ancilla qubits. 
It follows that all realizable unitaries should satisfy 
the $\mathbb{Z}_2$ symmetry of XY interaction in \cref{Z22}.  However, as the following theorem states, there are more constraints on the realizable unitaries (See \cref{sec:XY_only} for further details). 
\begin{theorem}\label{Thm2} \textbf{(Circuits with XY interaction alone)}
On a system with $n\ge 3$ qubits, any  unitary $V$ can be realized with interaction $R=(X\otimes X+Y\otimes Y)/2$ alone (without any ancilla qubits)  if, and only if, 
\begin{enumerate}
\item $V$ is energy-conserving, i.e., $[V , \sum_{j=1}^n Z_j]=0$ .
\item It satisfies the $\mathbb{Z}_2$ symmetry  $X^{\otimes n} V X^{\otimes n}=V$.
\item  $\text{det}(V\sct{m})=1: m=0,\cdots, n$, where 
$\text{det}(V\sct{m})$ is the determinant of $V\sct{m}$, the component of $V=\bigoplus_{m=0}^n V\sct{m}$ in the subspace with Hamming weight $m$, that is the eigensubspace of $\sum_{j=1}^n Z_j$ with eigenvalue $n-2m$.
\item \label{cond:detpm} When $n$ is even, $\text{det}(V\sct{n/2,\pm})=1$, where 
$V\sct{n/2,\pm}$ is the component of $V$ in the joint eigensubspaces of 
$\sum_j Z_j$ and $X^{\otimes n}$, with eigenvalues 0 and  $\pm 1$, respectively,  such that 
\be
V\sct{n/2}=V\sct{n/2,+}\oplus V\sct{n/2,-}\ .
\ee
\end{enumerate}
Furthermore, any unitary satisfying the above conditions can be realized with $\mathcal{O}(4^n n^{3/2})$  gates $\exp(\i\theta R): \theta\in(-\pi,\pi]$.
\end{theorem}
This means that the group of unitaries that can be realized with XY interaction alone on $n \ge 3$ qubits, denoted by $\mathcal{G}_n$, is isomorphic to
\begin{align}
  &\mathcal{G}_n \cong \prod_{m=1}^{\lfloor \frac{n}{2} \rfloor} \text{SU}({{n}\choose{m}}) \  \ \ \ \  
&&:n \text{ is odd }   \nonumber
  \\
&\mathcal{G}_n \cong  \prod_{m=1}^{\frac{n}{2}-1 }  \text{SU}({{n}\choose{m}})\times \Big[\text{SU}(\frac{1}{2}{{n}\choose{\frac{n}{2}}})\Big]^{\times 2}   \  \ \ \ \   
&&:n \text{ is even }\nonumber
\end{align}
where we have used the fact that the dimension of the subspace with Hamming weight $m$ is  ${{n}\choose{m}}$. Note that in the special case of $n=2$, the constraint in condition \ref{cond:detpm} does not hold and\footnote{In this case, the eigen-subspaces of $X^{\otimes 2}$ in the subspace with Hamming weight 1 correspond to vectors $|01\rangle\pm |10\rangle$, which are also eigenvectors of $\exp(\i\theta R)$ with eigenvalues $\e^{\pm \i \theta}$.}  
$$\mathcal{G}_2 \cong \text{U}(1)\ .$$ It is also worth emphasizing that while  the additional constraint in the case of even $n\ge 4$, i.e., condition \ref{cond:detpm},  is related to the aforementioned $\mathbb{Z}_2$ symmetry of XY interaction in \cref{Z22}, it is not necessarily satisfied by all unitaries that respect the $\mathbb{Z}_2$ symmetry (Namely, it is of the type of constraints discussed in \cite{marvian2022restrictions} which are the consequence of both locality and symmetry of Hamiltonian). For instance,  for a system with $n=4$ qubits, consider the family of unitaries 
\begin{align}\nonumber
\exp\big(\i\theta \big[|0011\rangle\langle 1100|+|1100\rangle\langle 0011|\big]\big)\ \ \ \ \  : \theta\in(-\pi,\pi]\ ,
\end{align}
i.e., unitaries generated by Hamiltonian $H=|0011\rangle\langle 1100|+|1100\rangle\langle 0011|$. 
While these 
unitaries respect conditions 1-3 of \cref{Thm2},   unless $\theta=0$, they do not satisfy condition \ref{cond:detpm} of this theorem and therefore they are not realizable with XY interaction alone\footnote{ In particular, in this case, 
 $V\sct{2,\pm}=\exp(\pm \i\theta |\psi_\pm\rangle\langle\psi_\pm|)$
where $|\psi_\pm\rangle=(|0011\rangle\pm |1100\rangle)/\sqrt{2}$ are the  eigenvectors of $H$ with eigenvalues $\pm 1$. Then, $\text{det}(V\sct{2,\pm})=e^{\pm\i\theta}$ which implies,  unless $\theta=0$, this unitary is not realizable with XY interactions alone.}. Finally, we note that our results on approximate universality imply that the group generated by $\sqrt{\i\text{SWAP}}$ gates on  $n\ge 3$ qubits is a dense subgroup of $\mathcal{G}_n$.\\

The rest of this paper is organized as follows: \\
\begin{itemize}
\item \cref{Sec:prem} formulates our goal by introducing energy-conserving unitaries and the notion of (semi-)universality. It contains a collection of elementary gates useful for subsequent circuit constructions. In \cref{ss:iswap_z}, circuits with $\iSWAP$ and single-qubit z-rotations are characterized. 

\item \cref{sec:2-qubit} discusses the structure of 2-qubit energy-conserving unitaries. It shows the semi-universality of (and thus the equivalence between) the gate sets \ref{gates_phaseZ} and \ref{gates_phaseH} in \cref{Thm1} for 2 qubits.

\item \cref{sec:3qubit} is focused on 3-qubit energy-conserving unitaries. It contains the implementation of controlled-Z and SWAP gates using a single ancilla qubit, as well as circuit identities that are useful afterwards. The construction of 3-qubit 2-level special unitary energy-conserving gates is presented in \cref{ss:2levelSV3}, which serves as a basis for the construction of general $n$-qubit energy-conserving unitaries.

\item \cref{sec:nqubit} is dedicated to the construction of $n$-qubit energy-conserving unitaries, and concludes the proof of \cref{Thm1}, for gate sets \ref{gates_phaseZ} and \ref{gates_phaseH}.

\item \cref{sec:XY_only} is focused on the set of  unitaries that are realizable with XY interaction alone, contains  the proof of \cref{Thm2} and completes the proof of \cref{Thm1} for gate set \ref{gates_XY_only}. %Some of the techniques relies on results in \cref{sec:nqubit}.

\item \cref{sec:approx} introduces approximate universality and provides approximate constructions of aforementioned circuits in \cref{sec:3qubit,sec:nqubit,sec:XY_only}. It combines a Lie algebraic characterization of unitaries generated by $\sqrt{\i\text{SWAP}}$ gates
with the Solovay-Kitaev theorem, and proves \cref{cor}.

\item \cref{Sec:discus} contains a short discussion on applications of energy-conserving quantum circuits, in areas such as quantum computing, quantum thermodynamics, and quantum clocks.  

\end{itemize}

%\newpage

\section{Preliminaries }\label{Sec:prem}

\subsection{Energy-conserving unitaries}\label{ss:energy_conserving}
Consider a system with $n$ qubits, each with the intrinsic Hamiltonian $-\Delta E\ Z/2 $, where $\Delta E>0$ is the energy difference between the ground state $|0\rangle$ and the excited state $|1\rangle$ of the qubit. Then, the total intrinsic Hamiltonian of this system is 
\be\label{int}
H_\text{intrinsic}=-\frac{\Delta E}{2} \sum_{j=1}^n Z_j=\Delta E(-\frac{n}{2}+\sum_{m=0}^n m\ \Pi\sct{m} )\ , 
\ee
where $Z_j$ denotes the Pauli $Z$ operator on qubit $j$ tensor product with the identity operator on the rest of qubits, and $\Pi\sct{m}$ is the projector to the eigen-subspace  $\mathcal{H}\sct{m}$ of $H_\text{intrinsic}$ with energy $(2m-n)\times {\Delta E/2}$.  
Then, the total Hilbert space   decomposes as 
\be\label{dec}
(\mathbb{C}^2)^{\otimes n}\cong \bigoplus_{m=0}^n \mathcal{H}\sct{m}\ .
\ee
Let  $\{|0\rangle, |1\rangle\}^{\otimes n}$ be the \emph{computational} basis for $n$ qubits. We will label a vector in this basis with a bit string $\textbf{b}=b_1\cdots b_n\in\{0,1\}^n$ as $\ket{\textbf{b}}$.  Then,  $\mathcal{H}\sct{m}$ is the subspace spanned by the elements of this basis with Hamming weight $m$ (Recall that the Hamming weight of a bit string is the number of bits with value 1).

We are interested in energy-conserving unitaries on this systems, i.e., those that conserve the total intrinsic Hamiltonian of the  qubits in the system. A unitary  $V$ on $n$ qubits is energy-conserving if, and only if, it  is block-diagonal with respect to the decomposition in \cref{dec}, such that  
\be
V=\bigoplus_{m=0}^n V\sct{m}\ .
\ee
 Following the notation in \cite{marvian2022restrictions},  we denote the set of such unitaries as 
\begin{align}
    \mathcal{V}_n^{\text{U}(1)} &= \big\{V: [V, H_\text{intrinsic}]=0, VV^\dag  = \mathbb{I}^{\otimes n} \big\} \nonumber\\  & =\bigoplus_{m=0}^n  \text{U}(\mathcal{H}\sct{m})\ , 
    \label{dec3}
\end{align}
where $\mathbb{I}$ denotes the identity operator on a single qubit, and $\text{U}(\mathcal{H}\sct{m})$ denotes the group of all unitaries acting on $\mathcal{H}\sct{m}$. 
Here, the superscript  $\text{U}(1)$ refers to the fact that $\mathcal{V}_n^{\text{U}(1)}$ is the  set of unitaries that commute with unitaries 
\be
(\exp(\i\theta Z)\big)^{\otimes n}=\exp(\i\theta \sum_{j=1}^n Z_j)\ \ \ :\ \theta\in (-\pi, \pi]\ ,
\ee
%\exp(- \i t H_\text{intrinsic}): t\in[0,2\pi\Delta E^{-1})$, 
which is a  representation of the group $\text{U}(1)=\{e^{i\theta}: \theta\in(-\pi,\pi]\}$ that describes  the time evolution of a periodic system (Ref. \cite{marvian2022restrictions} studies circuits with  general symmetries).

Decomposition in \cref{dec3} implies that energy-conserving unitaries can be smoothly parameterized using 
\be
\text{dim}(\mathcal{V}_n^{\text{U}(1)})=\sum_{m=0}^n{{n}\choose{m}}^2={2n \choose n} \approx  \frac{4^n }{\sqrt{\pi n}}\ , 
\ee
real parameters, where by $\approx $ we mean the ratio of two sides goes to 1, in the limit $n\rightarrow \infty$,  which can be established using the Stirling's approximation for factorials. 

As we discuss in the following, it is also useful to consider the subgroup of this group, formed from energy-conserving unitaries $V=\bigoplus_m V\sct{m}$, where
in each sector we have  $\text{det}(V\sct{m})=1$, i.e.,   
\begin{align}
    \mathcal{SV}_n^{\rm U(1)}&=\{V=\bigoplus_{m=0}^n V\sct{m}: V\sct{m}\in \text{SU}(\mathcal{H}\sct{m})   \} \nonumber \\ 
    &=\bigoplus_{m=0}^n \text{SU}(\mathcal{H}\sct{m})  \ , \label{sv}
\end{align}
where $\text{SU}(\mathcal{H}\sct{m})$ denotes the group of special unitaries acting on $\mathcal{H}\sct{m}$.

With this definition, any element of $\mathcal{V}_n^{\rm U(1)}$ has a decomposition as 
\be\label{dec5}
V=Q D=D Q\ ,
\ee
where unitary $Q\in \mathcal{SV}_n^{\rm U(1)}$, unitary
\be
D=\sum_{m=0}^n \exp\big[\i {\theta_m} {{{n}\choose{m}}}^{-1}\big]\ \Pi\sct{m}\ ,
\ee
is diagonal in the computational basis, and
\be
\theta_m=\text{arg}(\text{det}(V\sct{m}))\  ,
\ee 
is the phase of the determinant of $V\sct{m}$, and for convenience  we assume $\theta_m\in(-\pi, \pi]$. 
We note that, in general, the decomposition in \cref{dec5} is not unique.

In the following, for any unitary $V$, $\Lambda^{\textbf{c}}(V)$ denotes its controlled-version with control string $\textbf{c}\in\{0,1\}^k$, i.e.,
\be \label{eq:Lambda_preliminaries}
\Lambda^{\textbf{c}}(V) := \sum_{\textbf{c}' \neq \textbf{c}}\ket{\textbf{c}'}\bra{\textbf{c}'} \otimes \mathbb{I} + \ket{\textbf{c}}\bra{\textbf{c}}\otimes V\ .
\ee
If $V\in \mathcal{SV}^{\text{U}(1)}_n$, then its controlled version $\Lambda^{\textbf{c}}(V)$ is in $\mathcal{SV}^{\text{U}(1)}_{n+k}$.

\subsection{Semi-universality and universality with  XY interaction and local Z}\label{Sec:summary}

Ref. \cite{marvian2022restrictions}  shows that any energy-conserving unitary $V=\bigoplus_{m=0}^n V\sct{m}$ can be realized using XY interaction and local $Z$, up to certain constraints on the relative phases between sectors with different energies. Following \cite{marvian2023non}, we say a set of gates are semi-universal for energy-conserving unitaries, if they generate $\mathcal{SV}_n^{\rm U(1)}$ for all integer $n$.   
 Using the notion of semi-universality, the result of \cite{marvian2022restrictions}   can be rephrased as
\begin{theorem}[Based on \cite{marvian2022restrictions}]\label{Thm0}
For a system with $n$ qubits, the group $G_{XX+YY,Z}$  generated by 
2-qubit gates  $\exp(\i\alpha [XX+YY])$, single-qubit gates $\exp(\i \beta Z)$, and global phase $e^{\i \phi}\mathbb{I}$ for $\alpha,\beta,\phi\in(-\pi,\pi]$ is equal to the subgroup of all energy-conserving unitaries $V\in \mathcal{V}_n^{\text{U}(1)}$ satisfying the extra constraint
\be\label{Thmeq}
\theta_m={{n}\choose{m}} \times \Big[\frac{m}{n}\times (\theta_n-\theta_0)+ \theta_0\Big]\ \  :\  \text{mod}\   2\pi \ ,
\ee  
for all $m=0, \cdots, n$, where $\theta_m=\text{arg}(\text{det}(V\sct{m}))$
and $V\sct{m}$ is the component of $V$ in the sector with Hamming weight $m$. In particular, 
 $G_{XX+YY,Z}$  
 contains $\mathcal{SV}^{\text{U}(1)}_n$ defined in \cref{sv} and, therefore, the aformentioned gates are semi-universal. 
 \end{theorem}
We note that the circuit synthesis techniques presented in this paper provide an independent proof of the second part of this theorem. In particular, the semi-universality of this gate set is demonstrated in \cref{prop13}. For completeness, in \cref{App:proof}, 
 we also establish the first part; that is, we show that \cref{Thmeq} along with the energy conservation condition fully characterizes the group $G_{XX+YY,Z}$.\footnote{It is also worth  noting that in the context of control theory, using a Lie algebraic argument, Ref. \cite{wang2016subspace} shows  that in any individual Hamming-weight sector, Hamiltonians $XX+YY$, $ZZ$, and local $Z$ together generate all unitaries in that sector; a property  known as subspace controllability.}

\cref{Thmeq} imposes $n-1$ independent constraints on the set of unitaries in $G_{XX+YY,Z}$. Then, it follows from this theorem that the difference between the dimensions of the Lie group of all energy-conserving unitaries and the subgroup  
$G_{XX+YY,Z}$ is equal to
\be
\text{dim}(\mathcal{V}^{\rm U(1)}_n)-\text{dim}(G_{XX+YY,Z})=n-1\ .
\ee
Ref. \cite{marvian2022restrictions} also shows that using a single ancilla qubit, it is possible to circumvent these constraints. That is
\begin{corollary}\cite{marvian2022restrictions}
Any energy-conserving unitary can be realized with a single ancillary qubit, and gates  $\exp(i\phi Z)$ and 2-local gates  $\exp(i\phi [XX+YY])$.
\end{corollary}

The proof of this result in \cite{marvian2022restrictions} is Lie algebraic. In this work, on the other hand, we give explicit circuit construction methods for realizing general energy-conserving unitaries with a single ancilla qubit.

\subsection{2-level energy-conserving unitaries}

A key notion in the quantum circuit theory is the concept of 2-level unitaries, also known as Givens rotations \cite{NielsenAndChuang}. We say a unitary transformation is 2-level with respect to the computational basis $\{|0\rangle,|1\rangle\}^{\otimes n}$, if it acts trivially (i.e., as the identity operator)  on all the basis elements, except, at most 2. 

In the following, for any pair of bit strings $\textbf{b}, \textbf{b}'\in\{0,1\}^n$, $U(\textbf{b},\textbf{b}')$ denotes a 2-level unitary acting as $U\in\text{U(2)}$ on the subspace spanned by $\ket{\textbf{b}}$ and $\ket{\textbf{b}'}$, and acting trivially on the orthogonal complement of this subspace. More precisely, for a $2\times2$ unitary \be
U=\left(\begin{matrix}
    u_{11} & u_{12} \\ u_{21} & u_{22}
\end{matrix}\right),\ee
its corresponding 2-level unitary is
\begin{align}
    U(\textbf{b},\textbf{b}') &:=  u_{11}\ketbrasame{\textbf{b}} + u_{12}\ketbra{\textbf{b}}{\textbf{b}'}\nonumber\\
   & + u_{21}\ketbra{\textbf{b}'}{\textbf{b}} + u_{22}\ketbrasame{\textbf{b}'} \nonumber\\
   & + \sum_{\textbf{c}\neq \textbf{b},\textbf{b}'} \ketbrasame{\textbf{c}} \,.  \label{eq:def_2-level}
\end{align}

A 2-level unitary $U(\textbf{b},\textbf{b}')$ is in $\mathcal{SV}_n^{\rm U(1)}$ if, and only if  $\text{det}(U)=1$ and $w(\textbf{b})=w(\textbf{b}')$, where 
 for any bit string 
 $\textbf{b}=b_1\cdots b_n\in\{0,1\}^n$, 
\be 
w(\textbf{b})=\sum_{j=1}^n b_j\ ,
\ee
denotes the Hamming weight of $\textbf{b}$.

\subsection{Elementary gates}

We mainly study quantum circuits formed from two types of gates. First, single-qubit rotations around z, i.e., unitaries 
\be
R_z(\phi)=\exp(\i\frac{\phi}{2} Z)=\left(
\begin{array}{cc}
e^{\i \frac{\phi}{2}}  &     \\ & e^{-\i \frac{\phi}{2}}
  \end{array}
\right) \quad :\phi\in(-\pi,\pi]\ .
\ee
Two important specific cases are
\be
T= e^{\frac{\i\pi}{8}}
R_z(-\frac{\pi}{4})=\left(
\begin{array}{cc}
1  &     \\ & e^{\i\frac{\pi}{4}}
  \end{array}
\right)=\sqrt[4]{Z}
\ee
and 
\be
S=e^{\frac{\i\pi}{4} }R_z(-\frac{\pi}{2})
=
\left(
\begin{array}{cc}
1  &     \\ & \i
  \end{array}
\right)=\sqrt{Z}= T^2\ .
\ee
The second type of gates used in our circuits  are in the form
\begin{align}
    \exp( \i\alpha R) \quad  : \ \alpha\in (-\pi,\pi]\ ,
\end{align}    
 where 
 \begin{align}
 R := \frac12 (X\otimes X+Y\otimes Y)\ .
\end{align}
Two important special cases are the $\i\text{SWAP}$ gate  
\be\label{iSWAP}
\i\text{SWAP}=\iswap=\exp( \i\frac{ \pi}{2} R)=\left(
\begin{array}{cccc}
1  &   & &   \\
  &   & \i & \\
  & \i  &    & \\ 
  & & & 1 
\end{array}\right)\ ,
\ee
and the square root of $\i\text{SWAP}$ gate
\be\label{sqrt}
\sqrt{\i\text{SWAP}}=\siswap=\exp( \i\frac{ \pi}{4} R)=\left(
\begin{array}{cccc}
1  &   & &   \\
  & \frac{1}{\sqrt{2}}  &\frac{\i}{\sqrt{2}}  & \\
  & \frac{\i}{\sqrt{2}}  &   \frac{1}{\sqrt{2}}    & \\ 
  & & & 1 
\end{array}\right)\ ,
\ee
where the matrices are written in the computational basis $\{|00\rangle, |01\rangle, |10\rangle, |11\rangle\}$. 
 Note that $\sqrt{\i\text{Sw}}^\dag=(\sqrt{\i\text{Sw}})^7$. 
See, e.g., \cite{Burkard:99a, schuch2003natural} for further discussions about the properties and applications  of $\sqrt{\i\text{Sw}}$ and $\i\text{Sw}$ gates for circuit synthesis.

We also consider the SWAP gate
\begin{align}
\text{SWAP}=\text{Sw}&={\small\left(
\begin{array}{cccc}
1  &   & &   \\
  &   & 1 & \\
  &  1 &    & \\ 
  & & & 1 
\end{array}\right)} = ~ \begin{minipage}{3cm}
    \Qcircuit @C=1em @R=2.7em {
    &  \qswap &\qw    \\ 
     &  \qswap \qwx &\qw
   }
\end{minipage}\ ,
\end{align}
and the Controlled-$Z$ (CZ) gate
\begin{align}
\CZ&=\small{|0\rangle\langle0|\otimes \mathbb{I}+|1\rangle\langle 1|\otimes Z}={\small\left(
\begin{array}{cccc}
1  &   & &   \\
  &  1 &  & \\
  &   & 1   & \\ 
  & & & -1 
\end{array}\right)}= \begin{minipage}{3cm}
    \Qcircuit @C=1em @R=2.7em {
    &  \ctrl{1} &\qw    \\ 
     &  \ctrl{-1}& \qw  
       }
\end{minipage}\nonumber
\end{align}
\subsection{Circuits with $\i\text{SWAP}$  and single-qubit z rotations} \label{ss:iswap_z}

In this paper, we show how a general energy-conserving unitary can be realized with the $\exp(\i\alpha R)$ gate and single-qubit rotations around $z$. However, it is useful to first consider a more restricted family of circuits generated by the single-qubit rotations around $z$ together with  $\i\text{SWAP}$ gate.  To analyze such circuits, we consider the useful circuit  identity
\begin{align} \label{eq:decompiSWAP}
\begin{minipage}{3cm}
    \Qcircuit @C=1em @R=1.5em {
    & \multigate{1}{\i\text{Sw}} & \qw  \\
    & \ghost{\i\text{Sw}} & \qw  \\
    }
\end{minipage}
\ \ \ \  = \ \  \ \ 
\begin{minipage}{3cm}
    \Qcircuit @C=1em @R=1.5em {
    & \qswap & \qw&  \ctrl{1} & \qw  & \gate{S}& \qw  \\& \qswap \qwx   &\qw 
    & \ctrl{-1} & \qw &\gate{S}& \qw \\
    }
\end{minipage}
\end{align}
Note that the S gate commutes with the controlled-Z gate. Then, using this identity it can be easily seen that
\begin{proposition} \label{prop:iswap}
Suppose unitary $V$ is realized by a circuit formed from $\i\text{SWAP}$ gates. Then, $V$ has a decomposition as $V=V_3 V_2 V_1$, where $V_1$ is a permutation, i.e., is a composition of SWAP gates, $V_2$ is a composition of controlled-$Z$ gates, and $V_3$ is a composition of single-qubit $S$ gates. More generally, if 
in addition to $\i\text{SWAP}$,  the circuit also contains the single-qubit rotations around z, denoted by $R_z(\phi): \phi\in(-\pi,\pi]$, then the realized unitary $V$ has a similar decomposition, where  $V_3$ is now a product of arbitrary single-qubit rotations around the $z$ axis.   
\end{proposition}

In particular, note that while  $\i\text{SWAP}$ is an entangling gate, the family of energy-conserving unitaries that can be realized by combining this gate with single-qubit z-rotations, are very limited. Namely, the set of realizable unitaries is specified by $n$ real parameters and they map any element of the computational basis to an element of the computational basis, up to a global phase.  

Using this property of $\i\text{SWAP}$ circuits one can establish several other useful circuit identities, which are summarized in \cref{tab:identities}. Such identities will play a key role for constructions in the following sections.

\section{2-qubit energy-conserving unitaries}\label{sec:2-qubit}

\subsection{The structure and realization of $\mathcal{V}_2^{\rm U(1)}$} \label{ss:V2stucture}

For $n=2$ qubits, the energy levels of the intrinsic Hamiltonian 
$H_\text{intrinsic}$ in \cref{int} 
decomposes the Hilbert space into 3 eigen-subspaces corresponding to Hamming weights 0, 1, and 2.  Then, a general 2-qubit energy-conserving unitary is specified by $1+2^2+1=6$ real parameters and is in the form
\be\label{2-qubit}
V=
\left(
\begin{array}{ccc}
e^{\i\theta_0}  &   &   \\
  & V\sct1  &   \\
  &   &   e^{\i\theta_2}
\end{array}\right)
=Q D=D Q\ .
\ee
Here,  $V\sct1$ is an arbitrary  $2\times 2$ unitary in the sector with Hamming weight 1 spanned by $|01\rangle$ and $|10\rangle$. Furthermore, $D$  is diagonal in the computational basis
 and satisfies $[D, V]=0$, whereas 
 \be\label{Q}
 Q = \left(\begin{matrix}1&&\\&Q\sct1&\\&&1\end{matrix}\right)\in\mathcal{SV}^{\rm U(1)}_2\ ,
 \ee
 i.e., has determinant 1 in all sectors with Hamming weights 0, 1, and 2. 
 $Q\sct1$ can be written in the basis of $\ket{01}$ and $\ket{10}$ as a single-qubit gate. For example, 
 \begin{align}\ Q&=\exp(\i\alpha R) &&\longrightarrow\ \ \  Q\sct1 =\exp(\i \alpha X) \nonumber\\
 Q&=\sqrt{\i\text{SWAP}}&&\longrightarrow\ \ \  Q\sct1=\exp(\i\pi X/4) \nonumber \\
    Q&=\exp(\i\alpha [Z_1-Z_2])&&\longrightarrow\ \ \  Q\sct1=\exp(2\i\alpha Z)\ . \label{eq:VV1}
\end{align}

In particular, using the Euler decomposition for $Q\sct1$ as $Q\sct1 =  e^{\i \gamma Z} e^{\i \beta Y} e^{\i \alpha Z}$ for some $\alpha,\beta,\gamma \in (-\pi,\pi]$, we find that $Q$ has a decomposition as 
 \be\label{Q-decomp}
Q=e^{\i \frac{\gamma}{2} (Z_{1}-Z_2)} \sqrt{\i\text{Sw}}_{12} e^{\i \frac{\beta}{2} (Z_{1}-Z_2)}\sqrt{\i\text{Sw}}^\dag_{12} e^{\i \frac{\alpha}{2} (Z_{1}-Z_2)}  \ .
\ee
%\color{purple} Therefore,  in the case of $n=2$ qubits  $\i\sqrt{\text{SWAP}}$ and single-qubit rotations around $z$ are semi-universal (In the following sections, we prove this fact for all $n$). \color{black} 

Furthermore, 
\bes
\begin{align}
D&=e^{i\theta_0}\Pi\sct0\pm e^{\i\theta_1/2}\Pi\sct1+e^{i\theta_2} \Pi\sct2\\ &= e^{i \phi_2 Z_1Z_2}\ e^{i \phi_1 (Z_1+Z_2)} \ e^{i\phi_0}\ ,\label{D-formula}
\end{align}
\ees
 where 
 $\theta_1=\text{arg}(\text{det}(V_1))$ and
\bes
 \begin{align}
 4 \phi_0=\theta_0+\theta_1+\theta_2&=:\Phi_0\ , \\  8 \phi_1=2\theta_0-2\theta_2&=: \Phi_1 \ , \\ 
 4\phi_2=\theta_0-\theta_1+\theta_2&=:\Phi_2{\ \  :\  \text{mod}\   2\pi}\ .
 \end{align}
 \ees

 The phases $\Phi_l: l=0,1, 2$ are called the $l$-body phase associated to the unitary $V$ \cite{marvian2022restrictions}. 
 It is worth noting that while neither of phases $\theta_0, \theta_1, \theta_2$, or $\phi_0, \phi_1, \phi_2$  are  physically observable, all the $l$-body phases,  except $\Phi_0$, are  observable.    
 For instance, when  $V$ is the identity operator, \cref{D-formula} holds for all 
   $$\phi_0=\phi_1=\phi_2=\frac{k \pi}{2}\ \ \ \ : k=0,\cdots, 3\ ,$$
whereas for all 4 cases we have
$$\Phi_0=\Phi_1=\Phi_2=0{\ \  :\  \text{mod}\   2\pi}\ . $$

Note that  $\Phi_0=4\phi_0$ corresponds to a global phase. Furthermore, $\Phi_1=4\phi_1$ can be changed arbitrarily by applying the unitary $\exp({\i \phi_1 (Z_1+Z_2)}) $, which can be realized with local $Z$ interactions. On the other hand, as it follows immediately from \cref{Thmeq} in \cref{Thm0}, a general U(1)-invariant  unitary $V$ on a pair of qubits can be realized by XY interaction and local $Z$ if, and only if its 2-body phase is zero, i.e., 
 \be 
 \Phi_2=\theta_0-\theta_1+\theta_2=0{\ \  :\  \text{mod}\   2\pi}\ .
 \ee
 This constraint, for instance,  excludes the SWAP and CZ gates,  because for both of these operators $2$-body phase is $\Phi_2=\pi$.

  %Interestingly, this phase can be created if one is allowed to use an ancillary qubit (See the next section). 
  
  It is also worth mentioning that the 
  matrices $Q$ and $D$ in 
  decomposition 
in \cref{2-qubit} are unique up to a freedom in choosing $\theta_1$ or $\theta_1+\pi$. This freedom is related to the fact that both matrices $$
\left(
\begin{array}{cc}
  1 &    \\
  & 1  \end{array}\right)\  , \ \ \ \  \ 
\left(
\begin{array}{cc}
  -1 &    \\
  & -1  \end{array}\right)$$
   are in SU(2). Therefore, the component of $Q$ in the sector with  Hamming weight 1 is fixed up to a plus/minus sign.

\subsection{2-qubit semi-universality of the gate sets in Theorem \ref{Thm1}}

As we will explain in the following,  the above characterization 
of 2-qubit energy-conserving unitaries 
immediately implies the semi-universality of gate sets \ref{gates_phaseZ} and \ref{gates_phaseH} in \cref{Thm1},  for the special case of $n=2$ qubits.  
This, in particular,  means that using any of these gate sets one can realize the 
 family of  unitaries
$$    \exp(\i\alpha R): \alpha=(-\pi,\pi]\ ,$$
where $R=(X\otimes X + Y\otimes Y)/2$.
 
 Then, in the rest of the paper, we consider circuits that only contain this family of  unitaries as well as the $S$ gate and prove \cref{Thm1} for this gate set, which is a special case of gate set \ref{gates_phaseH} (In other words,  we  
show how a general energy-conserving unitary can be realized with these elementary gates and a single ancilla qubit). This result combined with the following discussion proves \cref{Thm1} for  gate sets \ref{gates_phaseZ} and \ref{gates_phaseH}.  Crucially, as we will show below,  the gates in each set  can be constructed using the other gate set exactly, using a finite $\mathcal{O}(1)$ number of gates. Therefore, to implement a general energy-conserving unitary, the number of required gates from each gate set is the same for all gate sets, up to an $\mathcal{O}(1)$ constant.

Note that, due to $\mathbb{Z}_2$ symmetry of XY interaction in \cref{Z22}, the gate set \ref{gates_XY_only},
which only includes $\exp(\i\alpha R): \alpha=(-\pi,\pi]$ without single-qubit rotations around z, is not semi-universal.   However, as we will explain in \cref{ss:Z2_symmetry}, this $\mathbb{Z}_2$ symmetry can be broken using a single ancilla qubit. It follows that in this case semi-universality is achievable with one extra qubit.

\subsubsection{Semi-universality of gate  set \ref{gates_phaseZ}}\label{ss:universal_set_gates_phaseZ}

First, consider the gate set
\begin{align} \label{eq:gates_siswap_Z}
    \siSWAP, \text{ and } \exp(\i\phi Z): \phi\in(-\pi,\pi] \,.
\end{align}
\cref{Q-decomp} already shows that any unitary in $\mathcal{SV}_2^{\rm U(1)}$ can be realized exactly with this gate set, using the following circuit:
$$
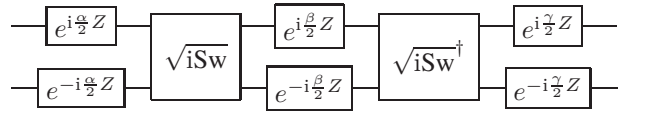
\Qcircuit @C=1em @R=.7em {
& \gate{e^{\i \frac\alpha2 Z}} &  \multigate{1}{\sqrt{\i \text{Sw}}} &\gate{e^{\i \frac\beta2 Z}}  &  \multigate{1}{\sqrt{\i \text{Sw}}^\dag} &\gate{e^{\i \frac\gamma2 Z}} & \qw \\
& \gate{e^{-\i \frac\alpha2 Z}} & \ghost{\sqrt{\i \text{Sw}}} & \gate{e^{-\i \frac\beta2 Z}} & \ghost{\sqrt{\i \text{Sw}}^\dag} & \gate{e^{-\i \frac\gamma2 Z}} &\qw }$$
where $\alpha,\beta$ and $\gamma$ come from the Euler decomposition.
It is also worth noting that one can save the number of gates at the cost of introducing relative phases between sectors with different Hamming weights, i.e., constructing an energy-conserving unitary outside $\mathcal{SV}^{U(1)}_2$. Note that in the decomposition of matrix $Q$ in  \cref{Q-decomp}  $e^{\i\alpha (Z_1-Z_2)}=e^{\i 2\alpha Z_1} e^{-\i\alpha (Z_1+Z_2)} $, and  the unitary $e^{-\i\alpha (Z_1+Z_2)} $ commutes with all other unitaries in this decomposition and can be absorbed in the diagonal matrix $D$. Similar argument works for $e^{\i\beta (Z_1-Z_2)}$ and $e^{\i\gamma (Z_1-Z_2)}$. We conclude that, up to relative phases  between sectors with different Hamming weights, any 2-qubit energy-conserving circuit has a decomposition as 
$$
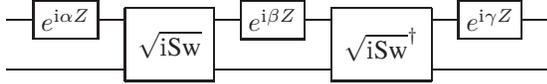
\Qcircuit @C=1em @R=.7em {
& \gate{e^{\i \alpha Z}} &  \multigate{1}{\sqrt{\i \text{Sw}}} &\gate{e^{\i \beta Z}}  &  \multigate{1}{\sqrt{\i \text{Sw}}^\dag} &\gate{e^{\i \gamma Z}} & \qw \\
&  \qw & \ghost{\sqrt{\i \text{Sw}}} & \qw  & \ghost{\sqrt{\i \text{Sw}}^\dag} &\qw&\qw }$$

For future applications, it is also useful to consider the following  realization of {the unitary $\exp(\i\alpha R)= \exp(\i\alpha [XX+YY]/2)$:}% in gate set \ref{gates_phaseH}:
$$\begin{minipage}{1em}
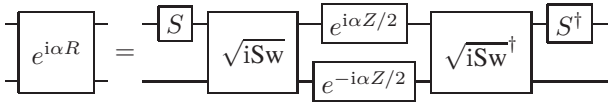
\Qcircuit @C=0.5em @R=1.25em{
 & \multigate{1}{e^{\i\alpha R}} & \qw\\
 & \ghost{e^{\i\alpha R}} & \qw
}
\end{minipage}=\begin{minipage}{1em}
\Qcircuit @C=0.6em @R=.7em {
 & \gate{S} &  \multigate{1}{\sqrt{\i \text{Sw}}} &\gate{e^{-\i \alpha Z/2}}  &  \multigate{1}{\sqrt{\i \text{Sw}}^\dag} &\gate{S^\dag} & \qw \\
 &  \qw & \ghost{\sqrt{\i \text{Sw}}} &\gate{e^{\i \alpha Z/2}}  & \ghost{\sqrt{\i \text{Sw}}^\dag} &\qw&\qw }
\end{minipage}
$$

Therefore, any circuit that involves gates  $\exp(\i\alpha R)$ for arbitrary values of $\alpha\in(-\pi,\pi]$, can be transformed to a circuit that only involves gates  $\sqrt{\i \text{SWAP}}$ and single-qubit rotations around $z$, and this increases the number of entangling gates by, at most, a factor of 2.

%\subsubsection{Example of gate set \ref{gates_UUZ}: Heisenberg interaction}
\subsubsection{Variation of gate set \ref{gates_phaseZ}: Heisenberg interaction}

A variant of the gate set in \cref{eq:gates_siswap_Z} is %can be obtained by replacing the $\siSWAP$ gate with unitary
\begin{align} \label{eq:gates_Heis_Z}
    \exp(\i \pi R_{\text{Heis}}/4), \text{ and } \exp(\i\phi Z): \phi\in(-\pi,\pi] \,,
\end{align}
where  $R_{\text{Heis}}:=(X\otimes X + Y \otimes Y + Z \otimes Z)/2$
is the Heisenberg interaction. We observe that 
\bes
\begin{align}
 \exp(\i \pi R_{\text{Heis}}/4)&=\siSWAP \exp(\i \pi Z_1 Z_2/8)\\ &=\exp(\i \pi Z_1 Z_2/8)\siSWAP \ .
\end{align}
\ees
Since $\exp(\i \pi Z_1 Z_2/8)$ commutes with the single-qubit rotations around z,   in the above circuit we can replace  $\siSWAP$ 
with $ \exp(\i \pi R_{\text{Heis}}/4)$. That is,  any $Q\in\mathcal{SV}_2^{\rm U(1)}$ can be realized with the following circuit

$$
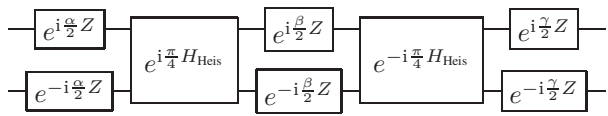
\Qcircuit @C=0.7em @R=.7em {
& \gate{e^{\i \frac\alpha2 Z}} &  \multigate{1}{e^{\i\frac{\pi}{4} H_{\text{Heis}}}} &\gate{e^{\i \frac\beta2 Z}}  &  \multigate{1}{e^{-\i\frac{\pi}{4} H_{\text{Heis}}}} &\gate{e^{\i \frac\gamma2 Z}} & \qw \\
& \gate{e^{-\i \frac\alpha2 Z}} & \ghost{e^{\i\frac{\pi}{4} H_{\text{Heis}}}} & \gate{e^{-\i \frac\beta2 Z}} & \ghost{e^{-\i\frac{\pi}{4} H_{\text{Heis}}}} & \gate{e^{-\i \frac\gamma2 Z}} &\qw }$$

Hence, the gate set in \cref{eq:gates_Heis_Z} is semi-universal.

\subsubsection{Semi-universality with gate set \ref{gates_phaseH}}

Recall the map defined in \cref{eq:VV1}, which gives the matrix representation of the component of 2-qubit operators in the sector with Hamming weight 1, relative to $|01\rangle, |10\rangle$ basis. Applying this map we find 
\bes \label{eq:H_h}
\begin{align}
H_{\text{int}} &\longrightarrow  H_{\text{int}}\sct1\\ 
S_1H_{\text{int}}S_1^\dag &\longrightarrow  S H_{\text{int}}\sct1 S^\dag
\end{align}
\ees
where $H_{\text{int}}\sct1$ is a $2\times 2$ Hermitian matrix and $S_1 = S\otimes\mathbb{I}$ denotes the $S$ gate on the first qubit. The fact that $H_{\text{int}}$ is not diagonal in the computational basis implies that $H_{\text{int}}\sct1$  is not diagonal. Then, we apply the following lemma which is shown in \cref{app:proof_decomp_HSHS}.
\begin{lemma} \label{lem:decomp_HSHS}
Suppose $H$ is a $2\times 2$ non-diagonal Hermitian matrix. Then, any unitary $U\in SU(2)$ has a decomposition as 
\be
U=\prod_{j=1}^l  \big[\exp(\i \alpha_j H) S \exp(\i \beta_j H) S^\dag\big]\ ,
\ee
where $\alpha_j,\beta_j\in\mathbb{R}$, and the length of this sequence, $l$, is bounded by a constant independent of $U$.
\end{lemma}

This lemma together with the correspondence in \cref{eq:H_h} imply that  any  2-qubit energy-conserving unitary $Q\in\mathcal{SV}_2^{U(1)}$, 
has a decomposition as
\begin{align}\label{dectr}
    Q = \prod_{j=1}^l \big[\exp(\i\alpha_j H_{\text{int}}) S_1  \exp(\i\beta_j H_{\text{int}} ) S_1^\dag\big] \, 
\end{align}
where $l$, the length of this sequence, does not depend on $Q$. We conclude that any $Q\in\mathcal{SV}_2^{U(1)}$
can be realized with Hamiltonian $H_{\text{int}}$ and $S$ gates, which means gate set \ref{gates_phaseH} is semi-universal on $n=2$ qubits.

\subsubsection{Example of gate set \ref{gates_phaseH}: XY interaction}

For XY interaction, the realization of a general 2-qubit energy-conserving unitary $Q\in\mathcal{SV}_2^{U(1)}$ has a simple form. Applying the Euler decomposition to $Q\sct1$ defined in \cref{Q}, this time for the x and y axes, we obtain 
\begin{align} \label{eq:XYeuler}
    Q\sct1 = e^{\i\gamma X}e^{\i\beta Y}e^{\i \alpha X}= e^{\i\gamma X} S e^{\i\beta X} S^\dag e^{\i \alpha X}\ ,
\end{align}  
for $\gamma, \beta, \alpha\in(-\pi,\pi]$. 
Recall that by the correspondence in \cref{eq:VV1}, in the subspace with Hamming weight 1 the action of $\exp(\i\alpha R)$  relative to the basis $\ket{01}$ and $\ket{10}$ is described by the unitary 
 $\exp(\i \alpha X)$ and the action of $S$ gate on the first qubit is $S = e^{\frac{\i\pi}{4}}\exp(-\i\pi Z/4)$.  It follows that an arbitrary unitary $Q\in\mathcal{SV}_2^{U(1)}$ can be realized as by the circuit
$$
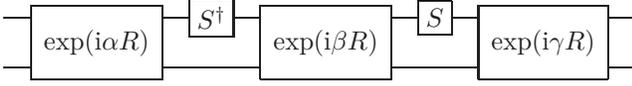
\Qcircuit @C=1em @R=.7em {
 &  \multigate{1}{\exp(\i\alpha R)} &\gate{S^\dag}  &  \multigate{1}{\exp(\i\beta R)} &\gate{S} & \multigate{1}{\exp(\i\gamma R)} & \qw \\
 &  \ghost{\exp(\i\alpha R)} & \qw & \ghost{\exp(\i\beta R)} &\qw & \ghost{\exp(\i\gamma R)} & \qw }$$

\subsubsection{Example of gate set \ref{gates_phaseH}: Heisenberg interaction}

Another canonical example of non-diagonal energy-conserving Hamiltonians  is the Heisenberg interaction $R_{\text{Heis}}=(X\otimes X+Y\otimes Y+ Z\otimes Z)/2$.  In this case, again the decomposition in \cref{dectr} finds   a simple form, which can be seen by noting that 
\begin{align}
     Z_1\exp(-\i\alpha R_{\text{Heis}})  Z_1 &\exp(\i\alpha R_{\text{Heis}}) \nonumber \\
    &= Z_1\exp(-\i\alpha R) Z_1 \exp(\i\alpha R)  \nonumber \\
    &=\exp(\i 2\alpha R) \label{eq:XYZtoXY} \ ,
\end{align}
where we have used the facts that $R_{\text{Heis}}=R+Z\otimes Z/2$,  
$[R, Z\otimes Z] = 0$, and $(Z\otimes \mathbb{I}) R(Z\otimes \mathbb{I})=- R $. 
The circuit diagram is shown as follows:
$$\begin{minipage}{1em}
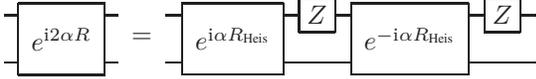
\vspace{0.2em}
\Qcircuit @C=0.5em @R=.9em{
 & \multigate{1}{e^{\i 2 \alpha R}} & \qw\\
 & \ghost{e^{\i 2 \alpha R}} & \qw
}
\end{minipage}~=~\begin{minipage}{1em}
\Qcircuit @C=0.6em @R=.7em {
 & \multigate{1}{e^{\i\alpha R_{\text{Heis}}}} &\gate{Z}  &  \multigate{1}{e^{-\i\alpha R_{\text{Heis}}}} &\gate{Z} & \qw \\
 & \ghost{e^{\i\alpha R_{\text{Heis}}}} & \qw  & \ghost{e^{-\i\alpha R_{\text{Heis}}}} &\qw&\qw }
\end{minipage}
$$
Therefore, gates $\exp(\i\alpha R_{\text{Heis}}): \alpha\in(-\pi,\pi]$ together with $S$ gate generate all gates in 
 {the previous example}, which are semi-universal for for $n=2$.

\subsection{Single-qubit rotations around z\\ (Proof of \cref{Thm1} for the special case of $n=1$)}
Using the above results, it can be easily shown that   any single-qubit rotation around z axis can be realized  using the gate set \ref{gates_phaseH} in \cref{Thm1}, i.e., $S$ gates and $\exp(\i \alpha H_\text{int}): 
\alpha\in{\R}$, and a single ancilla qubit. To see this note that the family of diagonal unitaries
\be\label{rotz}
e^{\i\phi Z}\otimes e^{-\i\phi Z}\ \ \ \ \ \ 
 :\  \phi\in(-\pi,\pi] \  , 
\ee
are in $\mathcal{SV}_2^{\text{U}(1)}$ and therefore are realizable by the aforementioned gate set. Then, for any single-qubit state $|\psi\rangle\in\mathbb{C}^2$, we have
\be
(e^{\i\phi Z}\otimes e^{-\i\phi Z})|\psi\rangle\otimes |0
\rangle_{\text{anc}}=e^{-\i\phi} e^{\i\phi Z}|\psi\rangle\otimes |0\rangle_{\text{anc}}\ , 
\ee
which means that, up to a global phase, unitary $e^{\i\phi Z}$ can be realized on a  system with $n=1$ qubit, using a single ancilla qubit. This proves  \cref{Thm1} for gate set \ref{gates_phaseH} in the special case of $n=1$ qubit.

It is worth noting that in the case of XY interaction the following circuit realizes unitaries in \cref{rotz}

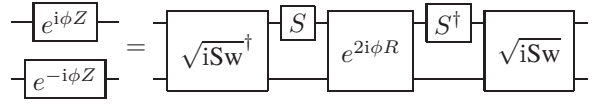
$$\begin{minipage}{1em}
\vspace{0.25em}
    \Qcircuit @C=0.5em @R=.7em {
    & \gate{e^{\i\phi Z}} & \qw \\
    & \gate{e^{-\i\phi Z}} & \qw }
\end{minipage}
= \begin{minipage}{1em}
    \Qcircuit @C=0.5em @R=0.95em {
&\multigate{1}{\siswap^\dag} & \gate{S}& \multigate{1}{e^{2\i\phi R}} & \gate{S^\dag} & \multigate{1}{\siswap} &  \qw \\
&       \ghost{\siswap^\dag} & \qw     &        \ghost{e^{2\i\phi R}} & \qw           &        \ghost{\siswap} &  \qw }
\end{minipage}$$

\section{3-qubit energy-conserving unitaries}\label{sec:3qubit}

Next, we study 3-qubit energy-conserving unitaries and show how they can be realized with gates  $\exp(\i\alpha R): \alpha\in(-\pi,\pi]$ and $S$ gate.  These methods can then be generalized to implement energy-conserving unitaries on an arbitrary number of qubits.  We start by constructing a useful  2-level 3-qubit energy-conserving unitary, using only combinations of $\i\text{SWAP}$ and  $\i\text{SWAP}^\dag$ gates.

\subsection{A useful 2-level unitary: Two controlled-$Z$ gates } \label{sec:CZ}

Using \cref{eq:decompiSWAP} one can easily show the identities 
\bes\label{swap2}
\begin{align}
\i \text{Sw}_{13}\ \i \text{Sw}^\dag_{23}\ \i \text{Sw}_{12}\ \i \text{Sw}_{23}&= \exp(\i \frac{\pi}{4} Z_2 [Z_3-Z_1])\\ 
\i \text{Sw}_{13}\ \i \text{Sw}_{23}\  \i \text{Sw}_{12}\ \i \text{Sw}_{23}&= -\i \exp(\i \frac{\pi}{4} Z_2 [Z_1+Z_3])\ ,
 \end{align}
 \ees
 where  $\i \text{Sw}_{ij}$ denotes $\i \text{SWAP}$ gate on qubits $i$ and $j$, defined in \cref{iSWAP}.  
The second identity and similar other identities can be obtained by replacing  $\i \text{Sw}_{ij}$ with $\i \text{Sw}^\dag_{ij}$, or vice versa.  It is worth noting that the specific combination of unitaries $\i \text{Sw}_{ij}$ appearing in these identities  has a nice interpretation in terms of the permutation group:  Applying 
\cref{eq:decompiSWAP}
 to all $\i \text{Sw}_{ij}$ in the left-hand side of \cref{swap2}, the left-hand side becomes  $ \text{Sw}_{13} \text{Sw}_{23} \text{Sw}_{12} \text{Sw}_{23}$,  up to a diagonal unitary in the computational basis,  where  $\text{Sw}_{ij}$ denotes the SWAP unitary. However, this combination of swaps is equal to the identity operator, which can be seen from the basic properties of the permutation group. We conclude that  $\i \text{Sw}_{13}\i \text{Sw}_{23} \i \text{Sw}_{12} \i \text{Sw}_{23}$ is equal to a unitary diagonal in the computational basis, that is determined by the right-hand side of \cref{swap2}.

%. It follows that under the action of  the above combination of unitaries, elements of the computational basis remain invariant, up to a phase, 

The first identity implies that 
\bes
\begin{align}
\i \text{Sw}_{13}&\i \text{Sw}_{23}^\dag \i \text{Sw}_{12} \i \text{Sw}_{23} S^\dag_1S_3= \exp(\i \frac{\pi}{4} [(Z_3-Z_1)(Z_2-I)])\nonumber\\ &=|0\rangle\langle 0|_2\otimes \mathbb{I}_{13}+|1\rangle\langle 1|_2\otimes Z_1 Z_3\ 
\\ &=\CZ_{12}\CZ_{23}\  .\label{art}
\end{align}
\ees
See \cref{tab:identities} for the corresponding circuit identity. This unitary is diagonal in the computational basis and is 2-level, i.e., it acts non-trivially only on states $|011\rangle$ and $|110\rangle$, and it gives $-1$ sign to both of these states. Note that using the second identity in \cref{swap2} we can obtain a similar construction of this unitary using 4 $\iSWAP$ gates. \\

For future applications, we also note that \cref{swap2} implies 
\be\label{Def-F}
F_{123}:=\iswap_{13} \iswap^\dag_{23}\iswap_{12}=(S_1\otimes S^\dag_2\otimes Z_3)\CZ_{12} \text{Sw}_{23}
\ee
See  \cref{tab:identities}
for the corresponding circuit identity.
 Also, note that changing each  $\i\text{Sw}$ gate to  $\i\text{Sw}^\dag$, or, vice versa, in the circuit on the left-hand side is equivalent to changing $S$ gate to $S^\dag$ and vice versa in the right-hand side. (This can be seen by considering the complex conjugate of both sides in the computational basis.)

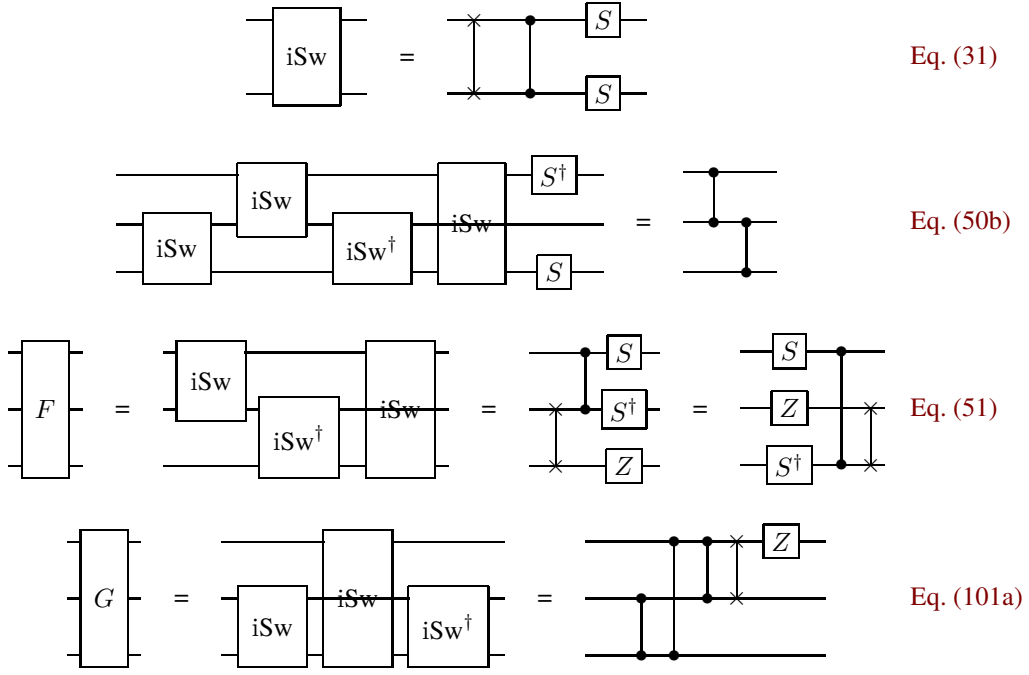
\begin{figure*}[t]
    \setlength\tabcolsep{5pt}
    \renewcommand\arraystretch{1.5}
    \begin{center}
        
    \begin{tabular}{cl}
    \begin{minipage}{3cm}
        \Qcircuit @C=1em @R=1.9em {
        & \multigate{1}{\i\text{Sw}} & \qw  \\
        & \ghost{\i\text{Sw}} & \qw  \\
        }
    \end{minipage}
    \quad = \quad
    \begin{minipage}{3cm}
        \Qcircuit @C=1em @R=1.5em {
        & \qswap & \qw&  \ctrl{1} & \qw  & \gate{S}& \qw  \\& \qswap \qwx   &\qw 
        & \ctrl{-1} & \qw &\gate{S}& \qw \\
        }
    \end{minipage}
    & \cref{eq:decompiSWAP}
    \\
    \\
    
    \begin{minipage}{1em}
    \Qcircuit @C=1em @R=.7em {
    & \qw &  \multigate{1}{\i\text{Sw}}& \qw &   \multigate{2}{\i\text{Sw}}& \gate{S^\dag} & \qw \\
    & \multigate{1}{\i\text{Sw}} & \ghost{i\text{Sw}} &  \multigate{1}{\i\text{Sw}^\dag} & \qw & \qw & \qw \\ 
      & \ghost{\i\text{Sw}}& \qw &  \ghost{\i\text{Sw}^\dag}&  \ghost{\i\text{Sw}} & \gate{S}&\qw  }
    \end{minipage}
    \quad = \quad
    \begin{minipage}{1em}
    \Qcircuit @C=1em @R=1.65em {
    &  \ctrl{1}& \qw &\qw    \\ 
     &  \ctrl{-1}& \ctrl{1} &\qw   \\
      & \qw& \ctrl{-1} &\qw 
       }
    \end{minipage}
    & \cref{art}
    \\
    \\
    
    \begin{minipage}{1em}
    \Qcircuit @C=.5em @R=1.25em @!R {
        & \multigate{2}{F} & \qw \\
        & \ghost{F} & \qw \\
        & \ghost{F} & \qw
        }
    \end{minipage}
    \quad = \quad
    \begin{minipage}{1em}
        \Qcircuit @C=.5em @R=1.25em @!R {
        &  \multigate{1}{\i\text{Sw}}   &\qw                              & \qw & \multigate{2}{\i\text{Sw}} & \qw\\
        & \ghost{i\text{Sw}}            & \multigate{1}{\i\text{Sw}^\dag} & \qw & \qw                & \qw\\        
        & \qw                           & \ghost{i\text{Sw}^\dag}         & \qw & \ghost{i\text{Sw}} & \qw\\
        }
        
    \end{minipage}
    
    \quad = \quad

    \begin{minipage}{1em}
    \Qcircuit @C=.5em @R=0.7em @!R {
    &\qw&\qw        &\qw  &\ctrl{1} &\gate{S}      &\qw \\
    &\qw&\qswap     &\qw  &\ctrl{-1}&\gate{S^\dag} &\qw \\
    &\qw&\qswap \qwx&\qw  &\qw      &\gate{Z}      &\qw
    } 
    \end{minipage}
    
    \quad = \quad

    \begin{minipage}{1em}
    \Qcircuit @C=.5em @R=0.7em @!R {
    &\qw & \gate{S}      & \qw  &\ctrl{2}  & \qw & \qw         &\qw \\
    &\qw & \gate{Z} & \qw  &\qw       & \qw & \qswap      &\qw \\
    &\qw & \gate{S^\dag} & \qw  &\ctrl{-2} & \qw & \qswap \qwx &\qw
    }
    \end{minipage}
         & \cref{Def-F} %tabular
         \\\\
    
    \begin{minipage}{1em}
    \Qcircuit @C=.5em @R=1.25em @!R {
        & \multigate{2}{G} & \qw \\
        & \ghost{G} & \qw \\
        & \ghost{G} & \qw
        }
    \end{minipage}
    \quad = \quad
    \begin{minipage}{1em}
        \Qcircuit @C=0.6em @R=1.25em {
     & \qw                         &   \multigate{2}{\i\text{Sw}}&  \qw                        &\qw      \\
     &  \multigate{1}{\i\text{Sw}} &  \qw                        &   \multigate{1}{\i\text{Sw}^\dag} &\qw  \\
     & \ghost{i\text{Sw}}          &  \ghost{i\text{Sw}}         &  \ghost{i\text{Sw}^\dag}   &\qw
     }
     \end{minipage}
     \quad = \quad
     \begin{minipage}{1em}
     \vspace{-0.5em}
        \Qcircuit @C=1em @R=0.875em @!R {
        &\qw&\qw      & \ctrl{2} &\ctrl{1}    &\qswap      &\gate{Z}&\qw \\
        &\qw&\ctrl{1} & \qw      &\ctrl{-1}   &\qswap \qwx &\qw     &\qw \\
        &\qw&\ctrl{-1}& \ctrl{-2}&\qw         &\qw         &\qw     &\qw
        } 
    \end{minipage}
    & \cref{def-G}
    \end{tabular}
    \end{center}
    \caption{Summary of useful circuit identities involving $\i\text{SWAP}$ and $S$ gates. 
    \label{tab:identities}}
\end{figure*}

\subsubsection*{ \textbf{\emph{The SWAP and controlled-$Z$   gates with an ancilla qubit}}}

\cref{art} immediately gives a method for realizing the controlled-$Z$ gate: Suppose we prepare  qubit 1 in state $|0\rangle$. Then, 
 the overall action of this circuit on qubits 2 and 3 will be CZ gate. Another important unitary transformation that can be realized in this way is the SWAP gate. In \cref{Fig1} we compare the circuit obtained in this way (the bottom circuit) with the standard way of implementing 
the SWAP unitary with  3 $\i\text{SWAP}$ gates,  which is originally presented in  \cite{schuch2003natural}.
  As explained in the caption of \cref{Fig1}, the realization of the SWAP gate with ancilla is more robust against certain types of errors, such as the fluctuations of the master clock. (It is also worth noting that the use of an ancilla  has another advantage: by measuring the ancilla qubit at the end of the process in the z basis, it is possible to detect the presence of certain $X$ errors in the circuit.)

In \cref{App:pi/4} we present other examples of identities similar to \cref{swap2}. Such identities, for instance, imply  that the gate controlled-$R_z(-\frac{\pi}{2})$ can be realized using 3 $\i\text{SWAP}$, 3 $\sqrt{\i\text{SWAP}}$ gates and an ancilla qubit.

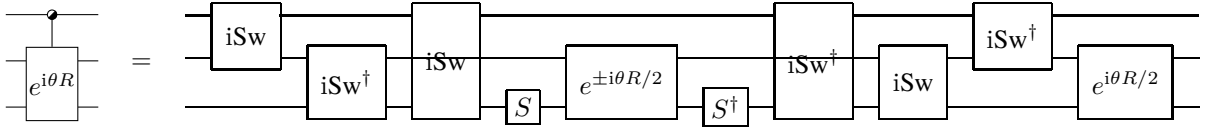
\begin{figure*}
\begin{minipage}{5em}
\begin{tikzpicture}[scale=0.82]
    \coordinate (left) at (-0.75, 0);
    \coordinate (right) at (0.75,0);
    \useasboundingbox (-0.75,-1.2) rectangle (0.75,1.2);
    \draw (left) -- (right);
    \draw ($(left)+(0,0.75)$) -- ($(right)+(0,0.75)$);
    \draw ($(left)-(0,0.75)$) -- ($(right)-(0,0.75)$);
    \node[tensor2h,minimum height=0.65*15mm] (V01) at (0,-0.375) {$e^{\i\theta R}$};
    \draw (V01) -- (0,0.75) pic[scale=0.82] {c01};
\end{tikzpicture}
\end{minipage}
\ $=$ \quad
\begin{minipage}{1em}
\vspace{0.3em}
  \Qcircuit @C=1em @R=.7em {
&  \multigate{1}{\i\text{Sw}}& \qw &   \multigate{2}{\i\text{Sw}}& \qw & \qw &   \qw & \multigate{2}{\i\text{Sw}^\dag} &  \qw  & \multigate{1}{\i\text{Sw}^\dag}&\qw & \qw  \\
 & \ghost{i\text{Sw}} &  \multigate{1}{\i\text{Sw}^\dag} & \qw & \qw & \multigate{1}{e^{\pm\i \theta R/2}} & \qw &  \qw &  \multigate{1}{\i\text{Sw}}  & \ghost{\i\text{Sw}^\dag} &  \multigate{1}{e^{\i \theta R/2}} & \qw    \\ 
  & \qw &  \ghost{\i\text{Sw}^\dag}&  \ghost{\i\text{Sw}} & \gate{S} & \ghost{e^{\pm\i \theta R/2}}& \gate{S^\dag} &  \ghost{\i\text{Sw}^\dag}&  \ghost{\i\text{Sw}} & \qw &\ghost{e^{\i \theta R/2}} & \qw }
\end{minipage}
  \caption{The circuit corresponding to \cref{x-rotation}: 
  Assuming the gate in the middle (the fifth  gate) is $\exp(-\i R\theta/2)$ with $R=(X\otimes X+Y\otimes Y)/2$, this circuit implements controlled-$\exp(\i\theta R)$, i.e.,  the 2-level 3-qubit unitary $U_1$ defined in \cref{x-rotation}. This family includes  useful unitaries, such as controlled-$\i\text{SWAP}$ (corresponding to  $\theta=\pi/2$, which means the fifth gate is 
 $\siSWAP^\dag$ and the tenth gate is $\siSWAP$)   as well as the controlled-$\sqrt{\i\text{SWAP} }$ (corresponding to $\theta=\pi/4$, which means the fifth gate is  $\exp(-\i\pi R / 8)$ and the tenth gate is $\exp(\i\pi R / 8)$). 
   On the other hand, when the gate in the middle is $\exp(+\i R\theta/2)$, the circuit realizes 
  $U_0$ in \cref{x-rotation}, which applies   $\exp(\i\theta R/2)$  when the control qubit is in state $|0\rangle$ rather than  $|1\rangle$. As we explain in the next section, using this construction recursively, we can obtain all energy-conserving 2-level unitaries with determinant  1.  }\label{Fig2} 
\end{figure*}

\begin{figure*}
\centering
\begin{tikzpicture}[yscale=0.667]
    \coordinate (Qleft) at (-3,0);
    \coordinate (Qright) at (-2,0);
    \draw (Qleft) -- (Qright);
    \draw ($(Qleft)+0.75*(0,1)$) -- ($(Qright)+0.75*(0,1)$);
    \draw ($(Qleft)-0.75*(0,1)$) -- ($(Qright)-0.75*(0,1)$);
    \node[tensor3h,minimum height=0.5*25mm] (Q) at ($(Qleft)!0.5!(Qright)$) {$Q$};
    %pic{c0} pic{c1} tensor2h
    \coordinate (start0) at (0, 0);
    \coordinate (start1) at (4.5, 0);
    \coordinate (left) at (-1.25, 0);
    \node (equal) at ($(Qright)!0.5!(left)$) {$=$};
    \coordinate (right) at ($(start1)+(4.25,0)$);
    \draw (left) -- (right);
    \draw ($(left)+0.75*(0,1)$) -- ($(right)+0.75*(0,1)$);
    \draw ($(left)-0.75*(0,1)$) -- ($(right)-0.75*(0,1)$);
    \node[tensor2h,minimum height=0.5*15mm] (V01) at ($(start0)-(0,0.375)$) {};%$V^{001,010}$
    \node[tensor3h,minimum height=0.5*25mm] (V02) at ($(start0)+(1.5,0)$) {};%$V^{001,100}$
    \node[tensor2h,minimum height=0.5*15mm] (V03) at ($(start0)+(3,0.375)$) {};%$V^{010,100}$
    \node[tensor2h,minimum height=0.5*15mm] (V11) at ($(start1)-(0,0.375)$) {};%$V^{110,101}$
    \node[tensor3h,minimum height=0.5*25mm] (V12) at ($(start1)+(1.5,0)$) {};%$V^{110,011}$
    \node[tensor2h,minimum height=0.5*15mm] (V13) at ($(start1)+(3,0.375)$) {};%$V^{101,011}$
    \foreach \ctl in {0,1}
    {
        \coordinate (control\ctl1) at ($(start\ctl) + 0.75*(0,1)$);
        \coordinate (control\ctl2) at ($(start\ctl) + 0.75*(1,0)$);
        \coordinate (control\ctl3) at ($(V\ctl3 |- 0,-0.75)$);
        \draw (control\ctl1) -- (V\ctl1);
        \draw (control\ctl2) |- ($(V\ctl2.north)+0.75*(0,0.5)$) -- (V\ctl2);
        \draw (control\ctl2) -- (V\ctl2.east);
        \draw (control\ctl3) -- (V\ctl3);
        \foreach \cnt in {1,2,3}
        {
            \pic[scale=1] at (control\ctl\cnt) {c\ctl};
        }
    }
    %\node[tensor2h,minimum height=15mm] {$e^{\i\frac{\theta_1}{2}(Z_3-Z_{\rm anc})}$};
\end{tikzpicture}
    \caption{\label{fig:SV3}The circuit for implementing arbitrary unitaries in $\mathcal{SV}_3^{\rm U(1)}$, i.e., all 3-qubit  energy-conserving unitaries with the property that the determinant of the component of unitary in each Hamming weight sector is 1. Here, each gate is in the form of \cref{Eq51} and hence can be realized by composing 3 circuits in the form of \cref{Fig2}.
    }
\end{figure*}
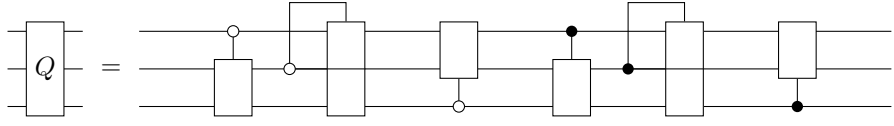

\subsection{General 2-level energy-conserving unitaries on 3 qubits}\label{ss:2levelSV3}

Next, we show how general 2-level energy-conserving unitaries in $\mathcal{SV}_3^{U(1)}$ can be realized. Here, we 
 adapt the approach developed in  \cite{marvian2023non} for quantum circuits with general Abelian symmetries to the case of U(1) symmetry.  To achieve this we use the gate $F_{123}$ defined in \cref{Def-F}.

 First, note that sandwiching $\exp(\i\theta R/2)$ on qubits 2 and 3 with $F_{123}$ and $F^\dag_{123}$ and using the facts that 
$S^\dag\otimes S^\dag$ and SWAP commute with $R$, 
we obtain 
\begin{align}\label{ew3}
&F^\dag_{123} S^\dag_3 \exp(\i \frac{\theta}{2} R_{23}) S_3 F_{123}\\ &=|0\rangle\langle 0|_1\otimes \exp(\i \frac{\theta}{2} R_{23})+|1\rangle\langle 1|_1\otimes\exp(-\i \frac{\theta}{2} R_{23})\ , \nonumber
\end{align}
or, equivalently,  the circuit identity
$$
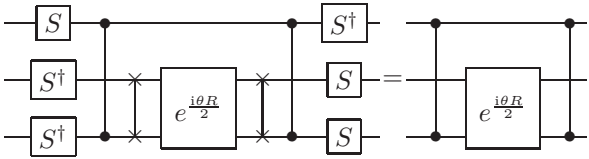
\Qcircuit @C=.5em @R=0.7em @!R {
&\qw & \gate{S}      & \qw  &\ctrl{2}  & \qw & \qw         &\qw & \qw                             & \qw & \qw         &\qw  &\ctrl{2}  & \qw & \gate{S^\dag}& \qw &    &  & \qw  &\ctrl{2}  & \qw & \qw                             & \qw &\ctrl{2}  & \qw \\
&\qw & \gate{S^\dag} & \qw  &\qw       & \qw & \qswap      &\qw & \multigate{1}{e^{\frac{\i\theta R}{2}}} & \qw & \qswap      &\qw  &\qw       & \qw & \gate{S}     & \qw & ~=~&  & \qw  &\qw       & \qw & \multigate{1}{e^{\frac{\i\theta R}{2}}} & \qw &\qw       & \qw \\
&\qw & \gate{S^\dag} & \qw  &\ctrl{-2} & \qw & \qswap \qwx &\qw & \ghost{e^{\frac{\i\theta R}{2}} }       & \qw & \qswap \qwx &\qw  &\ctrl{-2} & \qw & \gate{S}     & \qw &    &  & \qw  &\ctrl{-2} & \qw & \ghost{e^{\frac{\i\theta R}{2}}}        & \qw &\ctrl{-2} & \qw   
}$$

This implies
\begin{align}\label{x-rotation}
&\exp(\i \frac{\theta}{2} R_{23}) F^\dag_{123} S^\dag_3 \exp(\i \frac{(-1)^b\theta}{2} R_{23}) S_3 F_{123}
\\ &=|\overline{b}\rangle\langle \overline{b}|_1\otimes \mathbb{I}_{23}+|b\rangle\langle b|_1\otimes \exp(\i  \theta R_{23})=: U_b \ \ : \  b=0,1\ ,\nonumber
\end{align}
which corresponds to the circuit in \cref{Fig2}, where $\overline{b}$ denotes the negation of $b$. This unitary is 2-level: it acts as $\exp(\i \theta X)$ in  the 2D subspace 
\be\label{sub}
\text{span}_\mathbb{C}\{|\widetilde{0}\rangle\equiv |b01\rangle\ ,\ |\widetilde{1}\rangle\equiv |b10\rangle \} \ \ : \  b=0,1 \ ,
\ee
whereas it acts as the identity operator   elsewhere.  By sandwiching the second (or, third) qubit between $S^\dag$ and $S$, we obtain the 2-level unitary 
\be\label{y-rotation}
{S_2U_bS_2^\dag =}S^\dag_3 U_b S_3=|b\rangle\langle b|\otimes \mathbb{I}+|\overline{b}\rangle\langle \overline{b}|\otimes \exp({\i  \theta} L_{23})\ \ : \  b=0,1\ ,
\ee
which acts as $\exp(\i \theta Y)$ in the same 2D subspace, where we have defined
\bes
\begin{align}
L&=\frac{1}{2}(Y\otimes X-X\otimes Y)=(I\otimes S^\dag) R (I\otimes S)
\\ &=\i (|10\rangle\langle 01|-|01\rangle\langle 10|) \ .
\end{align}
\ees

Therefore, we can realize both x and y rotations in the 2D subspace in \cref{sub}. Furthermore, 
 applying the Euler decomposition in \cref{eq:XYeuler} for SU(2) unitaries, we can obtain a general SU(2) unitary in this 2D subspace.

In conclusion, this way we can realize any 2-level 3-qubit energy-conserving unitary in the form of
\begin{align}\label{Eq51}
    \Lambda^b(V) := \ketbra{\bar{b}}{\bar{b}} \otimes \mathbb{I} +  \ketbra{b}{b}\otimes V\ \ \ : \  V = \left(\begin{matrix}1 &&\\&V\sct1&\\&&1\end{matrix}\right)\ ,
\end{align} 
for any $V\sct1 \in \text{SU}(2)$, with any of the sets \ref{gates_phaseZ} and \ref{gates_phaseH} in \cref{Thm1}. 
 Here, $\Lambda^b(V)$ denotes a controlled gate that applies $V$,  when the control qubit is $b$, namely a controlled-$V$ gate for $b=1$ or an anti-controlled-$V$ gate for $b=0$.  
The number of gates used in this construction is:
\begin{enumerate}
    \item 32 gates in set \ref{gates_phaseH} (choosing $H_\text{int}=R$), if we also allow the application of $S^\dag$, or 40 gates if we implement $S^\dag$ as $S^3$ in set \ref{gates_phaseH};
    \item 80 gates in set \ref{gates_phaseZ}, allowing $\sqrt{\i\text{SWAP}}^\dag$ (one needs 6 gates to implement $\exp(\pm\i\theta R)$; see \cref{ss:universal_set_gates_phaseZ}), or 152 if one implement $\i\text{SWAP}^\dag$ as $\sqrt{\i\text{SWAP}}^6$ and $\sqrt{\i\text{SWAP}}^\dag$ as $\sqrt{\i\text{SWAP}}^7$.
\end{enumerate}

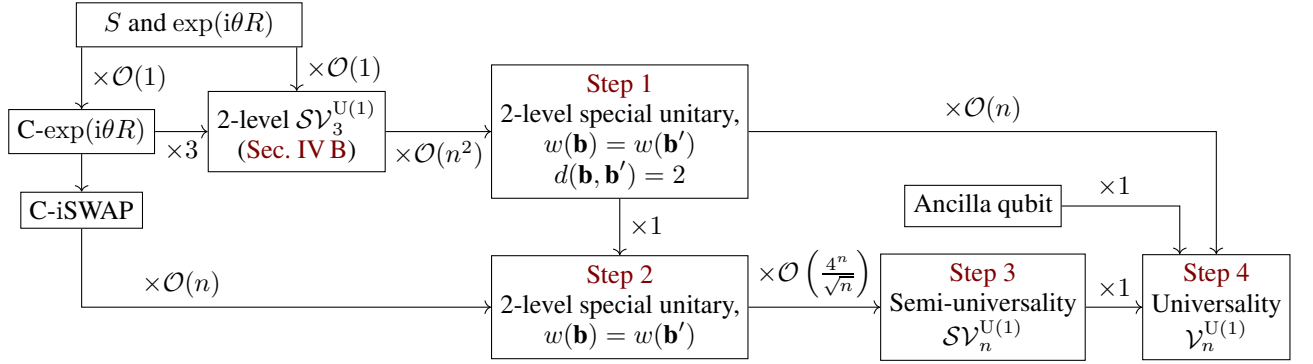
\begin{figure*}
    \centering
    \begin{tikzpicture}
        \node[draw=black,align=center] (CR){C-$\exp(\i\theta R)$};
        \node[draw=black,align=center,right=0.7 of CR] (SV3) {{2-level} $\mathcal{SV}_3^{\rm U(1)}$\\(\cref{ss:2levelSV3})};
        \node[draw=black,minimum width=3cm] (basic) at ($(CR.north)!0.5!(SV3.north) + (0,1)$) {$S$ and $\exp(\i\theta R)$};
        \node[draw=black,align=center,below=0.5 of CR] (ciswap) {C-$\i\text{SWAP}$};
        \node[right=1.4cm of SV3, draw=black,align=center] (d2) {\cref{step:dist2}\\ 2-level special unitary, \\ $w(\textbf{b})=w(\textbf{b}')$\\ $d(\textbf{b},\textbf{b}')=2$};
        \node[below=0.75 of d2, draw=black,align=center] (dany) {\cref{step:equalweight}\\ 2-level special unitary, \\ $w(\textbf{b})=w(\textbf{b}')$};
        \node[right=1.75cm of dany, draw=black,align=center] (semi) {\cref{step:semi-universal}\\Semi-universality\\$\mathcal{SV}_n^{\rm U(1)}$};
        \node[right=0.75cm of semi, draw=black,align=center] (univ) {\cref{step:universal}\\Universality\\$\mathcal{V}_n^{\rm U(1)}$};
        \node[draw=black,above=0.35cm of semi] (anc) {Ancilla qubit};
        \draw[->] (CR) -- (ciswap);
        \draw[->] (CR) -- node[below](times3){$\times3$}(SV3);
        \draw[->] (basic.south-|CR) -- node[right]{$\times \mathcal{O}(1)$} (CR);
        \draw[->] (basic.south-|SV3) -- node[right]{$\times \mathcal{O}(1)$} (SV3);
        \draw[->] (SV3) -- node[below]{$\times \mathcal{O}(n^2)$} (d2);
        \draw[->] (d2) -- node[right] (d2_dany) {$\times 1$} (dany);
        \draw[->] (ciswap) |- node[above,pos=0.7](ciswap_dany){\phantom{$\times\mathcal{O}(n)$}} (dany);
        \node at (ciswap_dany -| times3) {$\times \mathcal{O}(n)$};
        \draw[->] (dany) -- node[above] {$\times\mathcal{O}\left(\frac{4^n}{\sqrt{n}}\right)$} (semi);
        \draw[->] (semi) -- node[above](times1){$\times 1$} (univ);
        \draw[->] (d2) -| node[above, pos=0.25]{$\times \mathcal{O}(n)$} (univ);
        \draw[->] (anc) -| node[above,pos=0.25](anc_univ){\phantom{$\times1$}} ($(univ.north west) + (0.5,0)$);
        \node at (anc_univ -| times1) {$\times1$};
    \end{tikzpicture}
    \caption{Overview of the method for 
 synthesizing general energy-conserving unitaries. 
 The {number on the arrow} pointing from one box $A$ to another box $B$ indicates 
 the number of uses of $A$ in the construction of $B$. 
 C-$\exp(\i\theta R)$ denotes controlled-$\exp(\i\theta R)$ which can be realized with the circuit in \cref{Fig2}. 
 }
    \label{fig:3ton}
\end{figure*}

\subsubsection*{\textbf{\emph{Example: controlled-$\i\text{SWAP}$ gate}}} 

A useful unitary, which plays a crucial role in our construction, is the 3-qubit gate controlled-$\i\text{SWAP}$ denoted as
\be
\Lambda^1(\i\text{SWAP})=|0\rangle\langle 0|\otimes \mathbb{I}+|1\rangle\langle 1|\otimes \i\text{SWAP}\ .
\ee
This unitary can be realized by the circuit in \cref{Fig2}. In particular, this unitary corresponds to  
 {$\theta=\pi/2$}, in which case the fifth and tenth gates in the circuit are $\sqrt{\i\text{SWAP}}^\dag$ and  $\sqrt{\i\text{SWAP}}$ gates, respectively.

In summary,  controlled-$\i\text{SWAP}$ can be realized with 
$$3\ \i\text{SWAP}+3\  \i\text{SWAP}^\dag+ 1\  \sqrt{\i\text{SWAP}}+1\  \sqrt{\i\text{SWAP}}^\dag+ 1S+1S^\dag .$$

\subsection{Semi-universality on 3 qubits}\label{ss:general_sv3}

By composing the above 2-level unitaries one can obtain any  unitary  $Q\in\mathcal{SV}_3^{\rm U(1)}$. Recall that any such unitary  has a decomposition as
\begin{align}\label{eq:Q_Q1Q2}
    Q = \left(\begin{matrix}
        1 & & &\\
        & Q\sct1 & &\\
        & & Q\sct2 &\\
        & & & 1
    \end{matrix}\right)\ ,
\end{align}
where $Q\sct1, Q\sct2 \in \text{SU}(3)$ act  on the subspaces with Hamming weights 1 and 2, respectively.  According to \cref{lem:2-level_decomp} to be shown later, any unitary in $\text{SU}(3)$ can be decomposed into 3 2-level unitaries in $\text{SU}(3)$.  We conclude that any 3-qubit energy-conserving unitary   $Q\in \mathcal{SV}_3^{\rm U(1)}$  has a decomposition in the form of \cref{fig:SV3}, where each gate in this circuit is a 2-level energy-conserving unitary in the form of \cref{Eq51}.

\section{Synthesis of general energy-conserving unitaries with XY and Z interactions } \label{sec:nqubit}

In this section, we explain how general energy-conserving unitaries can be obtained from 2-level 3-qubit unitaries constructed in Section \ref{ss:2levelSV3}. 
 In the following, 
 for any pair of bit strings 
 $\textbf{b},\textbf{b}' \in\{0,1\}^n$, 
 $d(\textbf{b}, \textbf{b}'):=\sum_{j=1}^{n} |b_j-b'_j|$ denotes their Hamming distance, i.e., the number of bits taking different values in $\textbf{b}$ and  $\textbf{b}'$.

Our construction, which is  illustrated in \cref{fig:3ton}, 
 is based on the following steps: 
\begin{enumerate}
\item  Synthesizing 2-level special unitaries acting on any two basis elements $|\textbf{b}\rangle$ and  $|\textbf{b}'\rangle$, with Hamming weights $w(\textbf{b})=w(\textbf{b}')$, and Hamming distance $d(\textbf{b}, \textbf{b}')=2$ (See  \cref{lem:TnOn2}). 
\item  Generalizing the previous step to the case of arbitrary Hamming distance $d(\textbf{b}, \textbf{b}')$ (See  \cref{cor56}).

%Synthesizing 2-level special  unitaries on the 2D subspace spanned by any two basis elements $|\textbf{b}\rangle$ and  $|\textbf{b}'\rangle$, with Hamming weights $w(\textbf{b})=w(\textbf{b}')$.
\item 
Synthesizing the subgroup  $\mathcal{SV}_n^{\rm U(1)}$, defined in \cref{sv}. 
\item Synthesizing the group of all energy-conserving unitaries, denoted by $\mathcal{V}_n^{\rm U(1)}$.
\end{enumerate}
Note that only in the last step, one needs to use a single ancilla qubit. It is also worth noting that some of the techniques that are used in the proofs of Steps \ref{step:dist2} to \ref{step:semi-universal} follow similar constructions that have been developed  previously  in the quantum circuit theory for circuits without symmetry constraints (See \cite{kitaev2002classical} and \cite{NielsenAndChuang}). In \cref{step:universal}, where we use an ancilla qubit,  we apply a technique that was developed previously in \cite{marvian2023non}.

\refstepcounter{step}\label{step:dist2}
\subsection*{Step \arabic{step}: Basis elements with equal Hamming weights and Hamming distance 2}

First, we show that by applying the construction in \cref{Fig2} recursively, we can realize any unitary in the form 
$V=U\sct1(\textbf{b},\textbf{b}')$, defined in \cref{eq:def_2-level}, 
where $d(\textbf{b},\textbf{b}')=2$,  $w(\textbf{b})=w(\textbf{b}')$, and 
$U\sct1(\textbf{b},\textbf{b}')$ acts unitarily on the 2-dimensional subspace spanned by $|\textbf{b}\rangle$ and $|\textbf{b}'\rangle$ and satisfies  
\be
\text{det}(V)=\text{det}(U\sct1(\textbf{b},\textbf{b}'))=1\ .
\ee
%and  $\Pi_\perp$ is the projector to the orthogonal subspace.
 Here, we denote the $2\times2$ unitary as $U\sct1$ since, as we will see below in \cref{eq:2-level_controlled},
we will think of it as the component of a 2-qubit unitary $U$ in the sector with Hamming weight 1.

For simplicity of presentation, we relabel the qubits such that the common bits of $\textbf{b}$ and $\textbf{b}'$ are labelled from 1 to $k=n-2$, and the two different bits of $\textbf{b}$ and $\textbf{b}'$ are labeled $n-1$ and $n$. Rearranging the qubits in this way, noting that $\textbf{b}$ and $\textbf{b}'$ have equal Hamming weights and Hamming distance 2, we can write 
\be
|\textbf{b}\rangle=|\textbf{c}\rangle|0\rangle|1\rangle\ , \ \ \ \ |\textbf{b}'\rangle=|\textbf{c}\rangle|1\rangle|0\rangle\ ,
\ee
where $\textbf{c}\in\{0,1\}^{n-2}$ contains the common bits of $\textbf{b}$ and $\textbf{b}'$. With this definition the target gate $U\sct1(\textbf{b},\textbf{b}')$ can be written as a multi-controlled $\mathcal{SV}_2^{\rm U(1)}$ gate as
\be \label{eq:2-level_controlled}
U\sct1(\textbf{b},\textbf{b}')= \Lambda^{\textbf{c}}(U)\ , ~\ \ \ \ 
U = \left(
\begin{array}{ccc}
1  &   &   \\
  & U\sct1  &   \\
  &   &   1
\end{array}\right)\ ,
\ee
where 
\be \label{eq:Lambda}
\Lambda^{\textbf{c}}(U) := \sum_{\textbf{c}' \neq \textbf{c}}\ket{\textbf{c}'}\bra{\textbf{c}'} \otimes \mathbb{I} + \ket{\textbf{c}}\bra{\textbf{c}}\otimes U\ ,
\ee
is a controlled-unitary  with control string $\textbf{c}$.

Now we use a recursive construction of the circuit for $\Lambda^{\textbf{c}}(U)$. (Here, we are applying a technique that was originally used in \cite{kitaev2002classical,NielsenAndChuang} for general quantum circuits.) We decompose $\textbf{c}$ as $\textbf{c} = \textbf{c}_1 \textbf{c}_2$, where $\textbf{c}_1 $ has length $\lfloor k/2 \rfloor$ and $\textbf{c}_2$ has length $\lceil k/2 \rceil$.
Since $U\sct1 \in \SU(2)$, there exists operators $A\sct1,B\sct1\in\SU(2)$ such that $A\sct1 B\sct1 A\sct1^\dag B\sct1^\dag = U\sct1$. (To see this, note that a general SU(2) unitary $U\sct1$ is equal to $\exp(\i \theta Z)$,  up to a change of basis, as  $U\sct1=W \exp(\i \theta Z) W^\dag$. Then, one can choose $A\sct1=W \exp(\i \theta Z/2) W^\dag$ and $B\sct1=\i W X W^\dag$.) 
%For example, given $U\sct1$ written in its eigenbasis as
%\begin{align}
  %  U\sct1 = \left(\begin{matrix}e^{\i\theta} & \\ & e^{-i\theta}\end{matrix}\right) \,,
%\end{align}
%one may choose $A\sct1$ and $B\sct1$ in the eigenbasis of $U\sct1$ as
%\begin{align}
%    A\sct1 =  \left(\begin{matrix}e^{\i\theta/2} & \\ & e^{-i\theta/2}\end{matrix}\right), \quad B\sct1 =  \left(\begin{matrix} & \i  \\ \i & \end{matrix}\right) \,.
%\end{align}

Let $A:=1\oplus A\sct1 \oplus 1$ and $B:=1\oplus B\sct1\oplus 1$ be the extensions of $A\sct1$ and $B\sct1$ to 2-qubit unitaries, in the  same way $U$ extends $U\sct1$ to 
a 2-qubit unitary. Then $A B A^\dag B^\dag = U$.
We can therefore decompose $\Lambda^{\textbf{c}}(U)$ as 
\be\label{tr1}
\Lambda^{\textbf{c}}(U)=\Lambda^{\textbf{c}_1}(A)\Lambda^{\textbf{c}_2}(B)\Lambda^{\textbf{c}_1}(A^\dag)\Lambda^{\textbf{c}_2}(B^\dag)\ ,
\ee
with the first $\lfloor k/2 \rfloor$ qubits being the control qubits for $\Lambda^{\textbf{c}_1}(A)$ and $\Lambda^{\textbf{c}_1}(A^\dag)$, and the next $\lceil k/2 \rceil$ qubits being the control qubits for $\Lambda^{\textbf{c}_2}(B)$ and $\Lambda^{\textbf{c}_2}(B^\dag)$, as shown in \cref{fig:circuit_XYXY}.

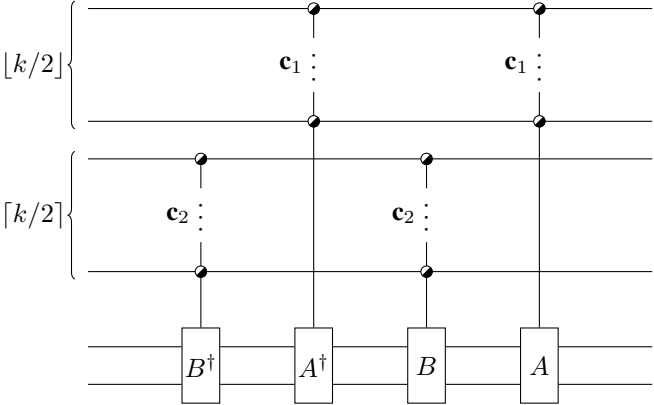
\begin{figure}[h]
    \centering
    \begin{tikzpicture}
        \coordinate (t1) at (0, 0);
        \coordinate (t2) at (0, -0.5);
        \def\cy{{0, 1, 2.5, 3, 4.5}}
        \foreach \i in {1,2,3,4}
        {
            \pgfmathsetmacro\result{\cy[\i]}
            \coordinate (c\i) at (0, \result);
        }
        \node[tensor2h] (g1) at ($(t1)!0.5!(t2)$) {$B^\dag$};
        \node[tensor2h] (g2) at ($(g1)+(1.5,0)$) {$A^\dag$};
        \node[tensor2h] (g3) at ($(g2)+(1.5,0)$) {$B$};
        \node[tensor2h] (g4) at ($(g3)+(1.5,0)$) {$A$};
        \coordinate (left) at (-1.5, 0);
        \coordinate (right) at (6, 0);
        \foreach \i in {1,2}
        {
            \draw[] (left|-t\i) \foreach \j in {1,2,3,4} { -- (g\j.west|-t\i) (g\j.east|-t\i) } -- (right|-t\i);
        }
        \foreach \i in {1,2,3,4}
        {
            \draw[] (left|-c\i) -- (right|-c\i);
        }
        \node[align=center,rotate=90] (dots1) at ($(g1|-c2)!0.5!(g1|-c1)$) {$\dots$};
        \node[align=center,rotate=90] (dots3) at ($(g3|-c2)!0.5!(g3|-c1)$) {$\dots$};
        \node[align=center,rotate=90] (dots2) at ($(g2|-c4)!0.5!(g2|-c3)$) {$\dots$};
        \node[align=center,rotate=90] (dots4) at ($(g4|-c4)!0.5!(g4|-c3)$) {$\dots$};
        \node at ($(dots1) + (-0.3,0)$) {$\textbf{c}_2$};
        \node at ($(dots2) + (-0.3,0)$) {$\textbf{c}_1$};
        \node at ($(dots3) + (-0.3,0)$) {$\textbf{c}_2$};
        \node at ($(dots4) + (-0.3,0)$) {$\textbf{c}_1$};
        \draw[] (g1|-c2) pic{c01} -- (dots1) -- (g1|-c1) pic{c01} -- (g1.north);
        \draw[] (g3|-c2) pic{c01} -- (dots3) -- (g3|-c1) pic{c01} -- (g3.north);
        \draw[] (g2|-c4) pic{c01} -- (dots2) -- (g2|-c3) pic{c01} -- (g2.north);
        \draw[] (g4|-c4) pic{c01} -- (dots4) -- (g4|-c3) pic{c01} -- (g4.north);

        \draw[decorate, decoration = {brace,mirror,raise=5pt}] ($(left|-c4)+(0,0.1)$) -- node[left,xshift=-5pt]{$\lfloor k/2 \rfloor$} ($(left|-c3)+(0,-0.1)$);
        \draw[decorate, decoration = {brace,mirror,raise=5pt}] ($(left|-c2)+(0,0.1)$) -- node[left,xshift=-5pt]{$\lceil k/2 \rceil$} ($(left|-c1)+(0,-0.1)$);
    \end{tikzpicture}
    \caption{\label{fig:circuit_XYXY} Circuit of a controlled gate with $k$ control qubits (See \cref{tr1}).}
\end{figure}

With this decomposition, we realize $\Lambda^{\textbf{c}}(U)$ with 4 controlled-unitaries, namely
$\Lambda^{\textbf{c}_1}(A),\Lambda^{\textbf{c}_2}(B),\Lambda^{\textbf{c}_1}(A^\dag)$ and $\Lambda^{\textbf{c}_2}(B^\dag)$, each of which has $\lfloor k/2 \rfloor$ or $\lceil k/2 \rceil$ number of control qubits. We can recursively decompose each of the 4 unitaries into gates with less number of control qubits, and finally reach the $k=1$ case. Note that when $k=1$, the single-controlled gate $\Lambda^{\textbf{c}}(U)$ is a {2-level} $\mathcal{SV}_3^{\rm U(1)}$ gate whose implementation has been addressed in \cref{ss:2levelSV3}. This gives the base case of the recursion.

The number of 2-level $\mathcal{SV}_3^{\rm U(1)}$ gates in this recursive construction is given by the following lemma, which
applies an argument previously used in the quantum circuit theory \footnote{See Chapter 8.1.3 of \cite{kitaev2002classical},  or Exercise 4.30 of \cite{NielsenAndChuang} 
}.

\begin{lemma}\label{lem:TnOn2}
For any pair of bit strings $\textbf{b},\textbf{b}'\in\{0,1\}^n$ with equal Hamming weights and Hamming distance 2, any  2-level unitary with determinant 1, which acts non-trivially only in the subspace spanned by $|\textbf{b}\rangle$ and $|\textbf{b}'\rangle$
  can be realized with no more than $9 n^2 / 8$ number of 3-qubit controlled gates in \cref{x-rotation}, which all belong to $\mathcal{SV}_3^{\text{U}(1)}$ and can be realized with the circuit in \cref{Fig2}.
\end{lemma}

\begin{proof}

Let $T(k)$ be the number of 3-qubit controlled gates used to implement a gate in the form of $\Lambda^{\textbf{c}}(U)$ with $|\textbf{c}|=k$ control qubits. 
 Based on the circuit in \cref{fig:circuit_XYXY}, we obtain the recursive relation:
\bes
\begin{align} 
    T(k) &= 2 T \left(\left\lfloor \frac k2 \right\rfloor\right) + 2 T \left(\left\lceil \frac k2 \right\rceil\right), ~ k \geq 2 \\
    T(1) &= 1
\end{align}
\ees
    We will prove by induction that
    \begin{align} \label{eq:Tk_ineq}
        T(k) \leq \left\{ \begin{array}{lr}
            (9k^2-1)/8 &  k \text{ is odd}\\
            (9k^2-4)/8 &  k \text{ is even}
        \end{array} \right.
    \end{align}
    $T(1)=1, T(2) = 4$ satisfy the inequality. For $k \geq 3$, assuming the inequality holds for all smaller $k$,
	\begin{enumerate}
		\item If $k$ is even, $T(k) = 4T(k/2) \leq 4 \times \frac{9(k/2)^2-1}{8} = \frac{9k^2-4}{8}$;
		\item If $k$ is odd, then one of $\lfloor k/2 \rfloor = \frac{k-1}{2}$ and $\lceil k/2 \rceil=\frac{k+1}{2}$ is even, and the other one is odd. In either case,
        $T(k) = 2T\left(\frac{k-1}{2}\right) +  2 T\left(\frac{k+1}{2}\right) \leq 2 \times \left(\frac{9(\frac{k-1}{2})^2}{8} + \frac{9(\frac{k+1}{2})^2}{8} -\frac18 - \frac48 \right) = \frac{9k^2-1}{8}$.
	\end{enumerate}
    Therefore, \cref{eq:Tk_ineq} holds for every $k \geq 1$. 
    
    To implement an $n$-qubit 2-level gate on the subspace spanned by   $|\textbf{b}\rangle$ and $|\textbf{b}'\rangle$, we need to implement a controlled gate with $n-2$ control qubits. The number of 3-qubit gates used is therefore $T(n-2) \le 9n^2/8$.

\end{proof}

\refstepcounter{step}\label{step:equalweight}
\subsection*{Step \arabic{step}: Basis elements with equal Hamming weights }

Next, using the construction developed in \cref{step:dist2}, we show how one can implement general 2-level unitaries that preserve the Hamming weight of states in the computational basis.
The technique used here is a variation  of a similar technique that has been previously used for  general quantum circuits   \cite{kitaev2002classical}.

\begin{lemma}\label{lem10}
 For a system with $n$ qubits, consider a 2-level unitary transformation $V$ 
that acts trivially on the subspace orthogonal to $|\textbf{b}\rangle,|\textbf{b}'\rangle$, where  $\textbf{b}, \textbf{b}'\in \{0,1\}^n$ are a pair of bit strings with equal Hamming weights. Any such unitary can be decomposed as 
 \be
V= K^\dag W K\ ,
\ee
where $W$ is also a 2-level unitary that acts on the subspace spanned by $|\textbf{b}'\rangle$ and $|\textbf{b}''\rangle$, where $\textbf{b}', \textbf{b}''\in \{0,1\}^n$ have equal Hamming weights and have Hamming distance $d(\textbf{b}',\textbf{b}'')=2$, and  unitary $K$ can be realized as a sequence of $d(\textbf{b},\textbf{b}')/2-1$  controlled-$\i\text{SWAP}$ gates.
\end{lemma}

We conclude that 
\begin{corollary}\label{cor56}
Any 2-level unitary with determinant 1,
acting on a pair of basis elements  
$|\textbf{b}\rangle$ and $|\textbf{b}'\rangle$ with equal Hamming weights can be realized with
$d(\textbf{b},\textbf{b}')/2-1$ controlled-$\i\text{SWAP}$,  $d(\textbf{b},\textbf{b}')/2-1$ controlled-$\i\text{SWAP}^\dag$ gates, plus one 2-level gate of the type constructed in \cref{step:dist2}. In total this requires  $\mathcal{O}(n^2+d(\textbf{b},\textbf{b}')) = \mathcal{O}(n^2)$ of gates in any of the gate sets \ref{gates_phaseZ}  and \ref{gates_phaseH} in \cref{Thm1}.%controlled-$\sqrt{\i\text{SWAP}}$ and single-qubit rotations around $z$ (or, $\exp(\i\theta R): \theta\in(-\pi,\pi]$).
\end{corollary}

\begin{proof}(\cref{lem10}) For any pair of bit strings $\textbf{b}$ and $\textbf{b}'$, there exists a sequence of $n$-bit strings as 
\be\label{seq}
\textbf{b}=\textbf{b}_0 \rightarrow \textbf{b}_1 \rightarrow \dots \rightarrow \textbf{b}_t =\textbf{b}'
 \ ,
\ee
such that (i) all bit strings in this sequence have equal Hamming weights, (ii) For any pair of consecutive bit strings $\textbf{b}_j$ and $\textbf{b}_{j+1}$ the Hamming distance is 2, and (iii) The length of this sequence is $t+1$, where $t=d(\textbf{b},\textbf{b}')/2$.

This sequence of bit strings is realized by applying  $t+1$ swaps that can be determined as follows: Suppose we specify each bit by its label (i.e., its location) which is an integer $i\in \{1,\cdots, n\}$, such that $\textbf{b}=b_1\cdots b_i\cdots b_n$. 
 Consider the bits in which  $\textbf{b}$ and $\textbf{b}'$ take different values.  There are $d(\textbf{b},\textbf{b}')$ such bits, where $d(\textbf{b},\textbf{b}')$ is the Hamming distance of $\textbf{b}$ and $\textbf{b}'$, which is an even integer because they have equal Hamming weights.  These bits can be partitioned into two subsets, each having size $t:=d(\textbf{b},\textbf{b}')/2$ with the following property: 
bits in the first subset take value 0 in $\textbf{b}$ and take value 1 in  $\textbf{b}'$. And, the bits in the second subset take value 1 in $\textbf{b}$ and take value 0 in  $\textbf{b}'$. More precisely, let $$l_1,l_2,\cdots  ,l_{t}$$ be $t$ distinct integers in $\{1, \cdots, n\}$ where $b_{l_j}=1$ and $b'_{l_j}=0$ for $j=1,\dots,t$. Similarly, let $$r_1,r_2,\cdots  , r_{t}$$ be $t$ distinct integers in $\{1, \cdots, n\}$
where $b_{r_j}=0$ and $b'_{r_j}=1$ for $j=1,\dots,t$. 
Then, define 
\be
|\textbf{b}_j\rangle=\text{SWAP}_{l_j,r_j} |\textbf{b}_{j-1}\rangle\ ,
\ee
where $|\textbf{b}_0\rangle=|\textbf{b}\rangle$, and $\text{SWAP}_{l_j,r_j}$ is the SWAP unitary acting on qubits $l_j$ and $r_j$. In this way, we obtain  a sequence  of bit strings in the form of \cref{seq} satisfying the desired properties. 

In general, the gate $\text{SWAP}_{l_j,r_j}$ acts non-trivially on state $|\textbf{b}'\rangle$. Now, suppose rather than this gate, we use controlled-$\i\text{SWAP}$,  
 with the control qubit chosen to differentiate $\ket{\textbf{b}'}$ and $\ket{\textbf{b}_{j-1}}$. Namely, we choose 
\bes
\begin{align}
\Lambda_{l_{j+1}}^{1}(\i\text{Sw}_{l_j,r_j})|\textbf{b}_{j-1}\rangle&=\i|\textbf{b}_{j}\rangle\\ \Lambda_{l_{j+1}}^{1}(\i\text{Sw}_{l_j,r_j})|\textbf{b}'\rangle&=|\textbf{b}'\rangle\ ,
\end{align}
\ees
where the unitary in the left-hand side is controlled-$\i\text{SWAP}$, i.e.,
\be
\Lambda_{l_{j+1}}^{1}(\i\text{Sw}_{l_j,r_j})= |0\rangle\langle 0|_{l_{j+1}}\otimes \mathbb{I}_{l_j,r_j}+|1\rangle\langle 1|_{l_{j+1}}\otimes \i\text{Sw}_{l_j,r_j}\ ,
\ee
where $l_{j+1}$ is the control qubit, and $l_j$ and $r_j$ are the target qubits.  
 This gate acts trivially on state $\textbf{b}'$ since by the above definition, $b'_{l_{j+1}}=0$ for every $j=0,\dots,t-1$. 

Then, defining 
\begin{align} \label{eq:K_as_cisw}
K&=\Lambda_{l_t}^{1}(\iswap_{l_{t-1},r_{t-1}}) \cdots \cdots \Lambda_{l_3}^{1}(\iswap_{l_2,r_2}) \Lambda_{l_2}^{1}(\iswap_{l_1,r_1}) \ ,
\end{align}
we find 
\bes \label{eq:KbKb'}
\begin{align}
K|\textbf{b}\rangle&=\i^{t-1} |\textbf{b}_{t-1}\rangle\\ K|\textbf{b}'\rangle&=|\textbf{b}'\rangle\ .
\end{align}
\ees
Therefore, by defining 
\be
W=K V K^\dag \ ,
\ee
we find that $W$ is a 2-level unitary acting on the subspace spanned by $|\textbf{b}''\rangle=|\textbf{b}_{t-1}\rangle$ and $|\textbf{b}'\rangle$, where  $d(\textbf{b}',\textbf{b}'')=2$.
 More explicitly, if $V = U(\textbf{b},\textbf{b}')$ with
\begin{align}
    U = \left(\begin{matrix}a&b\\c&d\end{matrix}\right)\,,
\end{align}
then $W = \widetilde{U}(\textbf{b}'',\textbf{b}')$ with
\begin{align} \label{eq:phase_from_iswaps}
    \widetilde{U} = \left(\begin{matrix}\i^{t-1}&\\&1\end{matrix}\right)U\left(\begin{matrix}\i^{1-t}&\\&1\end{matrix}\right) = \left(\begin{matrix}a&\i^{t-1}b\\\i^{1-t}c&d\end{matrix}\right)\,.
\end{align}

This completes the proof of  \cref{lem10}. 

\end{proof}

\refstepcounter{step}\label{step:semi-universal}
\subsection*{Step \arabic{step}: Semi-universality }

As defined in \cref{sv}, we call a gate set  semi-universal if it generates 
all  unitaries in $\mathcal{SV}^{\text{U}(1)}_{n}$ for all $n$, i.e.,  all energy-conserving unitaries  $V=\bigoplus_{m=0}^n V\sct{m}$ 
 satisfying  the additional constraint $\text{det}(V\sct{m})=1: m=0, \cdots, n$. It can be easily shown that any such unitary can be decomposed into a sequence of 2-level energy-conserving unitaries that satisfy the same constraint. To show this we use the following result. 
\begin{lemma}\cite{reck1994experimental}\label{lem:2-level_decomp}
    Given any basis for $\mathbb{C}^d$,  any unitary in $\text{SU}(d)$ can be decomposed into a product of no more than $d(d-1)/2$ unitaries in $\text{SU}(d)$ that are 2-level with respect to this basis. 
\end{lemma}
This lemma is a slight variation of a similar result in \cite{reck1994experimental}, which does not impose any constraints on the determinant of 2-level unitaries (See also \cite{kitaev2002classical, NielsenAndChuang}). For completeness, we present the proof in \cref{app:2-level-decomposition}.

Recall that in the decomposition $V=\bigoplus_{m=0}^n V\sct{m}$, unitary $V\sct{m}$ acts on the subspace with Hamming weight $m$, which has dimension ${n}\choose{m}$. 
Then, applying this lemma, we find that $V\sct{m}$ can be realized with $$\frac{1}{2}{n\choose m}\times \left[{n\choose m}-1\right]$$ 2-level energy-conserving unitaries, each of which acts on two computational basis vectors in the subspace of Hamming weight $m$, has determinant 1, and can be constructed with the method in \cref{step:equalweight}.

In conclusion, the total number of 2-level gates needed to implement $V$ is
\begin{align}
  \frac{1}{2}  \sum_{m=1}^{n-1} {n\choose m}\left[{n\choose m}-1\right] = \frac{{2n \choose n} - 2^{n}}{2} \approx \frac{4^n}{2\sqrt{\pi n}}\ ,
\end{align}
 where $\approx $ means the ratio of two sides goes to 1, in the limit $n\rightarrow\infty$.

Combining this with \cref{cor56} we 
conclude that 
\begin{proposition}\label{prop13}
Any energy-conserving unitary in $\mathcal{SV}^{\text{U}(1)}_n$ can be realized with 
$\mathcal{O}(4^n n^{3/2})$
gates in 
any of the gate sets \ref{gates_phaseZ} and \ref{gates_phaseH} in \cref{Thm1}, without ancillary qubits. 
\end{proposition}
Here, to count the number of gates, one can use 
 the diagram in \cref{fig:3ton}:   
Recall that 
controlled-$\i\text{SWAP}$ 
can be realized using the circuit in \cref{Fig2}, which requires $\mathcal{O}(1)$ gates in one of the aforementioned gate sets. Therefore, the total required number of elementary gates is
\begin{align} \label{eq:O4nn32} %\label{propSemi}
    \left[\mathcal{O}(n^2)+\mathcal{O}(n)\right]\times \mathcal{O}\left(\frac{4^n}{\sqrt{n}}\right)= \mathcal{O}(4^n n^{3/2}) \ .
\end{align}

Note that, in general, if the desired unitary acts non-trivially only on a subspace spanned by $D \le  2^n$ basis elements, then from \cref{lem:2-level_decomp} it can be decomposed into no more than $D(D-1) = \mathcal{O}(D^2)$ 2-level unitaries. Hence, any such unitary can be implemented with $[\mathcal{O}(n^2)+\mathcal{O}(n)]\times \mathcal{O}(D^2) = \mathcal{O}(n^2D^2)$ elementary gates in \cref{Thm1}.

\refstepcounter{step}\label{step:universal}
\subsection*{Step \arabic{step}: Universality}

Finally, applying  a mechanism developed  in \cite{marvian2023non}, we show how one can achieve universality using a single ancillary qubit. It is worth noting that this technique can be applied to symmetric circuit  with any arbitrary Abelian symmetry (See Lemma 5 of \cite{marvian2023non}).

Let $|\textbf{b}\rangle$ be an arbitrary element of the computational basis $\{|0\rangle,|1\rangle\}^{\otimes n}$, other than $|0\rangle^{\otimes n}$. Then, there is (at least) one qubit with reduced state $|1\rangle$. Let $|\textbf{b}'\rangle$ be the $n$-qubit state obtained from $|\textbf{b}\rangle$ by changing the state of this qubit from $|1\rangle$ to $|0\rangle$. 
Define the Hamiltonian
\be\label{tq}
\widetilde{H}_{\textbf{b}}= |\textbf{b}\rangle\langle \textbf{b}|\otimes |0\rangle\langle 0|_{\text{anc}} -|\textbf{b}'\rangle\langle \textbf{b}'|\otimes |1\rangle\langle 1|_{\text{anc}} \ .
\ee
Note that $(n+1)$-qubit states $|\textbf{b}\rangle|0\rangle_{\text{anc}} $ and $|\textbf{b}'\rangle|1\rangle_{\text{anc}}$  have equal Hamming weights and their Hamming distance is $2$. Furthermore, this Hamiltonian is traceless, which means $\exp(\i \widetilde{H}_{\textbf{b}} \theta)$ has determinant 1. It follows from \cref{step:dist2} that this 2-level energy-conserving unitary can be implemented without ancilla. Implementing this unitary on the system and ancilla we obtain state 
 \be
\exp[\i \widetilde{H}_{\textbf{b}} \theta] (|\psi\rangle \otimes |0\rangle_{\text{anc}})= \big(\exp[\i \theta |\textbf{b}\rangle\langle \textbf{b}| ] |\psi\rangle\big)\otimes  |0\rangle_{\text{anc}}\ ,
\ee
where $|\psi\rangle\in(\mathbb{C}^2)^{\otimes n}$ is the initial state of $n$ qubits in the system. Since the state of ancilla qubit remains unchanged, we can reuse again in a similar fashion. Therefore, in this way we can realize the unitary $\exp(\i |\textbf{b}\rangle\langle \textbf{b}| \theta)$. 
We can apply this procedure to any sector, except the sector with Hamming weight $m=0$, which corresponds to state $|0\rangle^{\otimes n}$.
It follows that using $(n-1)$ 2-level unitaries we can realize any unitary in the form  
\be \label{eq:D_bm}
D = |0\rangle\langle 0|^{\otimes n}+\sum_{m=1}^n \exp(\i \theta_m)  \ket{\textbf{b}_m}\bra{\textbf{b}_m} \ ,
\ee
for any arbitrary $\theta_1,\dots,\theta_n\in (-\pi,\pi]$, 
where $|\textbf{b}_m\rangle$ is an arbitrarily selected basis element in the sector with Hamming weight $m$.

Finally, recall that 
any energy-conserving unitary $V$  can be decomposed as $V=e^{\i\theta_0} D Q$, where $e^{\i\theta_0}=V\sct0=\langle 0|^{\otimes n}V|0\rangle^{\otimes n}$, 
 $D$ is in the form of \cref{eq:D_bm}, with $\theta_m = \arg(\det(V\sct{m})) - \theta_0$ for $1 \leq m \leq n$,  and $Q\in \mathcal{SV}_n^{\rm U(1)}$.  As mentioned above, $D$ can be implemented with an ancillary qubit. Furthermore, in the previous section we saw how  $Q\in \mathcal{SV}_n^{\rm U(1)}$ can be implemented without ancillary qubits.

Note that the unitary generated by Hamiltonian in \cref{tq} is of the type we studied in \cref{step:dist2} of this construction, and therefore can be realized with no more than $9 (n+1)^2/8$ gates ($n+1$ qubits including one ancilla qubit) in 
$\mathcal{SV}_3^{\rm U(1)}$. In total, this step requires $\mathcal{O}(n^2)$ elementary gates.

 In conclusion, any energy-conserving unitary can  be realized, up to a global phase, with one ancillary qubit. Furthermore,  using the diagram in  \cref{fig:3ton} and \cref{eq:O4nn32}, we  find  that  this construction
requires
\begin{align}
    \mathcal{O}(4^n n^{3/2}) + \mathcal{O}(n^2)(n-1) = \mathcal{O}(4^n n^{3/2}) \ ,
\end{align}
gates in one of these elementary gate sets, 
which completes the proof of \cref{Thm1} for  sets \ref{gates_phaseZ} and \ref{gates_phaseH}.

 % (Recall that the gate controlled-$\i\text{SWAP}$  can be implemented with a constant number of gates in any of the universal  universal set. In Section \ref{sec:3qubit}, we show that any 2-level gate in $\mathcal{SV}_3^{\rm U(1)}$ can be implemented with a constant number of gates in the universal set.

\section{Realizing all energy-conserving unitaries with XY interaction alone} \label{sec:XY_only}

So far in this paper, we have considered circuits that contain both XY interaction as well as single-qubit rotations around the z axis. In this section, we focus on synthesizing quantum circuits that only contain XY interaction.  In this case, the overall realized unitary $V$ on the system respects the $\mathbb{Z}_2$ symmetry corresponding to flipping all the qubits in the system. That is, 
\be\label{Z2}
X^{\otimes n} V X^{\otimes n}= V \ ,
\ee
which follows from the fact that unitary 
$\exp(\i\theta R_{ij})$ satisfies this symmetry for all qubit  pairs 
$i$ and $j$ and all $\theta\in(-\pi,\pi]$.  Interestingly,  \cref{Thm2}, which is proven below, implies that  this $\mathbb{Z}_2$ symmetry together with the U(1) symmetry corresponding to energy conservation, namely
\be\label{eq:XY_U1}
\big[V , \sum_{j=1}^n Z_j \big]=0 \,,
\ee
characterize the set of realizable unitaries, up to additional constraints on the overall phases in each invariant subspace: First, for any realizable unitary $V$, 
\be\label{cons-det}
\text{det}(V\sct{m})=1\ \ \ \ : m=0,\cdots, n \ ,
\ee
where $V\sct{m}$ is the component of $V=\bigoplus_{m=0}^n V\sct{m}$ in the sector with Hamming weight $m$. Second, in the case of even $n$,  
\be\label{cons-det2}
\text{det}(V\sct{n/2,\pm})=1\ ,
\ee
where, as defined in \cref{Thm2},  
$V\sct{n/2,\pm}$ is the component of $V\sct{n/2}$ in the eigensubspace of $X^{\otimes n}$ with eigenvalue $\pm 1$. 
In the following,  $\mathcal{G}_n$ denotes the group of  $n$-qubit unitaries satisfying these 4 conditions, i.e., \cref{Z2,eq:XY_U1,cons-det,cons-det2}.

The necessity of 
 conditions in \cref{cons-det,cons-det2} are discussed in \cite{marvian2022restrictions, marvian2023non} and for  completeness
is also explained  in \cref{App:proof} and \cref{half}. (Briefly, they follow from the fact that $\Tr(\Pi\sct{m} R_{ij})=0$ for all $m=0,\cdots, n$ where $\Pi\sct{m}$ is the projector to the sector with Hamming weight $m$. Furthermore, for $n\geq3$, it holds that $\Tr(X^{\otimes n} R_{ij})=0$.)

In the following, we present explicit circuit synthesis methods for implementing general unitary $V$ satisfying the above 4 conditions using XY interaction. 
This, in particular, completes the proofs of \cref{Thm2} and the main part of \cref{Thm1} for the case of gate set \ref{gates_XY_only}.  %But, first, we 
%In \cref{ss:Z2_symmetry}, we first explain how the $\mathbb{Z}_2$ constraint in \cref{Z2} can be circumvented with an ancilla qubit, and how the determinant constraint \cref{cons-det} can be circumvented with another ancilla qubit. 

\subsection{Overview of the synthesis method: 4-level unitaries}

Recall that the method we used in \cref{sec:nqubit} decomposes a general energy-conserving unitary to a sequence of 2-level energy-conserving unitaries, which in general do not respect the $\mathbb{Z}_2$ symmetry in \cref{Z2} and hence cannot be realized with XY interaction alone. Hence,  we consider a natural extension of 2-level unitaries that satisfy the $\mathbb{Z}_2$ symmetry, namely 4-level unitaries. In the following,  for any pair of distinct bit strings $\textbf{b}, \textbf{b}'\in\{0,1\}^n$,
define the Pauli $X$ and $Y$ operators in the subspace spanned by $\ket{\textbf{b}}$ and $\ket{\textbf{b}'}$ as
\bes
\begin{align}
{X}(\textbf{b}, \textbf{b}')&:=|\textbf{b}\rangle\langle\textbf{b}'|+|\textbf{b}'\rangle\langle\textbf{b}|\ ,\\
{Y}(\textbf{b}, \textbf{b}')&:=\i (|\textbf{b}'\rangle\langle\textbf{b}|-|\textbf{b}\rangle\langle\textbf{b}'|)\ . 
\end{align}
\ees
Let $\overline{\textbf{b}}$ be the  bitwise negation of bit $\textbf{b}$, which means 
\be
|\overline{\textbf{b}}\rangle=X^{\otimes n}|{\textbf{b}}\rangle\ .
\ee
Then, for all $\textbf{b}, \textbf{b}'\in\{0,1\}^n$ with  $\textbf{b}' \neq \textbf{b}, \overline{\textbf{b}}$ 
 and $\theta\in(-\pi,\pi]$,  consider the following 4-level unitaries 
\bes\label{AB}
\begin{align}
\mathbb{R}_x(\theta, \textbf{b}, \textbf{b}')&:=\exp[\i\theta {X}(\textbf{b}, \textbf{b}')] \exp[\i\theta {X}(\overline{\textbf{b}}, \overline{\textbf{b}'})] \label{eq:AB_A}\\ \mathbb{R}_y(\theta, \textbf{b}, \textbf{b}')&:=\exp[\i\theta {Y}(\textbf{b}, \textbf{b}')] \exp[\i\theta {Y}(\overline{\textbf{b}}, \overline{\textbf{b}'})] \  .
\end{align}
\ees
 They clearly respect the $\mathbb{Z}_2$ symmetry in \cref{Z2}, and if  $\textbf{b}$ and $\textbf{b}'$ have equal Hamming weights, i.e.,
\be
w(\textbf{b})=w(\textbf{b}')\ ,
\ee
then they are also energy-conserving.  

Note that we have excluded the case of $\textbf{b}'=\overline{\textbf{b}}$. In general,  $w(\overline{\textbf{b}})=n-w(\textbf{b})$, which means unitaries in \cref{AB}  corresponding to the case of $\textbf{b}'=\overline{\textbf{b}}$ are not energy-conserving, unless $n$ is even and $w(\overline{\textbf{b}})=w(\textbf{b})=n/2$.
We call the special case of Hamming weight $m=n/2$  the ``half-filled sector" and study it separately in \cref{ss:half-filled}. In this case, indeed there exists a family of 2-level energy-conserving  unitaries that respect the $\mathbb{Z}_2$ symmetry, 
as well as the condition $\text{det}(V\sct{n/2})=1$, namely 
\be \label{eq:X_b_bbar}
\exp[\i\theta {X}(\textbf{b}, \overline{\textbf{b}})]=\mathbb{R}_x(\frac{\theta}{2},\textbf{b},\overline{\textbf{b}})\ ,
\ee
which corresponds to the special case of $\textbf{b}'=\overline{\textbf{b}}$ in \cref{eq:AB_A}. Is this family of unitaries realizable using XY interaction alone? 

Interestingly, the answer is no! While the unitary in \cref{eq:X_b_bbar} respects the 3 conditions in \cref{Z2,eq:XY_U1,cons-det}, since $\text{det}(V\sct{n/2,\pm})=e^{ \pm\i \theta}$, unless $\theta=0$ it does not  respect the condition in \cref{cons-det2}. (Note that this  is indeed the generalization of the example we discussed in the introduction, below \cref{Thm2}.) It is also worth noting that, unitaries in \cref{eq:X_b_bbar} are  the only 2-level unitaries that satisfy the 3 conditions in \cref{Z2,eq:XY_U1,cons-det}. In conclusion, in the following we always restrict our attention to 4-level unitaries in \cref{AB}, i.e., we impose the condition $\textbf{b}' \neq  \overline{\textbf{b}}$.

\begin{figure*}
    \centering
    \begin{tikzpicture}
        \node[draw=black,align=center] (CR){$\mathbb{R}_x$ (\ref{AB})};
        \node[draw=black,align=center,right=0.7 of CR] (SV3) {4-qubit 4-level\\$V\sct1(\textbf{b},\textbf{b}')V\sct1(\overline{\textbf{b}},\overline{\textbf{b}'})$ \\ $w(\textbf{b})=w(\textbf{b}')$\\$\textbf{b}'\neq\textbf{b},\overline{\textbf{b}}$\\(\cref{ss:XX_YY_RX_RY})};
        \node[draw=black,minimum width=3cm] (basic) at ($(CR.north)!0.5!(SV3.north) + (0,1.2)$) {XY interaction $\exp(\i\theta R)$};
        \node[below=0.75cm of SV3, draw=black,align=center](dany) {$n$-qubit 4-level \\ $w(\textbf{b})=w(\textbf{b}')$ \\ $\textbf{b}'\neq\textbf{b},\overline{\textbf{b}}$ \\ (\cref{ss:XXYY_distany})};
        \node[right=1.8cm of SV3, draw=black, align=center] (halffilled) {2-level gates in \\half-filled subspace \\ (\cref{ss:half-filled})};
        \node[draw=black,align=center] (semisemi) at (dany-|halffilled) {Full set of unitaries $\mathcal{G}_n$\\realizable with $\exp(\i\theta R)$\\  (\cref{ss:XY_semi})};
        \node[right=0.75cm of semisemi, draw=black, align=center] (semi)  {Breaking $\mathbb{Z}_2$ symmetry\\$\mathcal{SV}_{n}^{\rm U(1)}$ (\cref{ss:Z2_symmetry})};
        \node[right=0.75cm of semi, draw=black,align=center] (univ) {Universality\\$\mathcal{V}_{n}^{\rm U(1)}$(\cref{ss:XY_univ})};
        \node[above=1cm of semi, draw=black] (anc0)  {Ancilla qubit};
        \node[above=1cm of univ, draw=black] (anc) {Ancilla qubit};
        \draw[->] (CR) -- node[above](times3){$\times3$}(SV3);
        \draw[->] (basic.south-|CR) -- node[right]{$\times \mathcal{O}(1)$} (CR);
        \draw[->] (basic.south-|SV3) -- node[right]{$\times \mathcal{O}(1)$} (SV3);
        \draw[->] (SV3) -- node[right]{$\times \mathcal{O}(n^2)$} (dany);
        \draw[->] ($(dany.east)+(0,0.4)$) -- +(0.7,0) |- node[above,pos=0.7] {$\times2$} (halffilled);
        \draw[->] (dany) -- node[below] {$\times\mathcal{O}\left(\frac{4^n}{\sqrt{n}}\right)$} (semisemi);
        \draw[->] (semisemi) -- node[above]{$\times1$} (semi);
        \draw[->] (halffilled) -- node[right]{$\times\mathcal{O}\left(\frac{4^n}{n}\right)$} (semisemi);
        \draw[->] (anc0) -- node[right]{$\times1$} (semi);
        %\draw[->] (semi) -- node[above](times1){$\times 1$} (univ);
        \draw[->] (semi) -- node[above]{$\times1$} (univ);
        %\draw[->] (d2) -| node[above, pos=0.25]{$\times (n-1)$} (univ);
        \draw[->] (anc) -- node[right](anc_univ){$\times1$} (univ);
    \end{tikzpicture}
    \caption{Overview of the method for 
 synthesizing general energy-conserving unitaries with XY interaction alone.  
 The {number on the arrow} pointing from one box $A$ to another box $B$ indicates %(upper bounds of) 
 the number of uses of $A$ in the construction of $B$. }
    \label{fig:XY_only_4ton}
\end{figure*}
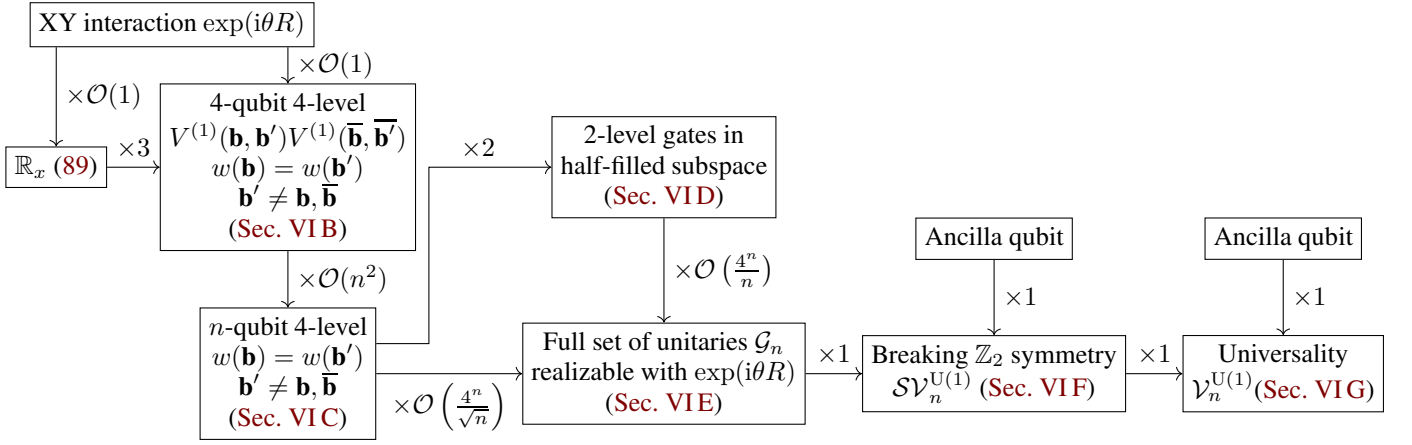

%Now we outline the strategy for implementing general energy-conserving $\mathbb{Z}_2$-invariant unitaries. 

In the following, we first show how the unitaries in \cref{AB} can be realized with XY interaction alone, and then use them to construct all unitaries respecting the 4 conditions in \cref{Thm2}. In particular, in \cref{ss:XX_YY_RX_RY} we show that the circuits for 4-level unitaries in \cref{AB} can be obtained based on the circuits for 2-level unitaries $\exp[\i\theta {X}(\textbf{b}, \textbf{b}')]$ and $\exp[\i\theta {Y}(\textbf{b}, \textbf{b}')]$, which were found in \cref{sec:3qubit}.

The restriction to XY interactions alone makes the constructions in this section more complicated and slightly different from the constructions in \cref{sec:nqubit}.  
However, the presence of  $\mathbb{Z}_2$ symmetry also leads to a simplification, which can be understood in terms of the following lemma.

\begin{lemma} \label{lem:con-ham}
Suppose  a pair of  $n$-qubit energy-conserving unitaries $V=\bigoplus_{m=0}^n V\sct{m}$ and $W=\bigoplus_{m=0}^n W\sct{m}$ both respect the $\mathbb{Z}_2$ symmetry in \cref{Z2}. Then $V=W$ if, and only if 
\be\label{eq:con-ham}
V\sct{m}=W\sct{m}\quad\ \  : m=0,\cdots, \lfloor \frac{n}{2} \rfloor\ .
\ee
\end{lemma}
\begin{proof}
    If
$
W=X^{\otimes n}WX^{\otimes n}$ and   $
V=X^{\otimes n}VX^{\otimes n}$, then for $m > \lfloor n/2 \rfloor$ it holds that
\bes
\begin{align}
\Pi\sct{m} W \Pi\sct{m}&= \Pi\sct{m} X^{\otimes n} W X^{\otimes n}\Pi\sct{m} \\ &=X^{\otimes n} \Pi\sct{n-m} W \Pi\sct{n-m} X^{\otimes n}\ \\ &=X^{\otimes n} \Pi\sct{n-m} V \Pi\sct{n-m}X^{\otimes n}\\ &= \Pi\sct{m} V \Pi\sct{m} \ ,
\end{align}
\ees
where the third line follows from \cref{eq:con-ham} since $n-m\leq\lfloor n/2 \rfloor$, and  we have used the fact that $\Pi\sct{m}X^{\otimes n}=X^{\otimes n}\Pi\sct{n-m}$. Therefore, \cref{eq:con-ham} implies $V=W$. The other direction is trivial.

\end{proof}

As we further discuss in \cref{ss:half-filled}, when $n$ is even, 
the sector with Hamming weight $m=n/2$ requires special treatment. But, for the rest of the Hilbert space, the above lemma implies a useful simplification. Namely, we can restrict our attention to the subspace 
\begin{align} \label{eq:subspc_nover2}
    \mathcal{H}\sct{ <  n/2 }:=\bigoplus_{m=0}^{ \lfloor (n-1)/2 \rfloor} \mathcal{H}\sct{m} \,.
\end{align} 
Inside this subspace, the unitaries in \cref{AB} are 2-level and energy-conserving. 
This fact allows us to use the strategies developed in \cref{sec:nqubit} based on 2-level unitaries.  \cref{fig:XY_only_4ton} presents an overview of the workflow in this section.

\subsection{A useful family of 4-qubit 4-level energy-conserving  $\mathbb{Z}_2$-invariant gates} \label{ss:XX_YY_RX_RY}

In this section, we construct a useful family of 4-level unitaries that can be realized using XY interaction alone, namely the following sequences of 2 controlled unitaries  
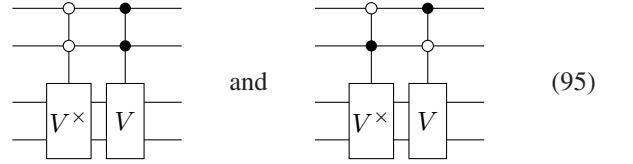
\begin{align}\label{Eq2024}
\begin{minipage}{10em}
\centering
\begin{tikzpicture}[scale=1]
    \useasboundingbox (-0.75, -1) rectangle (1.5,1.5);
    \coordinate (left) at (-0.75, 0);
    \coordinate (right) at (1.5,0);
    \coordinate (l1) at (0,-0.5);
    \coordinate (l2) at (0,0);
    \coordinate (l3) at (0,0.75);
    \coordinate (l4) at (0,1.25);
    \foreach \i in {l1,l2,l3}
    {
        \draw (left|-\i) -- (right|-\i);
    }
    \draw(left|-l4) -- (right|-l4);
    \node[tensor2h] (V01) at ($(l1)!0.5!(l2)$) {$V^\times$};
    \draw (V01) -- (V01|-l3) pic {c0} -- (V01|-l4) pic{c0};
    \node[tensor2h] (V02) at ($(V01)+(0.75,0)$) {$V$};
    \draw (V02) -- (V02|-l3) pic {c1} -- (V02|-l4) pic{c1};
    %\node[align=center, left=0cm of left|-l4] {Extra\\control};
\end{tikzpicture}
\end{minipage} \text{and}
\begin{minipage}{10em}
\centering
\begin{tikzpicture}[scale=1]
    \useasboundingbox (-0.75, -1) rectangle (1.5,1.5);
    \coordinate (left) at (-0.75, 0);
    \coordinate (right) at (1.5,0);
    \coordinate (l1) at (0,-0.5);
    \coordinate (l2) at (0,0);
    \coordinate (l3) at (0,0.75);
    \coordinate (l4) at (0,1.25);
    \foreach \i in {l1,l2,l3}
    {
        \draw (left|-\i) -- (right|-\i);
    }
    \draw(left|-l4) -- (right|-l4);
    \node[tensor2h] (V01) at ($(l1)!0.5!(l2)$) {$V^\times$};
    \draw (V01) -- (V01|-l3) pic {c1} -- (V01|-l4) pic{c0};
    \node[tensor2h] (V02) at ($(V01)+(0.75,0)$) {$V$};
    \draw (V02) -- (V02|-l3) pic {c0} -- (V02|-l4) pic{c1};
    %\node[align=center, left=0cm of left|-l4] {Extra\\control};
\end{tikzpicture}
\end{minipage} 
\end{align}
where
\begin{align}
    V = \left(
\begin{array}{ccc}
1  &   &   \\
  & V\sct1  &   \\
  &   &   1
\end{array}\right), \ \ V^\times = (X\otimes X)V(X\otimes X)\ ,
\end{align}
where $V^{(1)}$ is in SU(2) and acts in the subspace spanned by $|01\rangle$ and $|10\rangle$. The particular pairing of these two controlled unitaries guarantees that 
this composition respects  the $\mathbb{Z}_2$ symmetry of XY interaction.  
These unitaries then will be used as the building blocks for constructing general unitaries in $\mathcal{G}_n$.

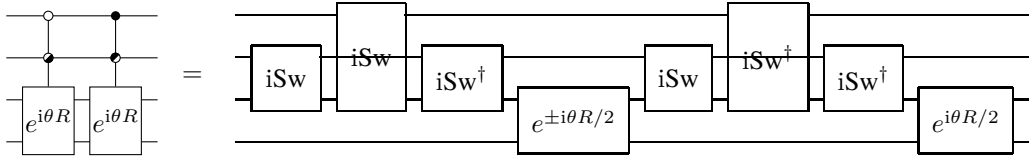
\begin{figure*}

\begin{minipage}{7em}
\begin{tikzpicture}[scale=0.745]
    \coordinate (left) at (-0.75, 0);
    \coordinate (right) at (1.95,0);
    \useasboundingbox (-0.75,-1.2) rectangle (1.75,2.1);
    \coordinate (l1) at (0,-0.75);
    \coordinate (l2) at (0,0);
    \coordinate (l3) at (0,0.75);
    \coordinate (l4) at (0,1.5);
    \foreach \i in {l1,l2,l3,l4}
    {
        \draw (left|-\i) -- (right|-\i);
    }
    \node[tensor2h,minimum height=0.6*15mm] (V01) at ($(l1)!0.5!(l2)$) {$e^{\i\theta R}$};
    \draw (V01) -- (V01|-l3) pic[scale=0.82] {c01} -- (V01|-l4) pic[scale=0.82]{c0};
    \node[tensor2h,minimum height=0.6*15mm] (V02) at ($(V01)+(1.2,0)$) {$e^{\i\theta R}$};
    \draw (V02) -- (V02|-l3) pic[scale=0.82] {c10} -- (V02|-l4) pic[scale=0.82]{c1};
\end{tikzpicture}
\end{minipage}
\ $=$ \quad
\begin{minipage}{1em}
\vspace{0.3em}

  \Qcircuit @C=0.6em @R=.7em {
 & \qw                         &   \multigate{2}{\i\text{Sw}}&  \qw                              & \qw                                 & \qw                             &  \multigate{2}{\i\text{Sw}^\dag}& \qw                             & \qw                              & \qw\\
 &  \multigate{1}{\i\text{Sw}} &  \qw                        &   \multigate{1}{\i\text{Sw}^\dag} & \qw                                 & \multigate{1}{\i\text{Sw}}      & \qw                             & \multigate{1}{\i\text{Sw}^\dag} & \qw                              & \qw  \\
 & \ghost{i\text{Sw}}          &  \ghost{i\text{Sw}}         &  \ghost{i\text{Sw}^\dag}          & \multigate{1}{e^{\pm\i \theta R/2}} & \ghost{\i\text{Sw}}             & \ghost{i\text{Sw}^\dag}         & \ghost{\i\text{Sw}^\dag}        & \multigate{1}{e^{\i \theta R/2}} & \qw    \\ 
 & \qw                         &  \qw                        &  \qw                              & \ghost{e^{\pm\i \theta R/2}}        & \qw                             & \qw                             & \qw                             & \ghost{e^{\i \theta R/2}}        & \qw 
 }
\end{minipage}
 \vspace{5mm}
  \caption{ The circuits for implementing
  $\mathbb{R}_x(\theta, 0101, 0110)$ and  $\mathbb{R}_x(\theta, 0001, 0010)$, based on \cref{teta,teta2}. The unitary $\mathbb{R}_x(\theta, 0101, 0110)$ applies  $e^{ \i\theta R}$ on the third and fourth qubits, when the parity of the first two qubits is odd, whereas $\mathbb{R}_x(\theta, 0001, 0010)$ 
  applies $e^{\i\theta R}$ when the parity is even. In the circuit in the right-hand side, these unitaries correspond to choosing the minus sign (odd parity) and plus sign (even parity) in the $e^{\pm \i\theta R/2}$ gate, respectively.
  }\label{fig:AB}
\end{figure*}

%satisfying \cref{Z2,eq:XY_U1,cons-det,cons-det2}.  

Recall that in the circuit synthesis method developed in \cref{sec:nqubit},  the main building block was  the
 3-qubit 2-level energy-conserving unitaries in \cref{ew3} (See \cref{Fig2}).
However, this gate does not respect the $\mathbb{Z}_2$ symmetry in \cref{Z2}, which explains why the single-qubit gates $S$ and $S^\dag$ appear in this circuit  and they cannot be avoided.  
To construct a gate that respects this symmetry, one may consider a modification of this circuit, obtained by removing $S$ and $S^\dag$ gates from \cref{ew3}. 
Then, we obtain the 3-qubit gate
\begin{align}
&F^\dag_{123} \exp(\i \frac{\theta}{2} R_{23}) F_{123}\\ &=|0\rangle\langle 0|_1\otimes \exp(\i \frac{\theta}{2} L_{23})+|1\rangle\langle 1|_1\otimes\exp(-\i \frac{\theta}{2} L_{23})\ ,  \nonumber
\end{align}
acting  on qubits $1,2,3$, where  
 $L_{jk}:=(Y_j X_k-X_jY_k)/{2}=\i (|10\rangle\langle 01|_{jk}-|01\rangle\langle 10|_{jk})$  
and $F_{ijk}:=\iswap_{ik} \iswap^\dag_{jk}\iswap_{ij}$ as defined in \cref{Def-F}. 
While this is a useful unitary, since it always acts non-trivially on qubits $2$ and $3$, regardless of the state of qubit 1, it can not be easily used in composition with other unitaries.

To overcome this problem, we construct the 4-qubit unitary 

\begin{align}
&\mathbb{R}_x(\theta, 0101, 0110) \\ &:= \exp[\i\theta (\ketbra{0101}{0110}+\ketbra{0110}{0101})] \nonumber\\
 &\hspace{24mm} \times\exp[\i\theta (\ketbra{1010}{1001}+\ketbra{1001}{1010})] \  ,\nonumber\\ 
&=(|00\rangle\langle 00|_{12}+|11\rangle\langle 11|_{12})\otimes \mathbb{I}_{34}\nonumber \\ &\hspace{24mm}   +(|01\rangle\langle 01|_{12}+|10\rangle\langle 10|_{12})\otimes \exp(\i\theta R_{34})\nonumber
\\ &= \sum_{b_1, b_2=0}^1 |b_1b_2\rangle\langle b_1b_2|_{12}\otimes  \exp(\i[b_1\oplus b_2]\theta R_{34})\nonumber\\
&=\exp(\i\frac{\theta}{2}  R_{34})\  \CZ_{23}\  \CZ_{13}\ \exp(-\i\frac{\theta}{2} R_{34})\ \CZ_{13}\  \CZ_{23} \label{teta0} \  
\end{align}
which is obtained from \cref{AB} by choosing $\textbf{b}=0101$ and $\textbf{b}'=0110$, and  $b_1\oplus b_2 \in\{0,1\}$ denotes the parity of the bits $b_1$ and $b_2$.   This unitary applies $\exp(\i\theta R_{34})$ on qubits 3 and 4, when the parity of qubits 1 and 2 is odd. Note that, using \cref{art}, this unitary can be rewritten as 
\begin{align}\label{teta}
&\mathbb{R}_x(\theta, 0101, 0110)=\exp(\i\frac{\theta}{2}  R_{34})\  G^\dag_{123}\ \exp(-\i\frac{\theta}{2} R_{34})\   G_{123}\ , \end{align}
where 
\bes
\begin{align}\label{def-G}
G_{123}&:= \iswap_{23}^\dag \iswap_{13} \iswap_{23}\\  &= Z_1 \text{Sw}_{12}\ \CZ_{12}\   \CZ_{13}\  \CZ_{23}\ ,
\end{align}
\ees
and in \cref{teta} 
its effect is equivalent to $\CZ_{13}\CZ_{23}$ (See \cref{tab:identities} for a diagram of the circuit identity in \cref{def-G}). 
Furthermore, by applying $\exp(\i\frac{\theta}{2} R_{34})$ instead of $\exp(-\i\frac{\theta}{2} R_{34})$ one obtains the gate 
\begin{align}\label{teta2}
&\mathbb{R}_x(\theta, 0001, 0010)=\exp(\i\frac{\theta}{2}  R_{34})\  G^\dag_{123}\ \exp(\i\frac{\theta}{2} R_{34})\   G_{123}\  .
\end{align}
The circuits for unitaries $\mathbb{R}_x(\theta, 0101, 0110)$ and $\mathbb{R}_x(\theta, 0001, 0010)$  are presented in \cref{fig:AB}.

Next, we obtain  $\mathbb{R}_y$ gates defined in \cref{AB} from $\mathbb{R}_x$ gates. Note that $\mathbb{R}_y(\theta, 0001, 0010)$ can be obtained from $\mathbb{R}_x(\theta, 0001, 0010)$ by sandwiching the third (or the fourth) qubit between $S$ and $S^\dag$. However, our goal in this section is to avoid single-qubit rotations around the z-axis and realize everything with XY interaction alone. 

Hence, to overcome this challenge  we use a different approach. Namely, we utilize a 3-qubit sequence  introduced by Lidar and Wu  in \cite{lidar2001reducing, wu2002power} to establish the possibility of universal quantum computing with XY interaction in a DFS. (It is worth noting that Kempe and Whaley adapt this sequence in \cite{kempe2002exact} 
to obtain the exact gate sequences for universal quantum computation using the XY interaction alone.) The sequence is given by 
\begin{align}
P_{123}(\phi)&= 
\sqrt{\i\text{Sw}}_{23}\  \i\text{Sw}_{12}\  \exp(\i \phi R_{13})\  \i\text{Sw}^\dag_{12} \ \sqrt{\i\text{Sw}}_{23}^\dag\nonumber \ .
\end{align}
This gate is diagonal in the computational basis. Indeed, one can show that 
\begin{align}
P_{123}(\phi)
&=|0\rangle\langle 0|_1 \otimes \exp(-\i \frac{\phi}{2}Z_2)\otimes \exp(\i \frac{\phi}{2}Z_3) \nonumber \\ &+|1\rangle\langle 1|_1\otimes \exp(\i \frac{\phi}{2}Z_2)\otimes  \exp(-\i \frac{\phi}{2}Z_3) \ ,
\end{align}
for arbitrary $\phi\in(-\pi,\pi]$. 
Then,  by sandwiching the unitary $\mathbb{R}_x(\theta, 0001, 0010)$ in \cref{teta} between  
\bes 
\begin{align}
P_{123}(\frac{\pi}{2}) &= \siswap_{23} \iswap_{12}\iswap_{13}\iswap_{12}^\dag \siswap^\dag_{23}\\
&=|0\rangle\langle 0|_1 \otimes S_2 \otimes S_3^\dag +|1\rangle\langle 1|_1\otimes S_2^\dag \otimes S_3 \\ 
&=  \CZ_{12} \CZ_{13}  S_2 S_3^\dag\label{PPP} \ ,
\end{align}
\ees
and its inverse we show that 
\begin{align}\label{tetay}
P_{123}(-\frac{\pi}{2})\ &\mathbb{R}_x(\theta, 0101, 0110) P_{123}(\frac{\pi}{2})=\mathbb{R}_y(\theta,0101,0110)\ .
\end{align}
To see this, first recall that $\mathbb{R}_x(\theta, 0101, 0110)$ is defined as $$\exp[\i \theta {X}(0101,0110)]\exp[\i \theta {X}(1010,1001)]\ .$$
Similarly, $\mathbb{R}_y(\theta, 0101, 0110)$ is defined as $$\exp[\i \theta {Y}(0101,0110)]\exp[\i \theta {Y}(1010,1001)]\ .$$  

Then, we note that 
\begin{align} 
&P_{123}(-\frac{\pi}{2}) \exp[\i \theta {X}(0101,0110)] P_{123}(\frac{\pi}{2}).\nonumber\\ &=\CZ_{13} S_3 \exp[\i \theta {X}(0101,0110)] S_3^\dag \CZ_{13} \nonumber
\\ &=\exp[\i \theta {Y}(0101,0110)] \ , \label{PPP2} %repeated label PPP, changed to PPP2
\end{align}
where we have applied \cref{PPP}. 
Furthermore, by sandwiching both sides of this equation between $X^{\otimes 4}$ and using the fact that $P_{123}(\phi)$ commutes with $X^{\otimes 4}$, we obtain
\begin{align}
&P_{123}(-\frac{\pi}{2}) \exp[\i \theta {X}(1010,1001)] P_{123}(\frac{\pi}{2})\nonumber \\ &=\exp[\i \theta {Y}(1010,1001)]\ . % removed \nonumber
\end{align}
Then, multiplying the above two equations, we arrive at \cref{tetay}.

In summary, in this way we can obtain 4 families of unitaries
\bes\label{eq:RxRy1234}
\begin{align} 
    &\mathbb{R}_x(\theta, 0101, 0110)\ , \label{eq:Rx1}\\
    &\mathbb{R}_x(\theta, 0001, 0010) = \mathbb{R}_x(\theta, 1110, 1101)\ , \label{eq:Rx2}\\
    &\mathbb{R}_y(\theta,0101,0110)\ ,\label{eq:Ry1} \\ &\mathbb{R}_y(\theta,0001,0010)=\mathbb{R}_y(\theta,1110,1101)\ ,\label{eq:Ry2}
\end{align}
\ees
and their  permuted versions,  which can all be realized with $\mathcal{O}(1)$ number of $\exp(\i\alpha R)$ gates  alone.

As a useful observation, any 4-bit strings $\textbf{b}$ and $\textbf{b}'$ satisfying $w(\textbf{b})=w(\textbf{b}')$ and $\textbf{b}'\neq \textbf{b},\overline{\textbf{b}}$ necessarily have two common bits and two distinct bits, and fit into either \cref{eq:Rx1,eq:Ry1} or \cref{eq:Rx2,eq:Ry2} by permuting the qubits. 
This means that we can implement $\mathbb{R}_x(\theta,\textbf{b}, \textbf{b}')$ and $\mathbb{R}_y(\theta,\textbf{b}, \textbf{b}')$ for all $\textbf{b}, \textbf{b}'\in\{0,1\}^4$ which satisfy the  conditions  $w(\textbf{b})=w(\textbf{b}')$ and $\textbf{b}'\neq \textbf{b},  \overline{\textbf{b}}$.

Furthermore, combining these unitaries we can obtain unitaries that act as arbitrary $V\sct1 \in \text{SU}(2)$ on the subspace spaned by $\{|\textbf{b}\rangle, |\textbf{b}'\rangle\}$. In particular, suppose $V\sct1$ has the  Euler decomposition
\be
V\sct1 = e^{\i\gamma X}e^{\i\beta Y}e^{\i\alpha X}\ .
\ee
Then, for $\textbf{b}'\neq\overline{\textbf{b}}$  the following sequence of unitaries
\begin{align}
    \mathbb{R}_x(\gamma,\textbf{b}, \textbf{b}')\mathbb{R}_y(\beta,\textbf{b}, \textbf{b}')\mathbb{R}_x(\alpha,\textbf{b}, \textbf{b}') = V\sct1(\textbf{b},\textbf{b}')V\sct1(\overline{\textbf{b}},\overline{\textbf{b}'})\ , \label{eq:V1_4qubit}
\end{align}
realize two copies of $V\sct1$, namely $V\sct1(\textbf{b},\textbf{b}')$ which acts in the subspace spanned by $|\textbf{b}\rangle$ and  $|\textbf{b}'\rangle$, and 
$V\sct1(\overline{\textbf{b}},\overline{\textbf{b}'})$ which acts in the subspace spanned by $|\overline{\textbf{b}}\rangle$ and  $|\overline{\textbf{b}'}\rangle$.

Choosing ${\textbf{b}}$ and  $\textbf{b}'$ to be bit strings in $\{0,1\}^4$ with Hamming weight 1, we can obtain 4-level unitaries that act as 2-level ones in the sectors with Hamming weights 1 and 3, either of which has dimension 4.  Combining such 2-level unitaries, we can realize an arbitrary unitary in SU(4) in the sector with Hamming weight 1. Then, the realized unitary in the sector with Hamming weight 3 is dictated by the $\mathbb{Z}_2$ symmetry. In \cref{ss:half-filled}, where we focus on the half-filled sector $m=n/2$, we come back to this example and explain how, in the sector with  Hamming weight 2, a general unitary satisfying the condition in \cref{Z2,eq:XY_U1,cons-det,cons-det2} can be realized.

We finish this section by rewriting \cref{eq:V1_4qubit} in terms of controlled unitaries in the form of \cref{Eq2024}, which will be useful for applications in the next section. Recall that here we consider $\textbf{b},\textbf{b}'\in\{0,1\}^{4}$ satisfying the constraints  $w(\textbf{b})=w(\textbf{b}')$, and $\textbf{b}'\neq \textbf{b}, \overline{\textbf{b}}$. As mentioned above, any such $\textbf{b}$ and $\textbf{b}'$ have two common bits and two distinct bits, which means up to a permutation, they can be written as 
$$\textbf{b}=\textbf{c}01\ \ \  \text{and}\ \  \ \  \textbf{b}'=\textbf{c}10\ ,$$
where $\textbf{c}$ contains the common bits. 
Then, we can interpret the unitary in \cref{eq:V1_4qubit} as a controlled unitary that acts only when the first two qubits are in $\textbf{c}$ or $\bar{\textbf{c}}$. To see this, recall that, 
 as we have seen before in   \cref{eq:2-level_controlled}, the 2-level unitary $V\sct1(\textbf{b}, \textbf{b}')$ can be written as the controlled gate
\begin{align}
    V\sct1(\textbf{b}, \textbf{b}') = \Lambda^{\textbf{c}}(V)\ , ~\ \ \ \ 
V = \left(
\begin{array}{ccc}
1  &   &   \\
  & V\sct1  &   \\
  &   &   1
\end{array}\right)
\end{align}
where $\Lambda^{\textbf{c}}(V)$ is a controlled-$V$ gate with control string $\textbf{c}$ defined in \cref{eq:Lambda}.  Similarly,
 \begin{align}
 V\sct1(\overline{\textbf{b}}, \overline{\textbf{b}'}) = \Lambda^{\overline{\textbf{c}}}(V^\times)\ , ~\ \ \ \ 
V^\times &:= (X\otimes X) V (X\otimes X)\ .
\end{align}
We conclude that  the 4-level gate in \cref{eq:V1_4qubit} can be rewritten as 
\begin{align}
    V\sct1(\textbf{b}, \textbf{b}')V\sct1(\overline{\textbf{b}}, \overline{\textbf{b}'}) = \Lambda^{\textbf{c}}(V)\Lambda^{\overline{\textbf{c}}}(V^\times) \label{eq:U_b_Lambda_c}\ ,
\end{align}
which corresponds to the circuit in \cref{Eq2024}.

\subsection{From 4-qubit 4-level gates to $n$-qubit 4-level gates} \label{ss:XXYY_distany}
Next, we use these 4-qubit unitaries  to construct general 4-level $n$-qubit unitaries in the form
\be \label{eq:4-level_U1}
U\sct1(\textbf{b}, \textbf{b}')U\sct1(\overline{\textbf{b}}, \overline{\textbf{b}'})\ \ \ \  : U\sct1\in\text{SU}(2)\ ,
\ee
where, $U\sct1(\textbf{b}, \textbf{b}')$ as defined in \cref{eq:def_2-level} is a 2-level unitary acting on $|\textbf{b}\rangle$ and $|\textbf{b}'\rangle$, and we assume 
\be\label{eq:b'notbbar}
\textbf{b}'\neq  \overline{\textbf{b}}, \textbf{b}\ ,
\ee
and $\textbf{b}$ and $\textbf{b}'$ have equal Hamming weights, i.e.  
\be\label{eq:b'notbbar22}
w(\textbf{b})=w(\textbf{b}')\ .
\ee
The assumption that $\textbf{b}'\neq \overline{\textbf{b}}$ in \cref{eq:b'notbbar} implies that $\textbf{b}$ and $\textbf{b}'$ have, at least, one common bit. Without loss of generality, we assume this common bit is the first bit of $\textbf{b}$ and $\textbf{b}'$ and the value of this bit is 1; otherwise, we proceed with $\overline{\textbf{b}}$ and $\overline{\textbf{b}'}$ instead. This means 
\begin{align}    \label{eq:common_bit} |\textbf{b}\rangle=|1\rangle|\textbf{c}\rangle, \quad |\textbf{b}'\rangle=|1\rangle|\textbf{c}'\rangle\ ,    
\end{align}
where $\textbf{c},\textbf{c}'\in\{0,1\}^{n-1}$ are the remaining bits of $\textbf{b}$ and $\textbf{b}'$, which have equal Hamming weights, i.e.,  $w(\textbf{c}) = w(\textbf{c}')$.
With this definition we have
\begin{align} \label{eq:Ubb_to_Ucc1}   U\sct1(\textbf{b},\textbf{b}') = |0\rangle\langle 0|_1\otimes \mathbb{I}+|1\rangle\langle 1|_1\otimes U\sct1(\textbf{c},\textbf{c}')_{2,\dots,n} \ .
\end{align}

Now consider the $n-1$ energy-conserving unitary $U\sct1(\textbf{c},\textbf{c}')_{2,\dots,n}$.  Using the techniques of \cref{sec:nqubit} \cref{step:dist2,step:equalweight}, we can decompose this unitary  as
\be\label{rty}
U\sct1(\textbf{c},\textbf{c}')_{2,\dots,n}=V_T \cdots V_1\ ,
\ee
where $T = \mathcal{O}(n^2)$, and each $V_j$ is a gate in the form of $\Lambda^1(W)$ (or $\Lambda^0(W)$).
In other words, it is a gate in the form of the left-hand side of the figure below. Note that these gates are not realizable with XY interaction alone, because, in general, they do not commute with $X^{\otimes 3}$, i.e., they break the $\mathbb{Z}_2$ symmetry of this interaction. However, we can extend these 3-qubit gates to 4-qubit gates acting on the original 3 qubits and qubit 1 (i.e., the control qubit in \cref{eq:Ubb_to_Ucc1}) using the following rule: 
\vspace{-3mm}
%(Explain the transformation ...)
\begin{align}
\begin{minipage}{8em}
\centering
\begin{tikzpicture}[scale=1]
    \useasboundingbox (-0.75, -1) rectangle (0.75,1.5);
    \coordinate (left) at (-0.75, 0);
    \coordinate (right) at (0.75,0);
    \coordinate (l1) at (0,-0.5);
    \coordinate (l2) at (0,0);
    \coordinate (l3) at (0,0.75);
    \foreach \i in {l1,l2,l3}
    {
        \draw (left|-\i) -- (right|-\i);
    }
    \node[tensor2h] (V01) at ($(l1)!0.5!(l2)$) {$W$};
    \draw (V01) -- (V01|-l3) pic {c1};
\end{tikzpicture}
\end{minipage}
&\longrightarrow
\begin{minipage}{10em}
\centering
\begin{tikzpicture}[scale=1]
    \useasboundingbox (-0.75, -1) rectangle (1.5,1.5);
    \coordinate (left) at (-0.75, 0);
    \coordinate (right) at (1.5,0);
    \coordinate (l1) at (0,-0.5);
    \coordinate (l2) at (0,0);
    \coordinate (l3) at (0,0.75);
    \coordinate (l4) at (0,1.25);
    \foreach \i in {l1,l2,l3}
    {
        \draw (left|-\i) -- (right|-\i);
    }
    \draw[highlighted] (left|-l4) -- (right|-l4);
    \node[tensor2h] (V01) at ($(l1)!0.5!(l2)$) {$W^\times$};
    \draw (V01) -- (V01|-l3) pic {c0} -- (V01|-l4) pic{c0};
    \node[tensor2h] (V02) at ($(V01)+(0.75,0)$) {$W$};
    \draw (V02) -- (V02|-l3) pic {c1} -- (V02|-l4) pic{c1};
    %\node[align=center, left=0cm of left|-l4] {Extra\\control};
\end{tikzpicture}
\end{minipage} \nonumber \\
V=\Lambda^1(W) \quad ~&\longrightarrow  \phantom{=} \Lambda^{1}(V) \Lambda^{0}(V^\times)\label{eq:to_parity_ctrl} \\ & \phantom{\longrightarrow} = \Lambda^{11}(W) \Lambda^{00}(W^\times)\nonumber
\end{align}
where the highlighted qubit is qubit 1. In other words, we replace each 3-qubit gate $V_j$ in \cref{rty} with  4-qubit gates 
\begin{align} \label{eq:Lambda1VLambda0V}
&\Lambda^1(V_j)\Lambda^0(V_j^\times)=\nonumber\\
&\big[|0\rangle\langle 0|_1\otimes \mathbb{I}+ |1\rangle\langle 1|_1\otimes V_j\big]\big[|1\rangle\langle 1|_1\otimes \mathbb{I}+ |0\rangle\langle 0|_1\otimes V^\times_j\big]\ ,  
\end{align}
where $V_j^\times  = X^{\otimes 3} V_j X^{\otimes 3}$, and the subscript 1 in the right-hand side means the control qubit is qubit 1.  
 Notice that similar to $V_j$, $V_j^\times$ is also a 2-level 3-qubit unitary. Then, the 4-qubit unitary in \cref{eq:Lambda1VLambda0V} is in the form of \cref{Eq2024}, which can be realized using XY interaction alone with the methods developed in the previous section.  In \cref{fig:ctrl_to_parity_ctrl} we present an example of this circuit conversion, which is further discussed below (In this Figure, we have highlighted qubit 
 1).  
%$\Lambda^{11}(W)\Lambda^{00}(W^\times)$ (or $\Lambda^{10}(W)\Lambda^{01}(W^\times)$ if $V_j=\Lambda^0(W)$) and

%Note that while, in general, the 3-qubit gate in the left-hand side of \cref{eq:to_parity_ctrl} breaks the $\mathbb{Z}_2$ symmetry, the 4-qubit gate in the right-hand side respects this symmetry and therefore is realizable with the techniques introduced in the previous section, using XY interaction alone. 

%Furthermore, 

It can be easily shown that the resulting circuit realizes  the desired unitary $U\sct1(\textbf{b},\textbf{b}') U\sct1(\overline{\textbf{b}},\overline{\textbf{b}'})$: To see this note  the resulting circuit  has two types of gates: $\Lambda^1(V_j)$, which are activated when the first qubit is $|1\rangle$, and $\Lambda^0(V_j^\times)$, which are activated when  
it  is $|0\rangle$. Clearly, every pair of gates from different types commute with each other. Then, using \cref{eq:Ubb_to_Ucc1} it can be seen that the first family realizes unitary $U\sct1(\textbf{b},\textbf{b}')$, namely
\begin{align}
U\sct1(\textbf{b},\textbf{b}') 
&= \Lambda^1(U\sct1(\textbf{c},\textbf{c}')) \nonumber\\
&= \Lambda^1(V_T\cdots V_1) \nonumber\\
& = \Lambda^1(V_T) \cdots \Lambda^1(V_1) \,, 
\end{align}
and the second family realizes the unitary 
\begin{align} \label{eq:Ubb_to_Ucc2}
    U\sct1(\overline{\textbf{b}},\overline{\textbf{b}'}) &=X^{\otimes n} U\sct1(\textbf{b},\textbf{b}')
    X^{\otimes n}\nonumber\\
    &=X^{\otimes n} \Lambda^1(V_T) \cdots \Lambda^1(V_1) 
    X^{\otimes n}\nonumber\\
    &= \Lambda^0(V_T^\times) \cdots \Lambda^0(V_1^\times) \ .
\end{align}

Therefore, we conclude that

\begin{lemma}
For any $\textbf{b}, \textbf{b}'\in\{0,1\}^n$ 
with equal Hamming weights, 
if    $\textbf{b}'\neq  \overline{\textbf{b}}, \textbf{b}$, then the  4-level unitary $U\sct1(\textbf{b},\textbf{b}') U\sct1(\overline{\textbf{b}},\overline{\textbf{b}'})$ can be realized with $\mathcal{O}(n^2)$ gates $\exp(\i \alpha R): \alpha\in(-\pi,\pi] $, without any ancilla qubits.
\end{lemma}
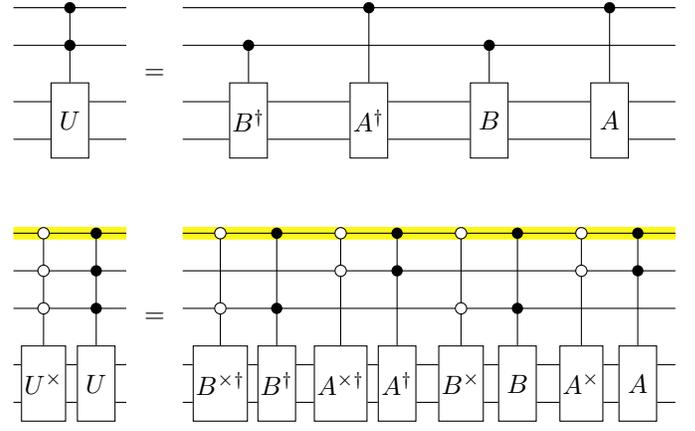
\begin{figure}
    \centering
    \newcommand{\gateskip}{1.6}
    \begin{tikzpicture}
        \coordinate (left) at (-0.5, 0);
        \coordinate (right) at (6.05, 0);
        \coordinate (g0mid) at (-2,0);
        \begin{scope}
            \coordinate (t1) at (0, 0);
            \coordinate (t2) at (0, -0.5);
            \def\cy{{0, 1, 1.5, 1.5, 2}}
            \foreach \i in {1,2,3,4}
            {
                \pgfmathsetmacro\result{\cy[\i] - 0.25}
                \coordinate (c\i) at (0, \result);
            }
            \node[tensor2h] (g1) at ($(t1)!0.5!(t2)+(0.375,0)$)  {$B^\dag$};
            \node[tensor2h] (g2) at ($(g1)+(\gateskip,0)$) {$A^\dag$};
            \node[tensor2h] (g3) at ($(g2)+(\gateskip,0)$) {$B$};
            \node[tensor2h] (g4) at ($(g3)+(\gateskip,0)$) {$A$};
            \foreach \i in {1,2}
            {
                \draw[] (left|-t\i) \foreach \j in {1,2,3,4} { -- (g\j.west|-t\i) (g\j.east|-t\i) } -- (right|-t\i);
            }
            \foreach \i in {1,2}
            {
                \draw[] (left|-c\i) -- (right|-c\i);
            }
            \draw (g1|-c1) pic{c1} -- (g1.north);
            \draw (g3|-c1) pic{c1} -- (g3.north);
            \draw (g2|-c3) pic{c1} -- (g2.north);
            \draw (g4|-c3) pic{c1} -- (g4.north);
            
            \node[tensor2h] (g0) at (g0mid|-g1) {$U$};
            \draw (g0|-c3) pic{c1} -- (g0|-c1) pic{c1} -- (g0.north);
            \coordinate (g0right) at ($(g0) + (0.75,0)$);
            \foreach \i in {t1,t2}
            {
                \draw ($(g0|-\i) - (0.75,0)$) -- (g0.west|-\i) (g0.east|-\i) -- ($(g0|-\i) + (0.75,0)$);
            }
            \foreach \i in {c1,c2}
            {
                \draw ($(g0|-\i) - (0.75,0)$) -- (g0right|-\i);
            }
            \node (equal) at ($(g0right|-t2)!0.5!(left|-c3)$) {$=$};
        \end{scope}
        \begin{scope}[yshift=-3.5cm]
            \coordinate (t1) at (0, 0);
            \coordinate (t2) at (0, -0.5);
            \def\cy{{0, 1, 1.5, 1.5, 2}}
            \foreach \i in {1,2,3,4}
            {
                \pgfmathsetmacro\result{\cy[\i] - 0.25}
                \coordinate (c\i) at (0, \result);
            }
            \node[tensor2h] (g1x) at ($(t1)!0.5!(t2)$) {$B^{\times \dag}$};
            \node[tensor2h] (g1) at ($(g1x)+(0.75,0)$) {$B^\dag$};
            \node[tensor2h] (g2x) at ($(g1x)+(\gateskip,0)$) {$A^{\times \dag}$};
            \node[tensor2h] (g2) at ($(g1)+(\gateskip,0)$) {$A^\dag$};
            \node[tensor2h] (g3x) at ($(g2x)+(\gateskip,0)$) {$B^\times$};
            \node[tensor2h] (g3) at ($(g2)+(\gateskip,0)$) {$B$};
            \node[tensor2h] (g4x) at ($(g3x)+(\gateskip,0)$) {$A^\times$};
            \node[tensor2h] (g4) at ($(g3)+(\gateskip,0)$) {$A$};
            \foreach \i in {1,2}
            {
                \draw[] (left|-t\i) \foreach \j in {1,2,3,4} { -- (g\j x.west|-t\i) (g\j x.east|-t\i) -- (g\j.west|-t\i) (g\j.east|-t\i) } -- (right|-t\i);
            }
            \foreach \i in {1,2}
            {
                \draw[] (left|-c\i) -- (right|-c\i);
            }
            \draw[highlighted] (left|-c4) -- (right|-c4);
            \draw[] (g1|-c4) pic{c1} -- (g1|-c1) pic{c1} -- (g1.north);
            \draw[] (g3|-c4) pic{c1} -- (g3|-c1) pic{c1} -- (g3.north);
            \draw[] (g2|-c4) pic{c1} -- (g2|-c3) pic{c1} -- (g2.north);
            \draw[] (g4|-c4) pic{c1} -- (g4|-c3) pic{c1} -- (g4.north);
            
            \draw[] (g1x|-c4) pic{c0} -- (g1x|-c1) pic{c0} -- (g1x.north);
            \draw[] (g3x|-c4) pic{c0} -- (g3x|-c1) pic{c0} -- (g3x.north);
            \draw[] (g2x|-c4) pic{c0} -- (g2x|-c3) pic{c0} -- (g2x.north);
            \draw[] (g4x|-c4) pic{c0} -- (g4x|-c3) pic{c0} -- (g4x.north);

            \node[tensor2h] (g0x) at ($(g0mid|-g1)-(0.35,0)$)  {$U^\times$};
            \node[tensor2h] (g0) at ($(g0mid|-g1)+(0.35,0)$)  {$U$};
            \coordinate (g0right) at ($(g0mid|-g1) + (0.75,0)$);
            \foreach \i in {c1,c2}
            {
                \draw ($(g0mid|-\i) - (0.75,0)$) -- ($(g0mid|-\i) + (0.75,0)$);
            }
            \foreach \i in {c4}
            {
                \draw[highlighted] ($(g0mid|-\i) - (0.75,0)$) -- ($(g0mid|-\i) + (0.75,0)$);
            }
            \foreach \i in {t1,t2}
            {
                \draw ($(g0mid|-\i) - (0.75,0)$) -- (g0x.west|-\i) (g0x.east|-\i) -- (g0.west|-\i) (g0.east|-\i) -- (g0right|-\i);
            }
            \draw (g0|-c4) pic{c1} -- (g0|-c2) pic{c1} -- (g0|-c1) pic{c1} -- (g0.north);
            \draw (g0x|-c4) pic{c0} -- (g0x|-c2) pic{c0} -- (g0x|-c1) pic{c0} -- (g0x.north);
            
            \node (equal) at ($(g0right|-t2)!0.5!(left|-c4)$) {$=$};
        \end{scope}
    \end{tikzpicture}
    \caption{ The transformation from the circuit for $\Lambda^{11}(U)$ (upper) to $\Lambda^{111}(U)\Lambda^{000}(U^\times)$ (lower) using the mapping in \cref{eq:to_parity_ctrl}. The highlighted qubit is the extra control qubit.}
    \label{fig:ctrl_to_parity_ctrl}
\end{figure}
\subsubsection*{Example: A 5-qubit 4-level unitary}

\cref{fig:ctrl_to_parity_ctrl} shows an example of this circuit conversion where the top is a circuit for the 2-level unitary 
\be
U\sct1(1101,1110)=\Lambda^{11}(U)\ ,
\ee
obtained using the techniques of \cref{sec:nqubit} \cref{step:dist2}, and the bottom is the circuit obtained by applying the transformation in \cref{eq:to_parity_ctrl}. It can be easily seen that the bottom circuit realizes 
\begin{align} \label{eq:5-qubit-target}
    U\sct1(11101,11110)U\sct1(00010,00001) = \Lambda^{111}(U)\Lambda^{000}(U^\times)\ ,
\end{align}
which corresponds to the unitary in \cref{eq:4-level_U1} with
\be
\textbf{b}=11101\ , \ \ \  \textbf{b}'=11110 \ ,
\ee
where
\begin{align}
    U = \left(
\begin{array}{ccc}
1  &   &   \\
  & U\sct1  &   \\
  &   &   1
\end{array}\right), \ \ U^\times = (X\otimes X)U(X\otimes X)\ .
\end{align}

\subsection{The special case of half-filled sector ($m=n/2$)} \label{ss:half-filled}

Next, we focus on the special case of the sector with Hamming weight $m=n/2$ when $n$ the number of qubits is even. This sector, which can be called the ``half-filled" sector, 
splits into two equal-sized subspaces corresponding to $\pm 1$ eigenvalues of $X^{\otimes n}$, denoted as
\be\label{ds}
\mathcal{H}\sct{n/2}=\mathcal{H}\sct{n/2,+}\oplus \mathcal{H}\sct{n/2,-}\ .
\ee
The fact XY interaction commutes with $X^{\otimes n}$, means that this Hamiltonian is block-diagonal with respect to the above decomposition. This in turn implies all realizable unitaries with XY interaction are also block-diagonal. That is
\be\label{dec54}
V\sct{n/2}=V\sct{n/2, +}\oplus V\sct{n/2,-}\ , 
\ee
where $V\sct{n/2, \pm}$ is the component of $V\sct{n/2}$ that acts in the subspace with Hamming weight $n/2$ and eigenvalue $\pm 1$ of $X^{\otimes n}$.  We show that 
\begin{lemma}\label{half}
For $n\geq4$, consider the unitary $V=V\sct{n/2}\oplus\mathbb{I}_{\perp}$  where $\mathbb{I}_{\perp}$ is the identity operator on the subspace orthogonal to $\mathcal{H}\sct{n/2}$. This unitary can be realized with 
XY interaction alone (without any ancillary qubit) if, and only if 
\be \label{eq:detpmis1}
\text{det}(V\sct{n/2,+})=\text{det}(V\sct{n/2,-})=1\ ,
\ee
where $V\sct{n/2,\pm}$ is the component of $V\sct{n/2}$ in the subspace $\mathcal{H}\sct{n/2,\pm}$, as defined in \cref{dec54}. Furthermore, any such unitary can be realized with $\mathcal{O}(n\times 4^n)$ 2-qubit unitaries $\exp(i\theta R)$.
\end{lemma}
%In the following, we prove this lemma which completes the proof of \cref{Thm2}.
\begin{proof}

The necessity of this condition follows from the arguments in \cite{marvian2022restrictions}: For any pair of qubits $i$ and $j$ and $n>2$ we have 
\be
\Tr(X^{\otimes n} R_{ij})=\Tr(R_{ij})=0\ ,
\ee
where  $R_{ij}=(X_i X_j+Y_iY_j)/2$ is the XY interaction. This, in turn, implies 
\be\label{orth}
\frac{\Tr(R_{ij}) \pm  \Tr(X^{\otimes n} R_{ij})}{2} = \Tr(P_{\pm} R_{ij})=0\ ,
\ee%
where $P_\pm=(\mathbb{I}\pm X^{\otimes n})/2$ are the projectors to the subspace with eigenvalue $\pm 1$ of $X^{\otimes n}$. For any unitary $W$ that commutes with $X^{\otimes n}$, let $W=W_+\oplus W_-$ be the decomposition of $W$ relative to the eigensubspaces of $X^{\otimes n}$. Then, for any unitary $W$  which is decomposable as $W = \prod_k \exp(\i\theta_k R_{i_kj_k})$,
\cref{orth} implies that  $\text{det}[P_\pm \exp(\i\theta_k R_{i_kj_k}) P_\pm] = 1$ for every $k$, where the determinant is calculated over the support of $P_\pm$. Therefore, noting that $R_{i_kj_k}$ and $P_\pm$ commute,
\begin{align}
&\text{det}(W_\pm)= \text{det}\left[ \prod_k P_\pm \exp(\i\theta_k R_{i_kj_k}) P_\pm \right] \nonumber \\
&= \prod_k \text{det}\left[ P_\pm \exp(\i\theta_k R_{i_kj_k}) P_\pm \right] = 1\ .
\end{align}
This proves the necessity of the condition. 
In the following, we prove the sufficiency of this condition.

First, for  convenience we choose   a basis for $\mathcal{H}\sct{n/2}$ which is consistent with the decomposition in \cref{ds} such that any basis element is either in $\mathcal{H}\sct{n/2,+}$ or $\mathcal{H}\sct{n/2,-}$. 
To achieve this we define
the basis
\be\label{def3}
|\textbf{b},\pm \rangle=\frac{|\textbf{b}\rangle \pm|\overline{\textbf{b}}\rangle}{\sqrt{2}}\ \ \ \ \ \ : w(\textbf{b})=\frac{n}{2} \ , \  \ 
  \textbf{b}< \overline{\textbf{b}}\ ,
\ee
for $\textbf{b}=b_1\cdots b_n\in\{0,1\}^n$,  
where the first condition means the Hamming weight of $\textbf{b}$ is $n/2$, and in 
 $\textbf{b} < \overline{\textbf{b}}$
we are interpreting $\textbf{b}$ and $\overline{\textbf{b}}$ as the  binary representations of integers (In other words, this condition means the left-most bit of  $\textbf{b}$ is 0). 
 This condition is imposed to  obtain a complete orthonormal basis;  otherwise we will consider  $|\textbf{b},+ \rangle$  and 
$|\overline{\textbf{b}},+ \rangle$ as two separate vectors, whereas according to the definition in \cref{def3} they are indeed equal. (Note that any total order on binary strings can be used here.)  Then, it is clear that the set of vectors $|\textbf{b},\pm \rangle$ with $\textbf{b}$ defined in \cref{def3} form a complete  orthonormal basis for $\mathcal{H}\sct{n/2,\pm}$. 

Next, we show how one can implement unitaries that are 2-level with respect to this basis.
Namely,  for 
$\textbf{b}<\overline{\textbf{b}}$
and $\textbf{c}<\overline{\textbf{c}}$, we consider the unitaries
\bes\label{rgw}
\begin{align}
    &\exp\Big(\i2\theta \big[ |\textbf{b},\pm\rangle\langle \textbf{c},\pm|\ +\ |\textbf{c},\pm\rangle\langle \textbf{b},\pm| \big]\Big)
     \label{eq:2to4_x}\\
    &\qquad  =\mathbb{R}_x(\theta, \textbf{b}, \textbf{c})\ 
    \mathbb{R}_x(\pm \theta, \textbf{b}, \overline{\textbf{c}})\nonumber\\ 
    &\exp\Big(\i2\theta \big[\i |\textbf{c},\pm\rangle\langle \textbf{b},\pm| \ -\ \i |\textbf{b},\pm\rangle\langle \textbf{c},\pm| \big]\Big)
    \label{eq:2to4_y} \\
    &\qquad  =\mathbb{R}_y(\theta, \textbf{b}, \textbf{c})\ 
    \mathbb{R}_y(\pm \theta, \textbf{b}, \overline{\textbf{c}})\nonumber\ ,
\end{align}
\ees
where the identities are proven below. Here, 
 $\mathbb{R}_x$ and $\mathbb{R}_y$ are defined in \cref{AB}, namely 
\bes\label{AB2}
\begin{align}
    \mathbb{R}_x(\theta, \textbf{b}, \textbf{c})&:=\exp\left[\i\theta \left({X}(\textbf{b}, \textbf{c})+{X}(\overline{\textbf{b}}, \overline{\textbf{c}})\right)\right] \\
    \mathbb{R}_y(\theta, \textbf{b}, \textbf{c})&:=\exp\left[\i\theta \left({Y}(\textbf{b}, \textbf{c})+{Y}(\overline{\textbf{b}}, \overline{\textbf{c}})\right)\right]\ ,
\end{align}
\ees
where  we have used the fact that  ${X}(\textbf{b}, \textbf{c})$ and ${X}(\overline{\textbf{b}}, \overline{\textbf{c}})$ commute, and the same fact holds for ${Y}(\textbf{b}, \textbf{c})$ and ${Y}(\overline{\textbf{b}}, \overline{\textbf{c}})$.  Note that $\textbf{b}<\overline{\textbf{b}}$ and $\textbf{c}<\overline{\textbf{c}}$ together imply that $\textbf{b}\neq \overline{\textbf{c}}$, because the left-most bit of $\textbf{b}$ is 0, whereas the left-most bit of $\overline{\textbf{c}}$ is 1. 

To show the identities in \cref{rgw}, we note that for  $\textbf{b}<\overline{\textbf{b}}$ and $\textbf{c}<\overline{\textbf{c}}$, it holds that
\begin{align}
|\textbf{b}\rangle\langle \textbf{c}|+|\overline{\textbf{b}}\rangle\langle \overline{\textbf{c}}|&=|\textbf{b},+\rangle\langle \textbf{c},+|+|\textbf{b},-\rangle\langle \textbf{c},-|\\ 
|\textbf{b}\rangle\langle \overline{\textbf{c}}|+|\overline{\textbf{b}}\rangle\langle {\textbf{c}}|&=|\textbf{b},+\rangle\langle \textbf{c},+|-|\textbf{b},-\rangle\langle \textbf{c},-|\ ,
\end{align}
which, in turn,  implies 
\begin{align}\label{rtq3}
X(\textbf{b},\textbf{c})+X(\overline{\textbf{b}},\overline{\textbf{c}})&=|\textbf{b},+\rangle\langle \textbf{c},+|+|\textbf{c},+\rangle\langle \textbf{b},+|  \nonumber  \\ &+|\textbf{b},-\rangle\langle \textbf{c},-|+|\textbf{c},-\rangle\langle \textbf{b},-|\ ,
\end{align}
and 
\begin{align}\label{rtq4}
X(\textbf{b},\overline{\textbf{c}})+X(\overline{\textbf{b}},{\textbf{c}})&=|\textbf{b},+\rangle\langle \textbf{c},+|+|\textbf{c},+\rangle\langle \textbf{b},+|  \nonumber  \\ &-|\textbf{b},-\rangle\langle \textbf{c},-| - |\textbf{c},-\rangle\langle \textbf{b},-|\ .
\end{align}
Note that the right-hand sides of \cref{rtq3} and \cref{rtq4} commute, which using \cref{AB2},  implies 
$\mathbb{R}_x(\theta_1, \textbf{b}, \textbf{c})$  and $\mathbb{R}_x(\theta_2, \textbf{b}, \overline{\textbf{c}})$ commute and \cref{eq:2to4_x} holds. \cref{eq:2to4_y} can be shown similarly.

Finally, recall that  under the assumption that $\textbf{b}\neq \textbf{c}, \overline{\textbf{c}}$ and $w(\textbf{b})=w(\textbf{c})$, the 4-level unitaries in \cref{AB2} can be realized using the method developed previously in \cref{ss:XXYY_distany}. 
We conclude that unitaries in \cref{rgw} can be realized with XY interaction alone. Next, we use these unitaries to construct other unitaries in the half-filled sectors.

Unitaries in \cref{rgw} are 2-level with respect to the basis in \cref{def3}, namely they act non-trivially on the 2D subspace spanned by 
 $|\textbf{b},+ \rangle$ and $|\textbf{c},+ \rangle$ (or $|\textbf{b},- \rangle$ and $|\textbf{c},- \rangle$), which is restricted to the subspace $\mathcal{H}\sct{n/2,+}$ (or $\mathcal{H}\sct{n/2,-}$). Furthermore, relative to the basis $|\textbf{b},\pm \rangle$ and $|\textbf{c},\pm \rangle$, the unitaries in \cref{rgw} act as $\exp(\i 2\theta X)$ and $\exp(\i 2\theta Y)$, i.e.,  rotations around x and y axes, respectively.

 Using the Euler decomposition, one obtains arbitrary 2-level special  unitaries (i.e., unitaries with determinant 
  1) in the subspace spanned by $|\textbf{b},\pm \rangle$ and $|\textbf{c},\pm \rangle$.  Therefore, using \cref{lem:2-level_decomp}, by combining such unitaries we obtain all unitaries acting on $\mathcal{H}\sct{n/2,+}$ and $\mathcal{H}\sct{n/2,-}$ that satisfy the condition in \cref{eq:detpmis1} of  \cref{half}. In particular, this requires $\mathcal{O}(D^2)=\mathcal{O}(4^n n^{-1})$ unitaries of the type in \cref{rgw}, where
\be
D=\text{dim}(\mathcal{H}\sct{n/2,\pm})=\frac{1}{2}{{n}\choose{n/2}} \approx \frac{2^n}{\sqrt{2\pi n}} \ .
\ee
As we have seen in \cref{ss:XXYY_distany}, each 4-level unitary $\mathbb{R}_{x}(\theta, \textbf{b}, \textbf{c})$ and $\mathbb{R}_{x}(\theta, \textbf{b}, \overline{\textbf{c}})$ can be realized with $\mathcal{O}(n^2)$ gates $\exp(\i \theta R)$. Therefore, we conclude that a general unitary $V\sct{n/2}$ satisfying the condition in \cref{eq:detpmis1} can be realized with $\mathcal{O}(n^2D^2) = \mathcal{O}(4^n n)$  gates $\exp(\i \theta R)$. This completes the proof of \cref{half}.

\end{proof}

\subsubsection*{Example: The sector with Hamming weight 2 of 4 qubits}

In \cref{ss:XX_YY_RX_RY}, we discussed the implementation of unitaries in sectors with Hamming weights 1 and 3 of 4 qubits. Here, we  focus on the sector with Hamming weight 2, and complete the
construction of unitaries in $\mathcal{G}_4$.

In this example, the basis defined in  \cref{def3} is the set of vectors
\bes
\begin{align}
    \ket{0011,\pm} &:=\frac{\ket{0011}\pm\ket{1100}}{\sqrt2}, \\
    \ket{0101,\pm} &:=\frac{\ket{0101}\pm\ket{1010}}{\sqrt2}, \\
    \ket{0110,\pm} &:=\frac{\ket{0110}\pm\ket{1001}}{\sqrt2}\ ,
\end{align}
\ees
which spans the 6-dimensional subspace
$$\mathcal{H}\sct{2}= \mathcal{H}\sct{2,+}\oplus \mathcal{H}\sct{2,-}\cong \mathbb{C}^3\oplus \mathbb{C}^3 \ .$$
Following the above construction, using the XY interaction alone, we can implement any unitary 
inside each of these 3-dimensional subspaces, provided that it has determinant 1 and is 2-level with respect to the above basis. 

For example, consider  $U\in \text{SU}(2)$ with  the Euler decomposition
$$
U=e^{\i\gamma X}e^{\i\beta Y}e^{\i\alpha X}\ .$$
Suppose we want to implement this unitary as a 2-level unitary on the basis vectors $\ket{0011,+}$ and $\ket{0101,+}$, which we denote as unitary $U(0011, +;  0101, +)$. Then, using  \cref{rgw}  we obtain the decomposition
\begin{align}
    &\quad U(0011, +;  0101, +) \nonumber \\
    &= \mathbb{R}_x(\frac{\gamma}{2},0011,0101)  \mathbb{R}_x(\frac{\gamma}{2},0011,1010) \nonumber\\
    &\quad\ \mathbb{R}_y(\frac{\beta}{2},0011,0101)\mathbb{R}_y(\frac{\beta}{2},0011,1010) \nonumber\\
    &\quad\ \mathbb{R}_x(\frac{\alpha}{2},0011,0101)\mathbb{R}_x(\frac{\alpha}{2},0011,1010) \,.
\end{align}

Note that applying \cref{teta,tetay} and their permuted versions, one can implement each unitary in this decomposition using XY interaction alone. Finally, combining such 2-level unitaries on subspace $\mathcal{H}^{(2,+)}$, we obtain the full SU(3) unitary group on this subspace. A similar construction works for $\mathcal{H}^{(2,-)}$. In summary, combined with the results of \cref{ss:XX_YY_RX_RY}, we obtain the group of all 4-qubit unitaries  
satisfying conditions in \cref{Z2,eq:XY_U1,cons-det,cons-det2}, which is isomorphic to the group
$$\mathcal{G}_4\cong \text{SU}(4)\times \text{SU}(3)\times \text{SU}(3) \ .$$

\color{black}
\begin{comment}

\begin{align}
    \mathbb{R}_x(\frac{\alpha}{2},0011,0101)\mathbb{R}_x(\frac{\alpha}{2},0011,1010).
\end{align}
Doing the same to $e^{\i\beta Y}$ and $e^{\i\gamma X}$, we obtain
\end{comment}

\subsubsection*{Example: A family of 6-qubit 2-level unitaries}
As an example, let us consider the  6-qubit unitary
\begin{align} \label{eq:hf_defV}
    V=\exp(i\theta  [\ket{\textbf{b},+}\bra{\textbf{b}
',+}+ \ket{\textbf{b}',+}\bra{\textbf{b}
,+}]) \ ,
\end{align}
where
\bes\label{psidef}
\begin{align}
\ket{\textbf{b},+}=\frac{\ket{010101}+\ket{101010}}{\sqrt{2}} \,, \\
\ket{\textbf{b}',+}=\frac{\ket{100101} + \ket{011010}}{\sqrt{2}} \ ,    
\end{align}
\ees
and 
$$\textbf{b}=010101\ ,\ \  \textbf{b}'=011010 \ . $$
Applying \cref{rgw}, the unitary  $V$ can be decomposed as
\begin{align}
    V=\mathbb{R}_x(\frac{\theta}{2}, \textbf{b}, \textbf{b}')\mathbb{R}_x(\frac{\theta}{2}, \textbf{b}, \overline{\textbf{b}'}) \,.
\end{align}

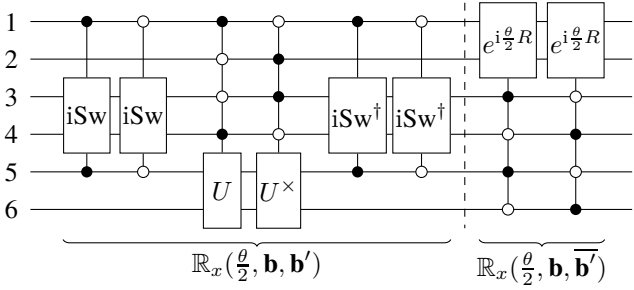
\begin{figure}
    \centering
    \begin{tikzpicture}
        \def\cx{{0.25, 1, 1.75, 2.8, 3.55, 4.6, 5.45, 6.6, 7.5, 8.25}}
        \foreach \i in {0,1,2,3,4,5,6,7,8,9}
        {
            \pgfmathsetmacro\result{\cx[\i]}
            \coordinate (g\i) at (\result,0);
        }
        \foreach \i in {1,2,3,4,5,6}
        {
            \coordinate (l\i) at (0,-0.5*\i);
            \node at (l\i) {\i};
        }
        \foreach \i in {1,2,3,4,5,6}
        {
            \draw (g0|-l\i) -- (g9|-l\i);
        }
        \node[tensor2h] (iswap1) at ($(l3-|g1)!0.5!(l4-|g1)$) {$\iswap$};
        \draw (iswap1) -- (iswap1|-l1) pic{c1} (iswap1) -- (iswap1|-l5) pic{c1};
        \node[tensor2h] (iswap2) at ($(l3-|g2)!0.5!(l4-|g2)$) {$\iswap$};
        \draw (iswap2) -- (iswap2|-l1) pic{c0} (iswap2) -- (iswap2|-l5) pic{c0};
        \node[tensor2h] (Ux) at ($(l5-|g3)!0.5!(l6-|g3)$) {$U$};
        \draw (Ux|-l1) pic{c1} -- (Ux|-l2) pic{c0} -- (Ux|-l3) pic{c0} -- (Ux|-l4) pic{c1} -- (Ux);
        \node[tensor2h] (U) at ($(l5-|g4)!0.5!(l6-|g4)$) {$U^\times$};
        \draw (U|-l1) pic{c0} -- (U|-l2) pic{c1} -- (U|-l3) pic{c1} -- (U|-l4) pic{c0} -- (U);
        \node[tensor2h] (iswap1dg) at ($(l3-|g5)!0.5!(l4-|g5)$) {$\iswap^\dag$};
        \draw (iswap1dg) -- (iswap1dg|-l1) pic{c1} (iswap1dg) -- (iswap1dg|-l5) pic{c1};
        \node[tensor2h] (iswap2dg) at ($(l3-|g6)!0.5!(l4-|g6)$) {$\iswap^\dag$};
        \draw (iswap2dg) -- (iswap2dg|-l1) pic{c0} (iswap2dg) -- (iswap2dg|-l5) pic{c0};
        \node[tensor2h] (ex) at ($(l1-|g7)!0.5!(l2-|g7)$) {$e^{\i\frac{\theta}{2} R}$};
        \draw (ex) -- (ex|-l3) pic{c1} -- (ex|-l4) pic{c0} -- (ex|-l5) pic{c1} -- (ex|-l6) pic{c0};
        \node[tensor2h] (e) at ($(l1-|g8)!0.5!(l2-|g8)$) {$e^{\i\frac{\theta}{2} R}$};
        \draw (e) -- (e|-l3) pic{c0} -- (e|-l4) pic{c1} -- (e|-l5) pic{c0} -- (e|-l6) pic{c1};

        \draw[dashed] ($(l1-|g7)!0.5!(l1-|g6) + (0,0.25)$) -- ($(l6-|g7)!0.5!(l6-|g6) + (0,-0.25)$);

        \coordinate[below=0.4 of l6] (lower);
        \draw[decorate, decoration = {brace,mirror,raise=0pt}] ($(iswap1.west|-lower)+(0,0)$) -- node[below,yshift=-0pt]{$\mathbb{R}_x(\frac{\theta}{2}, \textbf{b}, \textbf{b}')$} ($(iswap2dg.east|-lower)+(0,-0)$);
        \draw[decorate, decoration = {brace,mirror,raise=0pt}] ($(ex.west|-lower)+(0,0)$) -- node[below,yshift=-0pt]{$\mathbb{R}_x(\frac{\theta}{2}, \textbf{b}, \overline{\textbf{b}'})$} ($(e.east|-lower)+(0,-0)$);
        
    \end{tikzpicture}
    \caption{The circuit for realizing the 6-qubit unitary $V=\exp(i\theta  [\ket{\textbf{b},+}\bra{\textbf{b}
',+}+ \ket{\textbf{b}',+}\bra{\textbf{b}
,+}])$, where $\ket{\textbf{b},+}$ and $\ket{\textbf{b}',+}$ are defined in \cref{psidef} . The first six gates in the circuit implement $\mathbb{R}_x(\frac{\theta}{2}, \textbf{b}, \textbf{b}')$, and the last two gates implement $\mathbb{R}_x(\frac{\theta}{2}, \textbf{b}, \overline{\textbf{b}'})$. Each of the 4-level gates, namely the pair of controlled-$U$ and $U^\times$ and the pair of controlled-$e^{\i\frac{\theta}{2} R}$, can be implemented with 10 ($T(3)=10$ in \cref{lem:TnOn2}) 4-qubit 4-level gates using the construction in \cref{ss:XXYY_distany}.}
    \label{fig:halffilled_example}
\end{figure}

\begin{figure}
    \centering
    \begin{tikzpicture}
        \def\cx{{0.25, 1, 1.75, 2.8, 3.55, 4.6, 5.45, 6.05, 7.5, 8.25}}
        \foreach \i in {0,1,2,3,4,5,6,7,8,9}
        {
            \pgfmathsetmacro\result{\cx[\i]}
            \coordinate (g\i) at (\result,0);
        }
        \foreach \i in {2,3,4,5,6}
        {
            \coordinate (l\i) at (0,-0.5*\i);
            \node at (l\i) {\i};
        }
        \foreach \i in {2,3,4,5,6}
        {
            \draw (g0|-l\i) -- (g6|-l\i);
        }
        \node[tensor2h] (iswap1) at ($(l3-|g1)!0.5!(l4-|g1)$) {$\iswap$};
        \draw  (iswap1) -- (iswap1|-l5) pic{c1};

        \node[tensor2h] (Ux) at ($(l5-|g3)!0.5!(l6-|g3)$) {$U$};
        \draw  (Ux|-l2) pic{c0} -- (Ux|-l3) pic{c0} -- (Ux|-l4) pic{c1} -- (Ux);

        \node[tensor2h] (iswap1dg) at ($(l3-|g5)!0.5!(l4-|g5)$) {$\iswap^\dag$};
        \draw  (iswap1dg) -- (iswap1dg|-l5) pic{c1};

        \coordinate[below=0.4 of l6] (lower);
        
    \end{tikzpicture}
    \caption{The circuit for realizing the  5-qubit unitary $\exp\big(\i\frac{\theta}{2}X(\mathbf{c}, \mathbf{c}')\big)$, where $\mathbf{c}=01010$ and $\mathbf{c}'=00101$. The qubits shown here correspond to the qubits 2 to 6 in \cref{fig:halffilled_example}. Based on this circuit, we construct the first part of the circuit in \cref{fig:halffilled_example}, which realizes the unitary $\mathbb{R}_x(\frac{\theta}{2}, \textbf{b}, \textbf{b}')=\mathbb{R}_x(\frac{\theta}{2}, \overline{\textbf{b}}, \overline{\textbf{b}'})$  on 6 qubits.}
    \label{fig:halffilled_example_5qubit}
\end{figure}
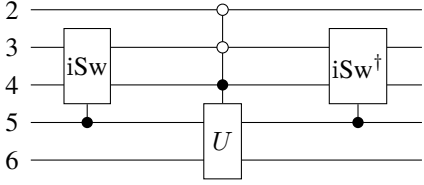
Then, as we further explain below,  applying the above methods we obtain the circuit in \cref{fig:halffilled_example} for implementing the unitary $V$, where the first part of the circuit, i.e., the first six unitaries,  realizes the unitary
\begin{align} \label{eq:RRx010101_011010}
    \mathbb{R}_x(\frac{\theta}{2}, \textbf{b}, \textbf{b}')=\mathbb{R}_x(\frac{\theta}{2}, \overline{\textbf{b}}, \overline{\textbf{b}'})\,,
\end{align}
and, the second part, i.e., the last two unitaries realizes 
\begin{align}
    &\mathbb{R}_x(\frac{\theta}{2}, \textbf{b}, \overline{\textbf{b}'})=\mathbb{R}_x(\frac{\theta}{2}, \overline{\textbf{b}}, {\textbf{b}'}) %\\
   % = & \exp\left[\frac{\theta}{2} X(010101, 100101)\right] \exp\left[\frac{\theta}{2} X(101010, 011010)\right]\ .
\end{align}
Note that each consecutive pair of unitaries in \cref{fig:halffilled_example} is an energy-conserving unitary that is 4-level in the computational basis, which can be itself decomposed into a sequence of gates $\exp(\i\alpha R): \alpha\in(-\pi,\pi]$ using the methods of \cref{ss:XXYY_distany} (In particular, for each of these 4-level unitaries  the conditions in \cref{eq:b'notbbar,eq:b'notbbar22} hold).

Both parts of this circuit are obtained by applying the method in \cref{ss:XXYY_distany}. 
 In particular, the first part of the circuit, which realizes the unitary in \cref{eq:RRx010101_011010}, is obtained by applying the conversion rule in \cref{ss:XXYY_distany} to the 5-qubit circuit in \cref{fig:halffilled_example_5qubit}. This 5-qubit circuit  realizes the unitary
 $\exp\big(\i\frac{\theta}{2}X(\mathbf{c}, \mathbf{c}')\big)\ , $
where $\mathbf{c}=01010$ and $\mathbf{c}'=00101$, which are obtained by removing the first bit of $\overline{\mathbf{b}}$ and $\overline{\mathbf{b}'}$, i.e.,  
$$\overline{\mathbf{b}}=1\mathbf{c}=101010\ ,\ \ \overline{\mathbf{b'}}=1\mathbf{c'}=100101\ . $$

%which acts as the unitary $\exp(\i\frac{\theta}{2}X)$ relative to the basis  $|\mathbf{c}\rangle$ and $|\mathbf{c}'\rangle$.  We consider this unitary because $\mathbf{b}$ and $\mathbf{b}'$

%Here, the first bit is the removed common bit, and \cref{eq:common_bit} holds  with 

The 5-qubit circuit in \cref{fig:halffilled_example_5qubit} itself is obtained by the construction in the proof of \cref{lem10}. In particular, this construction gives the decomposition  
\begin{align}
&\exp\big(\i\frac{\theta}{2}X(\mathbf{c}, \mathbf{c}')\big)=\exp\big(\i\frac{\theta}{2}X(01010,00101) \big) \nonumber \\ &=  \Lambda^1_5(\iswap_{34}^\dag) U\sct1(00110,00101) \Lambda^1_5(\iswap_{34}) \, \label{eq:halffilled_5qubit}
\end{align}
which corresponds to the circuit in \cref{fig:halffilled_example_5qubit}.\footnote{In particular, the sequence of bit strings in \cref{seq}, which gives the above circuit, is
$$    \textbf{b}=0\underset{l_1}{1}\underset{r_1}{0}\underset{l_2}{1}\underset{r_2}{0} \to  00110 \to 00101=\textbf{b}' \, .$$ The controlled-iSWAP circuit in \cref{eq:K_as_cisw} reads $K = \Lambda^1_{l_2}(\iswap_{l_1,r_1}) = \Lambda^1_5(\iswap_{34})$, which is the controlled-$\iSWAP$ gate acting on qubits 3 and 4 with the control qubit being qubit 5, the second last qubit in \cref{fig:halffilled_example} and in \cref{fig:halffilled_example_5qubit}. } Here, $U\sct1$ and $U$ are defined through \cref{eq:phase_from_iswaps} as
\bes
\begin{align}
    U &= \left(\begin{matrix}
        1&&\\&U\sct1 &\\&&1
    \end{matrix}\right) \ , \\
    U\sct1 &= \left(\begin{matrix}
        \i&\\&1
    \end{matrix}\right) e^{\i\frac{\theta}{2}X} \left(\begin{matrix}
        -\i&\\&1
    \end{matrix}\right) = \exp\left[\i\frac{\theta}{2}\left(\begin{matrix}&\i\\-\i&\end{matrix}\right)\right]  \\
    U^\times &:= (X\otimes X) U (X\otimes X)\ ,
\end{align}
\ees
where the matrix representation for $U$ is written in the computational basis of qubits 5 and 6 with the ordering $\ket{00}_{65},\ket{01}_{65},\ket{10}_{65},\ket{11}_{65}$, with qubit 6 first.\\

\newpage

\subsection{The full set of realizable unitaries with XY interactions\\ (Completing the proof of \cref{Thm2})} \label{ss:XY_semi}

Putting everything 
together, we obtain a method for implementing a general unitary satisfying the conditions in \cref{Thm2}, using the XY interaction alone. As emphasized in \cref{lem:con-ham}, the additional $\mathbb{Z}_2$ symmetry of XY interaction allows us to restrict our attention to the subspace with Hamming weight $\le n/2$; if the implemented unitary  coincides with the desired unitary in this subspace, then it is equal to the desired unitary.
 In \cref{ss:XXYY_distany}, we  constructed the 4-level unitary $U\sct1(\textbf{b},\textbf{b}')U\sct1(\overline{\textbf{b}},\overline{\textbf{b}'})$, which  assuming $w(\textbf{b})<n/2$ acts as a 2-level unitary in this subspace.

Now, suppose one wants to implement a unitary $V$ satisfying \cref{Z2,eq:XY_U1,cons-det,cons-det2}, namely  $V \in \mathcal{SV}_n^{\rm U(1)}$ that satisfies the $\mathbb{Z}_2$ symmetry as well as the corresponding determinant constraint in \cref{cons-det2}. 
Using the techniques of  \cref{sec:nqubit} \cref{step:semi-universal}, we obtain a circuit $C$ containing $\mathcal{O}(4^n/\sqrt{n})$ number of  gates $U\sct1(\textbf{b},\textbf{b}')$ with $w(\textbf{b})<n/2$ that realizes $V$ in the subspace $\mathcal{H}\sct{ <  n/2 }$. Notice that the resulting unitary  acts trivially on the subspace orthogonal to  $\mathcal{H}\sct{<n/2}$.

Then, replacing each unitary $U\sct1(\textbf{b},\textbf{b}')$ in this circuit with 4-level unitaries $U\sct1(\textbf{b},\textbf{b}')U\sct1(\overline{\textbf{b}},\overline{\textbf{b}'})$, we obtain a circuit $C'$ that respects the $\mathbb{Z}_2$ symmetry, and behaves identical to $C$ in
$\mathcal{H}\sct{<n/2}$, and hence realize the desired unitary $V$ in subspace $\mathcal{H}\sct{<n/2}$. From \cref{ss:XXYY_distany}, each unitary $U\sct1(\textbf{b},\textbf{b}')U\sct1(\overline{\textbf{b}},\overline{\textbf{b}'})$ with $w(\textbf{b})=w(\textbf{b}')$ can be implemented with $\mathcal{O}(n^2)$ gates $\exp(\i\alpha R): \alpha\in(-\pi,\pi]$.

Additionally, if $n$ is even, using the techniques developed in  \cref{ss:half-filled}, we can construct a circuit with gates $\exp(\i\alpha R): \alpha\in(-\pi,\pi]$  to implement $V\sct{n/2}$, the component of $V$ in the sector with Hamming weight $n/2$. Combining this circuit with $C'$, the resulting circuit implements $V$ in the sectors of Hamming weights $0,\dots,\lceil n/2 \rceil$.
By \cref{lem:con-ham}, the resulting circuit realizes $V$ in the full $n$-qubit Hilbert space.

The total number of gates $\exp(\i\alpha R): \alpha\in(-\pi,\pi]$ in this circuit is
\begin{align} \label{eq:XY_O4nn32}
    \mathcal{O}(n^2)\times \mathcal{O}\left(\frac{4^n}{\sqrt{n}}\right) + \mathcal{O}(4^n n)= \mathcal{O}(4^n n^{3/2}) \ .
\end{align}

\subsection{Breaking the $\mathbb{Z}_2$ symmetry of XY interaction with a single ancilla qubit}\label{ss:Z2_symmetry}

We have fully characterized $\mathcal{G}_n$, the set of realizable unitary transformations with XY interaction alone (without ancilla qubits). Namely, this is the set of  unitaries   satisfying all constraints \cref{Z2,eq:XY_U1,cons-det,cons-det2}, where \cref{cons-det2} is relevant only in the case of even $n$. Next, we show that how one can lift the two constraints related to the $\mathbb{Z}_2$ symmetry, i.e.,  \cref{Z2,cons-det2},  using a single ancilla qubit in the initial state $|0\rangle$ or $|1\rangle$, which fully breaks the symmetry (In this case the ancilla can be interpreted as a quantum reference frame or asymmetry catalyst \cite{marvian2013theory}).

In particular, for any unitary $V\in \mathcal{SV}_n^{\text{U}(1)}$, the unitary
\be \label{eq:V_anc_Z2}
{V}'=V\otimes |0\rangle\langle 0|+(X^{\otimes n} V X^{\otimes n})\otimes |1\rangle\langle 1|\ ,
\ee
satisfies
\be \label{eq:XY_V_anc}
{V}'(|\psi\rangle\otimes |0\rangle)=(V|\psi\rangle)\otimes |0\rangle\ ,
\ee
for all $n$-qubit states $|\psi\rangle\in(\mathbb{C}^2)^{\otimes n}$.  In the following, we show that this unitary satisfies all the constraints in  \cref{Z2,eq:XY_U1,cons-det,cons-det2}, and therefore is realizable  with XY interaction alone. 

First, it is straightforward  to see that 
\be \label{eq:XY_V_Z2}
X^{\otimes (n+1)} {V}' X^{\otimes (n+1)}={V}'\ ,
\ee
and
\be\label{eq:XY_V_U1}
(\sum_{j=1}^{n+1} Z_j) {V}'={V}' (\sum_{j=1}^{n+1} Z_j)\ .
\ee
Next, to verify the constraints in \cref{cons-det,cons-det2}, note that  the component of $V'$ in the sector with Hamming weight $m=0,\cdots, n+1$ is 
\be
V'{}\sct{m}=V\sct{m}\otimes |0\rangle\langle 0|+X^{\otimes n} {V}\sct{n+1-m} X^{\otimes n}\otimes |1\rangle\langle 1| \ ,
\ee
where $V\sct{m}=\Pi\sct{m} V \Pi\sct{m}$ is the component of $V$ in the sector with Hamming weight 
$m$, and  $V'{}\sct{m}$ and
 ${V}\sct{n+1-m}$ are defined in a similar fashion. Then,
\be
\text{det}(V'{}\sct{m})=\text{det}(V\sct{m}) \times  \text{det}(V\sct{n+1-m})=1\ , 
\ee
which implies condition in \cref{cons-det} is satisfied. Finally, we show that when $n+1$ is even, $V'$ also satisfies the condition in \cref{cons-det2}. To see this note that
 \be
V'{}\sct{\frac{n+1}{2}}=V\sct{\frac{n+1}{2}}\otimes |0\rangle\langle 0|+X^{\otimes n} {V}\sct{\frac{n+1}{2}} X^{\otimes n}\otimes |1\rangle\langle 1|\ .
\ee
Suppose $V\sct{(n+1)/2}$ has the eigendecompostion
\be \label{eq:V_half_eigendecomp}
V\sct{\frac{n+1}{2}}=\sum_{j} e^{\i\theta_j} |\psi_j\rangle\langle\psi_j|\ ,
\ee
where $\{|\psi_j\rangle\}$ is an orthonormal basis for the subspace of $(\mathbb{C}^2)^{\otimes n}$ with Hamming weight $(n+1)/2$. Then, one can easily see from  \cref{eq:V_anc_Z2,eq:V_half_eigendecomp} that $V'{}\sct{\frac{n+1}{2}}$ has the eigen-decomposition
\begin{align}
    & V'{}\sct{\frac{n+1}{2}} \nonumber \\
    &= \sum_{j} e^{\i\theta_j} \big( \ketbrasame{\psi_j}\otimes\ketbrasame{0}  +X^{\otimes n}\ketbrasame{\psi_j}X^{\otimes n}\otimes\ketbrasame{1} \big) \nonumber\\
    & = \sum_{j} e^{\i\theta_j} \big(|\Psi_j^+\rangle\langle \Psi_j^+|+|\Psi_j^-\rangle\langle \Psi_j^-|\big)\ , 
\end{align}
where 
\be
|\Psi^\pm_j\rangle=\frac{|\psi_j\rangle\otimes |0\rangle\pm X^{\otimes n}|\psi_j\rangle\otimes |1\rangle}{\sqrt{2}}\ ,
\ee
is an orthonormal basis for the eigen-subspace of $(\mathbb{C}^2)^{\otimes (n+1)}$ with Hamming weight $(n+1)/2$, and satisfies
\be
X^{\otimes (n+1)} |\Psi^\pm_j\rangle=\pm |\Psi^\pm_j\rangle\ . 
\ee
This immediately implies that 
\be
V'{}\sct{\frac{n+1}{2}}=V'{}\sct{\frac{n+1}{2},+} \oplus  V'{}\sct{\frac{n+1}{2},-}\ , 
\ee
where 
\be
V'{}\sct{\frac{n+1}{2},\pm}= \sum_{j} e^{\i\theta_j} 
 |\Psi_j^\pm\rangle\langle \Psi_j^\pm|\ ,
\ee
is the component of $V'{}\sct{\frac{n+1}{2}}$ in the eigen-subspace of $X^{\otimes (n+1)}$ with eigenvalue $\pm 1$. We conclude that
\be
\text{det}(V'{}\sct{\frac{n+1}{2},\pm})=\prod_{j} e^{\i\theta_j}=\text{det}(V\sct{\frac{n+1}{2}})=1\ ,  
\ee
where here the determinant is the product of non-zero eigenvalues. 

We conclude that the condition in \cref{cons-det2} is also satisfied. In summary, we showed that

\begin{corollary}
Any unitary transformation in $\mathcal{SV}_n^{\text{U}(1)}$ can be realized using 2-qubit gates $\exp(\i\alpha R): \alpha\in(-\pi,\pi]$ and a single ancilla qubit.
\end{corollary}
\subsection{All energy-conserving unitaries with 2 ancilla qubits} \label{ss:XY_univ}
Finally, combining this with the result of \cref{sec:nqubit} \cref{step:universal} that allows us to circumvent the constraints in \cref{cons-det} with an ancilla qubit, we can  
implement a general energy-conserving unitary (up to a global phase) using only XY interaction and 2 ancilla qubits. 

In particular, in \cref{sec:nqubit} \cref{step:universal} we showed that for any $n$-qubit energy-conserving unitary $V \in \mathcal{V}_n^{\rm U(1)}$, there exists a unitary in $\widetilde{V} \in \mathcal{SV}_{n+1}^{\rm U(1)}$, such that
\be
\widetilde{V}(|\psi\rangle\otimes |0\rangle_\text{anc1})=(V|\psi\rangle)\otimes |0\rangle_\text{anc1}\ ,
\ee
for all $|\psi\rangle\in(\mathbb{C}^2)^{\otimes n}$. 
Furthermore, from \cref{ss:Z2_symmetry},
\be
\widetilde{V}'=\widetilde{V}\otimes |0\rangle\langle 0|+X^{\otimes (n+1)}\widetilde{V}X^{\otimes (n+1)}\otimes |1\rangle\langle 1| \ ,
\ee
satisfies 
\be
\widetilde{V}'(|\Psi\rangle\otimes |0\rangle_\text{anc2})=(\widetilde{V}|\Psi\rangle)\otimes |0\rangle_\text{anc2}\ , 
\ee
for all $|\Psi\rangle\in\mathbb{C}^{\otimes (n+1)}$. Moreover, since $\widetilde{V}\in\mathcal{SV}_{n+1}^{\text{U}(1)}$, the argument in the previous section implies that $\widetilde{V}'$ is realizable with XY interaction alone. Choosing $|\Psi\rangle=|\psi\rangle\otimes |0\rangle$  we conclude that  
\begin{align}
\widetilde{V}'(|\psi\rangle\otimes |0\rangle_\text{anc1}\otimes |0\rangle_\text{anc2})=(V|\psi\rangle)\otimes |0\rangle_\text{anc1}\otimes |0\rangle_\text{anc2}\ ,
\end{align}
for all $n$ qubit state  $|\psi\rangle$. This completes the proof of the last part of \cref{Thm1}.

\section{Approximate Universality} \label{sec:approx}

For some applications, e.g., in the context of quantum computing,  we are interested in the notion of \emph{approximate} universality, where the desired unitary can be implemented using gates from a finite elementary gate set, for instance,
$S$ gate and  $\sqrt{\i \text{SWAP}}$. Clearly, with finite number of gates from this finite gate set, 
 generic energy-conserving unitary $V$  can only be realized with a non-zero error, which can be quantified by the operator norm distance
\begin{align}
	\|V-V'\|_\infty := \sup_{\ket{\psi}:\|\ket{\psi}\|=1}\|(V-V')\ket{\psi}\| \,,
\end{align}
where $V'$ is the realized unitary.

In the following we show that for $n=2$ qubits,  any desired energy-conserving unitary in $\mathcal{SV}_2^{U(1)}$ can be realized with arbitrary small error using  $T$ and  $\sqrt{\i\text{SWAP}}$  gates whereas this is not possible  if one uses $S=T^2$ instead of $T$. In the latter case, the generated group is finite, and hence is not dense in
$\mathcal{SV}_2^{U(1)}$. 
On the other hand, in the case of $n\ge 3$ qubits $T$ gates are not needed: Any desired element of $\mathcal{SV}_n^{U(1)}$ 
can be realized with single-qubit $S$ gate and $\sqrt{\i\text{SWAP}}$ gate, with arbitrary small error.  

The results in this section rely on the Solovay-Kitaev theorem \cite{kitaev1997quantum,kitaev2002classical}. In particular, we use a recent variant of this result \cite{kuperberg2023breaking}, which states that, for any gate set that generates a dense subgroup of $\text{SU}(d)$, if the inverse of each gate is contained in the gate set,  then 
any $V\in\text{SU}(d)$ can be approximated to precision $\epsilon$ with 
\be\label{SK0}
N=\mathcal{O}(\log^{\nu}(\epsilon^{-1}))\ ,
\ee
number of gates in the gate set, where 
\be\label{SK}
\nu>\SKexp\approx 1.44042\  . 
\ee
Furthermore, the sequences that achieve this bound can be found efficiently.

\subsection{Approximate semi-universality of $T+\sqrt{\i\text{SWAP}}$  for 2 qubits}

Here, we show how a general 2-qubit unitary in $\mathcal{SV}_2^{U(1)}$, i.e., a unitary in the form
\be\nonumber
V=
\left(
\begin{array}{ccc}
1  &   &   \\
  & V\sct1  &   \\
  &   &   1
\end{array}\right)\ \  : V\sct1\in \text{SU}(2)\ ,
\ee
can be realized with arbitrary small error using $T$ and $\sqrt{\i\text{SWAP}}$ gates.

Recall that in the case of 2 qubits, any element of $\mathcal{SV}_2^{U(1)}$ can be realized with $S$ gate and 
 $\exp(i\alpha R): \alpha\in(-\pi,\pi]$ gate (See the circuit below \cref{eq:XYeuler}).
 Luckily \cref{2-qubit} indicates that   the approximate implementation of the gate $\exp(\i\alpha R)$ can be fully characterized using the same techniques and results that have been previously developed in the case of approximate single-qubit unitaries.  
 To see this, recall the correspondence in \cref{eq:VV1} which implies $\siSWAP$ and $\exp(\i\alpha Z)$ gates acts as $\exp(\i\pi X/4)$ and rotation around $z$ in the subspace spanned by $|0 1\rangle$ and $|1 0\rangle$.
 Using  the standard results in quantum computing, one can show that $\exp(\i\pi X/4)$ and $\exp(\i\pi Z/8)$ generate a dense subgroup of SU(2). To see this note that
\be
S^\dag \exp(\i\pi X/4) S^\dag=\frac{1}{\sqrt{2}}\left(
\begin{array}{cc}
1  &  1   \\ 1 & -1
  \end{array}
\right)=H\ ,
\ee
which is the Hadamard gate. But, it is well-known that $T$ and $H$ gates together generate a dense subgroup of SU(2) \cite{boykin1999universalieee,kitaev2002classical}. Combined with the above equation and the fact that 
$S^\dag=T^6$, this implies that $T$ and $\exp(\i\pi X/4)$ generate a dense subgroup of SU(2).

We conclude that any $V\sct1\in\text{SU}(2)$ can be approximately implemented  with $T$ and $\exp(\i\pi X/4)$ gates. By the correspondence in \cref{eq:VV1}, this further implies that any $V\in\mathcal{SV}_2^{U(1)}$ can be approximated with the same number of $T$ and $\sqrt{\i\text{SWAP}}$ gates. In particular,  there exists $\widetilde{V}$, such that $\|V-\widetilde{V}\|_\infty\le \epsilon$ and $\widetilde{V}$ can be realized with $\mathcal{O}(\log^{\nu}(\epsilon^{-1}))$ of $T$  and $\sqrt{\i\text{SWAP}}$ gates, for any $\nu>\SKexp$. Indeed, the above observation implies that one can use the standard existing softwares for approximate implementation of single-qubit gates with $H$ and $T$ gates, to find the sequence of $\sqrt{\i\text{SWAP}}$ and $T$ gates that realize any 2-qubit energy-conserving unitary $V$ in the above form.
\subsection{Approximate semi-universality of $S+\sqrt{\i\text{SWAP}}$  for 3 qubits}

Next, we study the case of $n=3$ qubits. As the first setp, we characterize the group generated by $\sqrt{\i\text{SWAP}}$ gates alone.
 Consider the  subgroup of 3-qubit unitaries $\mathcal{SV}_3^{\text{U}(1)}$ satisfying the additional  $\mathbb{Z}_2$ symmetry in \cref{Z2}, i.e., the group
\be\label{G3}
\mathcal{G}_3=\{U\in \mathcal{SV}_3^{U(1)} : 
X^{\otimes 3} UX^{\otimes 3}= U  \}\ .
\ee
Relative to the computational basis with the following (unconventional) ordering
\be\label{ordering}
|000\rangle;\  |001\rangle , |010\rangle, |100\rangle ;\  |110\rangle , |101\rangle, |011\rangle ;\  |111\rangle\ ,
\ee
elements of $\mathcal{G}_3$ are represented as 
\begin{align}\label{G32}
	U = \left(\begin{matrix}
		1 &&& \\ &U\sct1&& \\ &&U\sct1& \\ &&&1
	\end{matrix}\right)\ \ \ : U\sct1\in \text{SU}(3)\ .
\end{align}
The symmetries of XY interaction imply that any unitary generated by this interaction on 3 qubits, including $\sqrt{\i\text{SWAP}}$ on any pair of qubits, is inside   $\mathcal{G}_3$. Conversely, we show that the unitaries generated by $\sqrt{\i\text{SWAP}}$ on all pairs of qubits generate a dense subgroup of  $\mathcal{G}_3$.

\begin{theorem}\label{Thm:dense}
Consider the group of 
3-qubit unitaries generated by gates 
\be
\sqrt{\i\text{SWAP}}_{12}\ , \ \sqrt{\i\text{SWAP}}_{23} \  , \ \sqrt{\i\text{SWAP}}_{13} \ . 
\ee
This group is a dense subgroup of $\mathcal{G}_3$ defined in \cref{G3} (or, equivalently, \cref{G32}). Hence, by the Solovay-Kitaev theorem, for any unitary $U\in \mathcal{G}_3$ and any  $\epsilon>0$, there exists a sequence $V_1,\cdots, V_N$
of length  $N=\mathcal{O}(\log^{\nu}(\epsilon^{-1}))$ of these gates 
such that $\|U-V_N\cdots V_1\|_\infty\le \epsilon $, where $\nu$ is given in \cref{SK}. Furthermore, this sequence can be found efficiently.  
\end{theorem}
The proof of this theorem is presented in \cref{ss:proof_Thm:dense}.  Before proving this theorem, we discuss two of its important implications. First, combined with  \cref{Thm2}, this theorem implies that the group generated by $\sqrt{\i\text{SWAP}}$ gates on  $n\ge 3$ qubits is a dense subgroup of the group $\mathcal{G}_n$.

Second, note that by combining $\mathcal{G}_3$ with a generic energy-conserving unitary, one obtains the group $\mathcal{SV}_3^{\text{U}(1)}$, i.e., the group of all unitaries
in the form    
\begin{align}\label{SU}
	V = \left(\begin{matrix}
		1 &&& \\ &V\sct1&& \\ &&V\sct2& \\ &&&1
	\end{matrix}\right)\ \ \ \ \ \ : V\sct1 , V\sct2  \in \text{SU}(3)\ ,
\end{align}
with respect to the basis in \cref{ordering}. 

\begin{lemma}\label{lem65}
Suppose a 3-qubit energy-conserving unitary $J$ with respect to the basis in \cref{ordering} has the decomposition    
\begin{align}
	J = \left(\begin{matrix}
		e^{i\phi_0} &&& \\ &J\sct1&& \\ &&J\sct2& \\ &&& e^{i\phi_3}
	\end{matrix}\right)\ .
\end{align}
 If $J\sct2 \neq J\sct1 e^{\i \gamma}$ for some phase $e^{\i \gamma}$, or equivalently, if $|\Tr(J\sct1^\dag J\sct2)|<3$, then the group generated by 
$J$ and $\mathcal{G}_3$ contains a dense subgroup of $\mathcal{SV}_3^{U(1)}$, i.e., all unitaries in the form of \cref{SU}. 
In particular, this is the case for $J=S\otimes \mathbb{I}\otimes \mathbb{I}$, where  $ S=e^{\frac{\i\pi}{4}} R_z(-\frac{\pi}{2})
=
\left(
\begin{array}{cc}
1  &     \\ & \i
  \end{array}
\right)$.
\end{lemma}
We prove this lemma in \cref{prooflem65}. Note that the assumption that $J\sct2 \neq J\sct1 e^{\i \gamma}$, holds for generic energy-conserving unitary $J$. It is also worth noting that if instead of $S$ gate, one uses $Z$ gate, i.e., 
$J=Z\otimes \mathbb{I}\otimes \mathbb{I}$, then this assumption is not satisfied, i.e., $|\Tr(J\sct1^\dag J\sct2)|=3$. Indeed, it turns out that in this case, even though $Z$ gate breaks the $\mathbb{Z}_2$ symmetry, i.e., does not commute with  $X^{\otimes 3}$, it cannot extend  $\mathcal{G}_3$ to  $\mathcal{SV}_3^{\text{U}(1)}$.

Combing this lemma with \cref{Thm:dense} we conclude that
\begin{corollary}
The group generated by unitaries $\sqrt{\i\text{SWAP}}_{12}, \sqrt{\i\text{SWAP}}_{23}, \sqrt{\i\text{SWAP}}_{13}$ and the single-qubit $S$ gates contains a dense subgroup of  $\mathcal{SV}_3^{U(1)}$.  
\end{corollary}

\subsection{Approximate semi-universality of $S+\sqrt{\i\text{SWAP}}$  for $n\ge 3$ qubits (Proof of  \cref{cor})} \label{ss:proof_cor}

Recall that in \cref{prop13} and \cref{Thm2} 
 we studied unitaries that can be realized with  $\exp(\i R\alpha): \alpha\in(-\pi,\pi]$ gates, with and without $S$ gates, respectively. The above observation in  \cref{Thm:dense} implies that on a system with $n\ge 3$ qubits, and for any pair of qubits $i$ and $j$, the 2-qubit unitary $\exp(\i \alpha R_{ij}): \alpha\in(-\pi,\pi]$ can be realized with arbitrary small error using   gates $\sqrt{\i \text{SWAP}}$ that act on qubits $i, j$ and a third different qubit. Therefore, by combining these results we arrive at
\begin{corollary}
For any system with $n\ge 3$ qubits: 
\begin{itemize}
\item The group generated by $\sqrt{\i\text{SWAP}}$ and single-qubit $S$ gates contains a dense subgroup of  $\mathcal{SV}_n^{\text{U}(1)}$.
\item The group generated by $\sqrt{\i\text{SWAP}}$ gate is a dense subgroup of  $\mathcal{G}_n$, i.e., the group of unitaries satisfying the 4 conditions in \cref{Thm2}.
\end{itemize}
\end{corollary}
As we have seen in \cref{Thm1} and its proof in  \cref{sec:nqubit,sec:XY_only}, 
any  energy-conserving unitary on $n$ qubits can be implemented using a unitary in $\mathcal{SV}_{n+1}^{\text{U}(1)}$ and one ancilla qubit or a unitary in $\mathcal{G}_{n+2}$ and 2 ancilla qubits. In both cases we need $\mathcal{O}(4^n n^{3/2})$ gates $\exp(\i R \alpha): \alpha\in(-\pi,\pi]$. Then, to guarantee that the total error in implementing the desired $n$-qubit unitary is less than $\epsilon$, it suffices to have 
the error in implementing each single gate $\exp(\i R \alpha): \alpha\in(-\pi,\pi]$, less than or equal to $\epsilon'=\epsilon/\mathcal{O}(4^n n^{3/2})$. Then, according to our \cref{Thm:dense} this can be achieved with $\mathcal{O}(\log^{\nu}(\epsilon'^{-1}))=\mathcal{O}(\log^{\nu}(4^n n^{3/2}\epsilon^{-1}))=\mathcal{O}((n-\log\epsilon)^\nu)$ number of $\sqrt{\i\text{SWAP}}$ gates. Therefore, in total, to achieve the overall error $\epsilon$ in implementing the desired $n$-qubit unitary we need, at most, $\mathcal{O}(4^n n^{3/2}(n-\log \epsilon)^{\nu})$ number of $\sqrt{\i\text{SWAP}}$. This completes the proof of \cref{cor}.

\subsection{Proof of \cref{Thm:dense} } \label{ss:proof_Thm:dense}

Here, we explain the main steps in the proof of \cref{Thm:dense}. Further details and the proof of lemmas can be found in \cref{prooflem65,ss:proof_lemma_dense}. The first step in the proof of  \cref{Thm:dense} is the following lemma, which fully characterizes the eigen-decomposition of operator  
\be
W_{123} := \sqrt{\i\text{SWAP}}_{23}\sqrt{\i\text{SWAP}}_{13}\ .
\ee

\begin{lemma}\label{lem:dense}
The unitary  $W_{123} := \siswap_{23}\siswap_{13}$ has 3 distinct eigenvalues, namely  $\{1,e^{\pm\i\theta}\}$, where  
\be \label{eq:costheta}
\cos\frac{\theta}{2}=\cos^2 \frac{\pi}{8}=\frac{1}{2}(1+\frac{1}{\sqrt{2}})\ ,
\ee
and $\theta \approx 0.348886 \pi$ is an irrational multiple of $\pi$. Each eigenvalue $e^{\pm \i \theta}$ has multiplicity two with eigenvectors $\ket{\psi_{\pm\theta}}$ and $X^{\otimes 3}\ket{\psi_{\pm\theta}}$, where 

\begin{align} 
	\ket{\psi_{+\theta}} &:= \frac{1}{\sqrt{5-2\sqrt2}}(\ket{001} + \zeta \ket{010} + \zeta^*\ket{100}) \\
\ket{\psi_{-\theta}} &:= \frac{1}{\sqrt{5-2\sqrt2}}(\ket{001} - \zeta^* \ket{010} - \zeta\ket{100}) \ ,
\end{align}
and 
\be
 \zeta :=  %\frac{2\sqrt2}{\sqrt{4\sqrt2+7}+\i} =
		 \frac{\sqrt{5-2\sqrt2}}{2}+\frac{\sqrt2-1}{2}\i \ .
\ee
In summary, 
\be
W_{123} := \siswap_{23}\siswap_{13}=\exp(\theta A_{123})\ ,
\ee
where $A_{123} := \i(P_+ - P_-) $ and 
\begin{align} 
	P_\pm &:=  \ketbrasame{\psi_{\pm\theta}} + X^{\otimes3}\ketbrasame{\psi_{\pm\theta}}X^{\otimes3} \,.\label{eq:PplusPminus} 
 \end{align}
\end{lemma}
See \cref{ss:proof_lemma_dense} for the proof of this lemma. The fact that $\theta/\pi$ is an irrational number, i.e., $e^{\i\theta}$ is an irrational rotation, follows from the same argument that previously established \cite{boykin1999universalieee}
 the universality of $H$ and $T$ gates for a single qubit (See \cref{ss:proof_lemma_dense}).

For the following applications, we present the explicit form of $A_{123}$ in the computational basis relative to the order in \cref{ordering}:
\begin{align}
	A_{123} &= \left(\begin{matrix}
		0&&&\\& A_{123}\sct1 && \\ &&A_{123}\sct1& \\ &&&0 
\end{matrix}\right), ~
	A_{123}\sct1 := \left(\begin{matrix}
		0 & \i \alpha & \i \alpha \\ 
		\i \alpha & 0 & -\beta \\ 
		\i \alpha & \beta  & 0 
	\end{matrix}\right), \nonumber \\
	\alpha &= (5-2\sqrt2)^{-\frac12} \approx 0.678598 \nonumber \\
	\beta &= (7+4\sqrt2)^{-\frac12} \approx 0.281085
\end{align}
The significance of an irrational rotation is that its repeated application generates a dense subgroup of $\text{U}(1)$ (See, e.g.,  Chapter 4.5 of \cite{NielsenAndChuang}). Applying this result to $W_{123} = e^{\theta A_{123}}$, we conclude that 
 for any $t\in \R$ and any $\delta>0$, there exists an integer $k$ such that $\|(e^{\theta A_{123}})^k - e^{ t A_{123}}\|_\infty < \delta$. In other words, any gate in the form of $e^{t A_{123}}$  can be approximated with repeated application of $W_{123}$ to arbitrary precision.

Therefore, applying \cref{lem:dense} we know that the group generated by $W_{123}$, $W_{213}$, and $W_{312}$ is a dense subgroup of the Lie group generated by unitaries $\exp(t A_{123})$,  $\exp(t A_{213})$, $\exp(t A_{312})$, i.e., the group
\be\label{group3}
\langle \exp(t A_{123})\ , \exp(t A_{213})\ , \exp(t A_{312})\ :   t\in\mathbb{R}\rangle\  .
\ee
Clearly, because $A_{123}\sct1$ and its permuted versions, $A_{213}\sct1$ and $A_{312}\sct1$, are traceless and skew-Hermitian,  this group is a subgroup of $\mathcal{G}_3$. In the following, we prove that this group is indeed equal to  $\mathcal{G}_3$. Since $\mathcal{G}_3$ is a compact connected Lie group, isomorphic to SU(3), it suffices to show that the real Lie algebra generated by $A_{123}\sct1$, $A_{213}\sct1$, and $A_{312}\sct1$, denoted by
\be
\mathfrak{h}_3=\mathfrak{alg}_\mathbb{R}\{A_{123}\sct1, A_{213}\sct1, A_{312}\sct1 \}\ , 
\ee
is equal to the Lie algebra $\mathfrak{su}(3)$. First, note that $A_{213}\sct1$ can be obtained from  $A_{123}\sct1$ by exchanging qubits 1 and 2, which results in swapping the last  two rows and columns of $A_{123}\sct1$ (since the basis vectors are ordered as $\ket{001},\ket{010},\ket{100}$) and gives
\be
B=\frac{1}{2\beta}(A_{123}\sct1-A_{213}\sct1)=\left(\begin{matrix}
		0 & 0 & 0 \\ 
		0 & 0 & -1 \\ 
		0 & 1  & 0 
	\end{matrix}\right)\in \mathfrak{h}_3\ .
\ee 
Similarly, the permuted versions of this matrix are also in  $\mathfrak{h}_3$, i.e.,  
 $$\left(\begin{matrix}
		0 & -1 & 0 \\ 
		1 & 0 & 0 \\ 
		0 & 0  & 0 
	\end{matrix}\right)\in \mathfrak{h}_3\ 
 , \ \ \ \ \ \left(\begin{matrix}
		0 & 0 & -1 \\ 
		0 & 0 & 0 \\ 
		1 & 0  & 0 
	\end{matrix}\right) \in \mathfrak{h}_3\ 
  . $$
 These matrices generate the Lie algebra $\mathfrak{so}(3)$, i.e., the Lie algebra of real skew-symmetric matrices. Therefore, 
\be
\mathfrak{so}(3)\subset \mathfrak{h}_3  \subseteq \mathfrak{su}(3)\ .
\ee
 As we explain below, $\mathfrak{h}_3$ has other elements that are still in $\mathfrak{su}(3)$ and  not in $\mathfrak{so}(3)$. But, it is well-known that there is no Lie algebra in between $\mathfrak{so}(d)$ and $\mathfrak{su}(d)$, which immediately implies $\mathfrak{h}_3 =\mathfrak{su}(3)$.  To show this explicitly, note that 
\be
C=\frac{1}{2\alpha}(A_{123}\sct1+A_{213}\sct1)=\left(\begin{matrix}
		0 & \i & \i \\ 
		\i & 0 & 0 \\ 
		\i & 0  & 0 
	\end{matrix}\right)\in \mathfrak{h}_3\ .
\ee 
This, in particular, implies 
\be
\frac{[B,C]+C}{2}=\left(\begin{matrix}
		0 & 0 & \i \\ 
		0 & 0 & 0 \\ 
		\i & 0  & 0 
	\end{matrix}\right)\in \mathfrak{h}_3\ .
\ee
It is straightforward to show that this matrix and its permuted versions, together with matrices in  $\mathfrak{so}(3)$ generate the Lie algebra $\mathfrak{su}(3)$. 

We conclude that $\mathfrak{h}_3=\mathfrak{su}(3)$, which in turn implies the group defined in \cref{group3} is equal to $\mathcal{G}_3$. Combined with \cref{lem:dense} this implies that the group generated by $W_{123}$, $W_{312}$ and $W_{231}$ is a dense subgroup of $\mathcal{G}_3$, which proves the first part of \cref{Thm:dense}. 

The second part of the theorem, which bounds the number of the required gates, follows from the Solovay-Kitaev theorem, as stated in \cref{SK0,SK}. \\

\section{Discussion}\label{Sec:discus}

Building on the previous ideas in the quantum circuit theory \cite{kitaev2002classical, NielsenAndChuang, barenco1995elementary}, and, specifically, recent works on symmetric quantum  circuits \cite{marvian2022restrictions, marvian2023non}, we 
 introduced new circuit synthesis 
techniques for implementing all energy-conserving unitaries, or, equivalently, all U(1)-invariant unitaries, using XY interaction. We also showed how these techniques can be generalized beyond XY interaction, to all energy conserving interactions that allow qubits to exchange energy (i.e., interactions that are not diagonal in the computational basis).

In the introduction, we briefly discussed applications of energy-conserving quantum circuits for suppressing noise in quantum computers. Another important area of application of our results is quantum thermodynamics. In this field, one is often interested in implementing energy-conserving unitaries. Indeed, the resource-theoretic approach to quantum thermodynamics starts with the assumption that  all energy-conserving unitaries are free, i.e., can be implemented with negligible cost \cite{FundLimitsNature, brandao2013resource, janzing2000thermodynamic,  lostaglio2015quantumPRX, chitambar2019quantum}. In the context of this resource theory, researchers have developed various protocols  utilizing energy-conserving unitaries. However, the problem of implementing energy-conserving unitaries themselves has not been studied much. In particular, prior to this work, it was not known how a general desired energy-conserving unitary on $n$ qubits can be realized using 2-qubit energy-conserving unitaries.

In addition to quantum thermodynamics, energy-conserving unitaries also play a crucial role in the context of quantum clocks and quantum reference frames \cite{QRF_BRS_07},  covariant error correcting codes \cite{faist2020continuous, hayden2021error, kong2021charge}, and the resource theory of asymmetry. (In particular, energy-conserving unitaries are the only unitaries that can be implemented without having access to synchronized clocks.)  Other examples of applications of energy-conserving unitaries, or, equivalently circuits with U(1) symmetry, include variational quantum machine learning
\cite{meyer2023exploiting, nguyen2022theory, sauvage2022building, zheng2023speeding},  variational quantum eigensolvers for quantum chemistry \cite{barron2021preserving, shkolnikov2021avoiding, gard2020efficientsymmetry, streif2021quantum, wang2020x, barkoutsos2018quantum},  and quantum gravity \cite{nakata2023black}.

 \section*{Acknowledgments}
   This work is supported by a collaboration between the US DOE and other Agencies. This material is based upon work supported by the U.S. Department of Energy, Office of Science, National Quantum Information Science Research Centers, Quantum Systems Accelerator. Additional support is acknowledged from  NSF QLCI grant OMA-2120757,  
  NSF Phy-2046195, and ARL-ARO QCISS grant number W911NF-21-1-0005. 
 G. B. is supported partly by the Hong Kong Research Grant Council (RGC) through the Research Impact Grant R7035-21F, and partly by the National Research Foundation, Singapore and A*STAR under its CQT Bridging Grant.

\bibliographystyle{unsrt}
\bibliography{Ref_tidy,additional_refs}

\newpage

\appendix

\newpage

\newpage

\newpage

\onecolumngrid

\appendix
\newpage

\section{Constraints on the relative phases (Proof of \cref{Thmeq})}\label{App:proof}

%In this section we prove Eq.(\ref{}) and 

In this Appendix,  we review the argument of \cite{marvian2022restrictions} that shows how locality restricts realizable unitaries and, in particular,  imposes constraints on the relative phases between sectors with different energies.  See \cite{marvian2022restrictions} for further details.

In particular, we show that any unitary $V$ that can be realized with XY interaction local $Z$ Hamiltonian satisfies the constraint
\be\label{dkd}
\theta_m={{n}\choose{m}} \times \Big[\frac{m}{n}\times (\theta_n-\theta_0)+ \theta_0\Big]\ \  :\ \text{mod}\   2\pi \ , \tag{\ref{Thmeq}}
\ee  
for all $m=0,\cdots, n$.

We also show that any energy-conserving unitary $V$ satisfying this constraint has a decomposition as 
\be\label{efef}
V=e^{\i\alpha} \exp(\i \beta Z_j)  V'  \ ,
\ee
where $\alpha,\beta\in(-\pi,\pi]$, $j\in\{1,\cdots, n\}$ is any arbitrary qubit,  and  $V'\in \mathcal{SV}_n^{U(1)}$.
 Recall that according to the first part of \cref{Thm0}, which was originally proved in \cite{marvian2022restrictions} and is also proved using a constructive approach in this paper, we have
 \be \mathcal{SV}_n^{\text{U}(1)}\subset  G_{XX+YY,Z} \ , 
\ee
i.e., any unitary in $\mathcal{SV}_n^{\text{U}(1)}$ can be realized with XY and single-qubit Z Hamiltonians. In conclusion, we find that 
$G_{XX+YY,Z}$ is the subgroup of $\mathcal{V}_n^{\text{U}(1)}$ satisfying \cref{dkd}, as stated in \cref{Thm0}.

\subsection{Proof of \cref{Thmeq}}

Suppose unitary $V$ can be realized with Hamiltonians $H_1,\cdots, H_T$ which all satisfy 
\be
[H_j , \sum_{r=1}^n Z_r]=0\ .
\ee
More precisely, assume
\be\label{dec4}
V=e^{\i \gamma_T H_T} \cdots e^{\i \gamma_1 H_1} \ .
\ee
Then, $V$ is energy-conserving, i.e., decomposes as
\be
V=\bigoplus_{m=0}^n V\sct{m}\ ,
\ee
where $V\sct{m}$ is the component of $V$ in the sector with Hamming weight $m$.  
Defining
\be
\theta_m=\text{arg}(\text{det}(V\sct{m}))\ ,
\ee
one can easily show  that \cite{marvian2022restrictions}
\be
\theta_m=\sum_{j=1}^n \gamma_j \times \Tr(\Pi\sct{m} H_j)\ \ \ \ \ :\ \text{mod } 2\pi\ .
\ee
Suppose each Hamiltonian $H_1,\cdots, H_T$, is either $\mathbb{I}$, $Z_j: j=1,\cdots n$, or $R_{ij}$. Then, $\Tr(\Pi\sct{m} H_j)$ is equal to one of the followings  
\bes
\begin{align}\label{eq4}
\Tr(\Pi\sct{m})&={{n}\choose{m}}\\ 
\Tr(Z_j\Pi\sct{m})&=\frac{n-2m}{n}\times {{n}\choose{m}}\\
\Tr(R_{ij}\Pi\sct{m})&=0\ .
\end{align}
\ees
Here, the second equality follows from 
\begin{align}
\Tr(Z_j\Pi\sct{m})&=\frac{1}{n}\sum_{j'=1}^n \Tr(Z_{j'} \Pi\sct{m})\\ &=\frac{n-2m}{n}\times \Tr(\Pi\sct{m})\\ &=\frac{n-2m}{n}\times {{n}\choose{m}} \ ,
\end{align}
where in the first line we use  the permutational symmetry of $\Pi\sct{m}$, and  the second line follows from 
\be
\sum_{j'=1}^n Z_{j'}=\sum_{m=0}^n  (n-2m) \Pi\sct{m}\ .
\ee
Finally, note that $\Tr(R_{ij}\Pi\sct{m})=0$ follows from the fact that 
$\langle\textbf{b}|R_{ij}|\textbf{b}\rangle=0$ for all elements of the computational basis. 

Then, assuming in decomposition in \cref{dec4} each Hamiltonian is one of the 3 types  
 $\mathbb{I}$, $Z_r: r=1,\cdots n$, or $R_{ij}$, we conclude that
\begin{align}
\theta_m&=\sum_{j=1}^n \gamma_j \times \Tr(\Pi\sct{m} H_j)\nonumber
\\ &={{n}\choose{m}}\sum_{j\in A} \gamma_j +\frac{n-2m}{n}\times {{n}\choose{m}}\sum_{j\in B} \gamma_j \nonumber\\
&={{n}\choose{m}} \big[\alpha  +\beta \times  (1-\frac{2m}{n})\big]\  \ :\text{mod } 2\pi\ , \label{rtr}
\end{align}
where $A$ and $B$ are subsets of $\{1,\cdots, T\}$ corresponding to all $j\in \{1,\cdots, T\}$, for which  $H_j=\mathbb{I}$, and $H_j=Z_r$, respectively, and 
 we have defined 
\bes
\begin{align}
\alpha&=\sum_{j\in A} \gamma_j \\
\beta&=\sum_{j\in B} \gamma_j\ .
\end{align}
\ees
Considering $m=0$ and $m=n$, we find
\begin{align}
\theta_0&=\beta+\alpha\\
\theta_n&=-\beta+\alpha\ ,
\end{align}
which implies
\be
\alpha=\frac{\theta_0+\theta_n}{2}+ b \pi\  \ :\text{mod } 2\pi\ ,
\ee
 and 
\be
\beta=\frac{\theta_0-\theta_n}{2}+ b \pi\  \ :\text{mod } 2\pi\ ,
\ee
where $b=0,1$ is unspecified. 

Putting these values of $\alpha$ and $\beta$ in \cref{rtr} we arrive at
\begin{align}
\theta_m&={{n}\choose{m}} \big[\frac{\theta_0+\theta_n}{2}+ b \pi  +(\frac{\theta_0-\theta_n}{2}+ b \pi)\times  (1-\frac{2m}{n})\big]\\ &={{n}\choose{m}}\times  \big[\theta_0-\frac{2m}{n} 
(\frac{\theta_0-\theta_n}{2}+ b \pi)\big] \  \ :\text{mod } 2\pi\ 
\\ &={{n}\choose{m}}\times  \big[\theta_0 (1-\frac{m}{n})+\frac{m}{n}\theta_n- b\frac{m}{n} 2\pi\big]\ ,
\end{align}
where $b=0,1$. 
Finally, we note that for all $m=1,\cdots, n$, it holds that
\be
{{n}\choose{m}}\times \frac{m}{n} 2\pi=\frac{(n-1)!}{(m-1)! (n-m)!} 2\pi=0\ \ \  :\ \text{mod } 2\pi\ .
\ee
We conclude that
\be
\theta_m={{n}\choose{m}} \times \Big[\frac{m}{n}\times (\theta_n-\theta_0)+ \theta_0\Big]\ \  :\ \text{mod}\   2\pi \ ,
\ee  
which proves \cref{Thmeq}.

Finally,  note that multiplying the unitary $V$ with 
$e^{-\i\alpha} \exp(-\i \beta Z_j)$, where $j=1,\cdots, n$ is an arbitrary qubit, we obtain the energy-conserving unitary
\be
V':=e^{-\i\alpha} \exp(-\i \beta Z_j)V\ ,
\ee
with the property that $\det(V'{}\sct{m})=1$, which can be seen by applying \cref{rtr} to $\theta'_m=\text{arg}(\det(V'{}\sct{m}))$.  Therefore, $V'\in\mathcal{SV}_n^{\text{U}(1)}$, which proves \cref{efef}.  

\section{Proof of \cref{lem:decomp_HSHS}}
\label{app:proof_decomp_HSHS}

First, we show the result for the special case where $H$ is traceless and then extend the result to the general case. Any traceless Hermitian $2\times 2$ matrix can be written as
\be
H=\vec{m}\cdot \vec{\sigma}=m_x\sigma_x+m_y\sigma_y+m_z\sigma_z\ ,
\ee 
where $\vec{m}\in\mathbb{R}^3$. 
Then, 
\be
SHS^\dag=m_x\sigma_y-m_y\sigma_x+m_z\sigma_z=\vec{n}\cdot \vec{\sigma}\ ,
\ee
where $\vec{n}$ is obtained by rotating $\vec{m}$ around the z axis by angle $\pi/2$. Then, unless x and y components of $\vec{m}$ are zero, $\vec{n}\neq \pm \vec{m}$, which in turn implies $H$ does not commute with $S H S^\dag$.    Explicitly, one can check that 
\be
\Tr\big(\big[H, S H S^\dag\big] \sigma_z\big)= 4\i (m_x^2+m_y^2)\ ,
\ee
which implies the commutator is non-zero, unless $m_x=m_y=0$. 

Therefore, $\exp(\i\alpha H)$ and $\exp(\i\alpha S H S^\dag)$ are rotations around different axes. Then, by Euler decomposition, any special unitary $U\in \text{SU}(2)$ can be realized with a finite sequence of such rotations as
\begin{align} \label{rtw1}
    U= \prod_{j=1}^l \big[\exp(\i\alpha_j H)S\exp(\i\beta_j  H )S^\dag\big]\ ,
\end{align}
where $\alpha_j,\beta_j\in\mathbb{R}$, and the length of this sequence, $l$, is a constant that does not depend on $U\in \SU(2)$.

This proves the lemma when $H$ is traceless. In general, when $H$ is not traceless, the above result implies that for any $U\in \SU(2)$ there exists a sequence in the form of \cref{rtw1} such that
\be
\tilde{U}=e^{\i\phi} U=\prod_{j=1}^l \big[\exp(\i\alpha_j H)\exp(\i\beta_j S H S^\dag)\big]\ .
\ee
Now recall that any $V\in \text{SU}(2)$ can be decomposed as $V=ABA^\dag B^\dag$, for $A, B\in \text{SU}(2) $. Then, from the above result we know that there exists phases $e^{\i\gamma}$ and $e^{\i\delta}$ such that 
\be
\tilde{A}=e^{\i\gamma} A\ \ \ ,\  \tilde{B}=e^{\i\delta} B , 
\ee
have decomposition in the form of \cref{rtw1}. Then,
\be
\tilde{A}\tilde{B}\tilde{A}^\dag\tilde{B}^\dag=ABA^\dag B^\dag=V\ .
\ee
This implies that when $H$ is not traceless, any element of $V\in \SU(2)$ has a decomposition in the form \cref{rtw1} with $4\times l$ elements. This completes the proof of the lemma.

\section{Controlled-$R_z(-\frac{\pi}{2})$ using a single ancilla qubit}\label{App:pi/4}

We saw how based on  
 \cref{swap2} one can obtain a circuit for implementing controlled-$Z$ gate. Other useful unitaries can be obtained in a similar fashion.    
For instance, one can check the identity \begin{align}
&\sqrt{\iswap}_{12} \iswap_{13} \sqrt{\iswap}_{23} \iswap_{12} \sqrt{\iswap}_{13}\iswap_{23}\nonumber \\ &=(-\i)\  \exp(\i \frac{\pi}{4} Z_1Z_2)\ \exp(\i\frac{\pi}{8}Z_3(Z_1+Z_2) )\ ,
\end{align}
where  $\sqrt{\iswap}_{ij} $ denotes $\sqrt{\i\text{SWAP}}$ gate on qubits $i$ and $j$.   Setting qubit 1 to be $\ket0$, and assuming qubits 2 and 3 are in an arbitrary state $|\psi\rangle$, we obtain
\begin{align}
(-\i)\  \exp(\i \frac{\pi}{4} Z_1Z_2)\ \exp(\i\frac{\pi}{8}Z_3(Z_1+Z_2) )(|0\rangle_1|\psi\rangle_{23})=(|0\rangle_1)\otimes (-\i)\  \exp(\i \frac{\pi}{4} Z_2)\ \exp(\i\frac{\pi}{4} |0\rangle\langle0|_2\otimes Z_3 )|\psi\rangle_{23} \,.
\end{align}
Then, one can cancel the unitary $\exp(\i\frac{\pi}{4}Z_2)$ term by applying its inverse. In conclusion, we find that the unitary
\begin{align}
& \exp(-\i\frac{\pi}{4}Z_2) \sqrt{\iswap}_{12} \iswap_{13} \sqrt{\iswap}_{23} \iswap_{12} \sqrt{\iswap}_{13}\iswap_{23}
\end{align}
implements (up to a global phase) $\exp(\i\frac{\pi}{4} |0\rangle\langle0|_2\otimes Z_3)$, the anti-controlled $R_z(\frac{\pi}{2})$ gate on qubits 2 and 3, when qubit 1 is $\ket0$. Furthermore, concatenating the above gates with $\exp(-\i\frac{\pi}{4}Z_3)$, one obtain controlled $R_z(-\frac{\pi}{2})$ on qubits 2 and 3.

\section{Proof of \cref{lem:2-level_decomp} } \label{app:2-level-decomposition}
\begin{proof}

For completeness we include the proof of \cref{lem:2-level_decomp}, which follows exactly the proof of a similar result for general unitary transformations (not special unitaries), presented originally in \cite{reck1994experimental} (See also \cite{kitaev2002classical, NielsenAndChuang}). 

The proof is by induction.  For $d=2$, the proposition trivially holds since $U$ itself is in $\text{SU}(2)$. For $d > 2$, assume the proposition holds for any unitary in $\text{SU}(d-1)$. Then, for any  $U\in\text{SU}(d)$, let $\textbf{a}=(a_1,\dots,a_d)^T$ be its first column of matrix in the basis $\ket{1},\dots,\ket{d}$. 

    For any $a\neq 0$ or $b\neq 0$, define $V(a,b):=(|a|^2+|b|^2)^{-1/2}\left(\begin{matrix}
        a^* & b^* \\ -b & a
    \end{matrix}\right)$, which is an element of $\text{SU}(2)$. We define $V(0,0)$ as the identity matrix. This unitary satisfies that
    \begin{align}
        V(a,b)\left(\begin{matrix}a\\b\end{matrix}\right) = \left(\begin{matrix} \sqrt{|a|^2+|b|^2}\\0\end{matrix}\right)
    \end{align}
    Let $V(a,b)_{1,k}$ be the {2-level} unitary of $V(a,b)$ acting on the subspace spanned by $\ket{1}$ and $\ket{k}$. If we apply $V(a_1,a_2)$ on the first two components of $\textbf{a}$, we get
    \begin{align}
        V(a_1,a_2)_{1,2} \textbf{a} = \left(\begin{matrix}
            \sqrt{|a_1|^2+|a_2|^2} \\ 0 \\ a_3 \\ \vdots \\ a_d
        \end{matrix}\right)
    \end{align}
    This sets the second component of $\textbf{a}$ to zero.
    Further left-multiplying with $V\left(\sqrt{|a_1|^2+|a_2|^2}, a_3\right)_{1,3}$, we can set the third component to zero. Repeating this for every other components, we get
    \begin{align}
        &V \textbf{a} = (1,0,\dots,0)^T\\
        &V := V\left(\sqrt{\sum_{i=1}^{d-1} |a_i|^2},a_d\right)_{1,d} \dots  V(a_1,a_2)_{1,2}
    \end{align}
    noting that $\sqrt{\sum_{i=1}^{d} |a_i|^2}=1$. Since $U$ has $\textbf{a}$ as its first column, $VU$ will be a matrix whose first column is $(1,0,\dots,0)^T$. Since $VU \in \text{SU}(d)$, its first row must be a unit vector, and thus must be $(1,0,\dots,0)$. $VU$ therefore has the following block-diagonal form:
    \begin{align}
        VU = \left(\begin{matrix}
            1 &  \\  & W
        \end{matrix}\right)
    \end{align}
    where $W \in \text{SU}(d-1)$. By assumption, $W$ can be written as a product of $(d-1)(d-2)/2$ 2-level gates, and $V$ is a product of $d-1$ 2-level gates by definition, thus $U = V^\dag(1\oplus W)$ can be written as a product of $(d-1)(d-2)/2 + d-1 = d(d-1)/2$ 2-level unitaries in $\SU(d)$. This completes the proof.

\end{proof}

\section{Proof of \cref{lem65}}\label{prooflem65}

Recall that
\be\label{AppG3}
\mathcal{G}_3=\{U\in \mathcal{SV}_3^{U(1)} : 
X^{\otimes 3} UX^{\otimes 3}= U  \}\ ,
\ee
and the 
elements of $\mathcal{G}_3$ have matrix representation
\begin{align}\label{AppG32}
	U = \left(\begin{matrix}
		1 &&& \\ &U\sct1&& \\ &&U\sct1& \\ &&&1
	\end{matrix}\right)\ \ \ : U\sct1\in \text{SU}(3)\ .
\end{align}
 Let 
$\widetilde{\mathcal{G}}_3$ be the group generated by 
$\mathcal{G}_3$ and the unitary $J$. 
$\mathcal{G}_3$ contains an element $U$ in the form of \cref{AppG3}, with $U\sct1=e^{i\alpha} J\sct1^\dag$, where $e^{i\alpha}$ is a properly chosen phase such that $U\sct1$ has determinant 1. It follows that 
$\widetilde{\mathcal{G}}_3$, being generated by $J$ and $\mathcal{G}_3$, contains $UJ$, which has the form of \begin{align}\label{UJ}
	U J  = \left(\begin{matrix}
		e^{i\phi_0} &&& \\ &e^{\i\alpha}\mathbb{I} && \\ && e^{\i\alpha} J\sct1^\dag J\sct2& \\ &&& e^{i\phi_3}
	\end{matrix}\right)\ .
\end{align}
Furthermore, the assumption of the lemma implies that $e^{\i\alpha} J\sct1^\dag J\sct2$ is not a global phase, i.e., is not in the center of U(3), which means it does not commute with some elements of SU(3). Based on this observation, we show that the group $\widetilde{\mathcal{G}}_3$ generated by $J$ and $\mathcal{G}_3$ contains  $\mathcal{SV}^{U(1)}_3$. (The argument is similar to the proof of Goursat’s lemma.)

Consider the elements $\widetilde{\mathcal{G}}_3$ that are in the following form 
\begin{align}\label{normal}
	V= \left(\begin{matrix}
		e^{\i\beta_0} &&& \\ &e^{\i\beta_1}\mathbb{I}\sct1 && \\ && V\sct2& \\ &&& e^{\i\beta_3}
	\end{matrix}\right)\ \ \  : {V}\in \widetilde{\mathcal{G}}_3\ ,
\end{align}
where $\mathbb{I}\sct1$ is the identity operator on the 3D subspace with Hamming weight 1 and  $V\sct2\in \text{U}(3)$ is an arbitrary unitary on the subspace with Hamming weight 2, and $e^{\i\beta_{0,1,3}}$ are phases. 
 All such unitaries $V$ in the form of \cref{normal} constitutes a subgroup of $\widetilde{\mathcal{G}}_3$, denoted as  $\widetilde{\mathcal{N}}$.
Let $\mathcal{N}$ be the subgroup of U(3)
formed from all unitaries  $V\sct2\in\text{U}(3)$ such that there exists $V$ related to $V\sct2$ by \cref{normal} and $V \in \widetilde{\mathcal{N}}$. With this definition,  \cref{UJ} implies that
\be
e^{\i\alpha} J\sct1^\dag J\sct2\in \mathcal{N}\ .
\ee
Furthermore, it can be easily shown that with this definition,  
$\mathcal{N}$ is a normal subgroup of U(3). To see this, note that if $V\in \widetilde{\mathcal{N}}$, then for any $U\in \widetilde{\mathcal{G}}_3$, the unitary $U V U^\dag$ is also in the form given in \cref{normal}, which means $\widetilde{\mathcal{N}}$ is a normal subgroup of $\widetilde{\mathcal{G}}_3$. Moreover,  
for any $U\sct1\in \text{SU}(3)$ one can choose $U\in \mathcal{G}_3\subset \widetilde{\mathcal{G}}_3$ with decomposition in the form of \cref{AppG32}, 
 and $UVU^\dag\in\widetilde{\mathcal{N}}$ implies $U\sct1 V\sct2 U\sct1^\dag \in \mathcal{N}$. Finally, we note that any element of U(3) can be written as a global phase times an element of SU(3). 
 In summary, we conclude that if $V\sct2\in\mathcal{N}$, then for any $U\sct1\in \text{U}(3)$, it holds that $U\sct1 V\sct2 U\sct1^\dag\in \mathcal{N}$, which means $\mathcal{N}$ is a normal subgroup of U(3). 
 \color{black}
 
 Next, note that since SU(3) is a simple Lie group,  any normal subgroup of U(3) either contains SU(3) or only contains global phases, i.e., is in the center of U(3). But, we just showed that $\mathcal{N}$ contains  
$e^{\i\alpha} J\sct1^\dag J\sct2$ and the assumption of lemma implies that this unitary is not in the center of U(3). It follows that $\mathcal{N}$ contains SU(3).

This means that for any $V\sct2\in\text{SU}(3)$, there exists $V\in \widetilde{\mathcal{G}}_3$ with decomposition in the form of \cref{normal}.   
Furthermore, because SU(3) is a simple Lie group, it is equal to its commutator subgroup.  The commutator subgroup generated by elements of $\widetilde{\mathcal{G}}_3$ in the form of \cref{normal}, contains all elements in the following form
\begin{align} \label{eq:1IV1}
 \left(\begin{matrix}
		1 &&& \\ &\mathbb{I}\sct1 && \\ && V\sct2& \\ &&& 1
	\end{matrix}\right)\ \ \ : V\sct2\in \text{SU}(3)\ .
\end{align}
Therefore, we find that for any $V\sct2\in \text{SU}(3)$, there exists an element of $\widetilde{\mathcal{G}}_3$ with the  decomposition in \cref{eq:1IV1}. 
  Composing these unitaries with elements of $\mathcal{G}_3$, which are in the form given in \cref{AppG32} one obtains the entire $\mathcal{SV}^{U(1)}_3$, i.e., all unitaries in the form given in \cref{SU}. This completes the proof of the lemma.

\section{Eigenvalues of $W_{123}= \sqrt{\i\text{SWAP}}_{23}\sqrt{\i\text{SWAP}}_{13}$\\ (Proof of \cref{lem:dense})} \label{ss:proof_lemma_dense}

In this section, we study the eigen-decomposition  of operator
\be
W_{123}= \sqrt{\i\text{SWAP}}_{23}\sqrt{\i\text{SWAP}}_{13}\ .
\ee
The fact that $W_{123}$ is energy-conserving and respect the $\mathbb{Z}_2$ symetry $X^{\otimes 3} W_{123} X^{\otimes 3}=W_{123}$ implies that with  the ordering in \cref{ordering}, it has a decomposition as 
 
\begin{align}\label{eq:W_W1}
	W = \left(\begin{matrix}
		1 &&& \\ &W_{123}\sct1&& \\ &&W_{123}\sct1& \\ &&&1
	\end{matrix}\right)\,,\quad \text{where}\quad
       &W_{123}\sct1= \frac12 \left(\begin{matrix}
		1 & \i\sqrt{2}& \i \\ \i & \sqrt{2} & -1 \\ \i\sqrt{2} &0& \sqrt{2} 
	\end{matrix} \right)
\end{align}

The eigenvalues of $W_{123}\sct1$ are $\{1,e^{\pm \i\ \theta}\}$, where
\begin{align}
	e^{\i\theta} = -\frac14 + \frac{\sqrt2}{2} + \frac{\sqrt{7+4\sqrt2}}{4}\i \ ,
\end{align}
and $\theta$ satisfies \cref{eq:costheta}. The eigenvalues of $W_{123}$ are thus the same as $W_{123}\sct1$, ignoring multiplicity.

To show that $\theta/\pi$ is an irrational number, i.e., $e^{\i\theta}$ is a irrational rotation, we use an argument that was previously used in \cite{boykin1999universalieee} to show the universality of $H$ and $T$ gates. 
To link $W_{123}$ with $H$ and $T$, we compute the characteristic polynomial of $W_{123}\sct1$ as
\begin{align} \label{eq:charpolyW}
    \text{det}\left(\lambda \mathbb{I} - W_{123}\sct1\right) = (\lambda - 1) ( \lambda^2 -  \sqrt{2} \lambda + \frac{1}{2}\lambda + 1) \,.
\end{align}

On the other hand, in the proof of universality of $H$ and $T$ gates, it is shown that the gate sequence $HTHT$ generates an irrational rotation \cite{NielsenAndChuang,boykin1999universalieee}. More specifically, the characteristic polynomial of
$e^{\i \frac{3}{2} \pi} (HTHT)^2$ (where the global phase $e^{\i \frac{3}{2} \pi}$ is added to make the resulting gate in $\text{SU}(2)$) is
\begin{align} \label{eq:charpolyHTHT}
    \text{det}\left(\lambda \mathbb{I} - e^{\i \frac{3}{2} \pi} (HTHT)^2\right) = \lambda^2 -  \sqrt{2} \lambda + \frac{1}{2}\lambda + 1 \,,
\end{align}
which is a factor of the polynomial in \cref{eq:charpolyW}. The roots of this polynomial are exactly $e^{\pm\i\theta}$.

This shows that the two complex eigenvalues of $W_{123}\sct1$ (and thus $W$), $e^{\pm\i\theta}$, are exactly the eigenvalues of $e^{\i \frac{3}{2} \pi} (HTHT)^2$. The result that $HTHT$ gives an irrational rotation indicates that $e^{\pm\i\theta}$ is an irrational rotation.

The original proof of $e^{\pm\i\theta}$ being an irrational rotation uses the following lemma:
\begin{lemma}[\cite{boykin1999universalieee}] \label{lem:irrrot}
	Let $e^{\i\theta}$ be the root of a minimum monic polynomial $p(x)$ over the field of rational numbers. Then $e^{\i\theta}$ is rational rotation if and only if $p(x)$ is a cyclotomic polynomial.
\end{lemma}

The minimum monic polynomial for $e^{\i\theta}$  is $x^4 + x^3 + \frac14 x^2 + x + 1$, which is not cyclotomic \cite{boykin1999universalieee}. Therefore, by \cref{lem:irrrot}, $e^{\i\theta}$ is an irrational rotation.

\end{document}

%% file: tikzsetting.tex
\usetikzlibrary{decorations.pathreplacing, positioning, shapes.misc, calc}
\tikzset{tensor/.style={rectangle,color=black,draw=black,fill=white,
                    inner sep=1pt,minimum size=5mm}}
\tikzset{tensor2h/.style={rectangle,color=black,draw=black,fill=white,
                    inner sep=1pt,minimum width=5mm,minimum height=10mm}}
\tikzset{tensor3h/.style={rectangle,color=black,draw=black,fill=white,
                    inner sep=1pt,minimum width=5mm,minimum height=20mm}}
\tikzset{parameter/.style={rectangle,color=black,draw=black,fill=black!10,thick,
                    inner sep=1pt,minimum size=5mm}}
\tikzset{virtual/.style={rectangle,inner sep=1pt,minimum size=5mm}}
\tikzset{prepare/.style={rounded rectangle, rounded rectangle east arc=none,color=black,draw=black,fill=white,inner sep=1pt,minimum size=5mm}}
\tikzset{measure/.style={rounded rectangle, rounded rectangle west arc=none,color=black,draw=black,fill=white,inner sep=1pt,minimum size=5mm}}
\tikzset{
    triple3/.style args={[#1] in [#2] in [#3]}{
        #1,preaction={preaction={draw,#3},draw,#2}
    }
}
\tikzset{triple/.style={triple3={[line width=0.125mm,black] in [line width=2mm,white] in [line width=2.25mm,black]}}}
\tikzset{thick triple/.style={triple3={[line width=0.25mm,black] in [line width=1.75mm,white] in [line width=2.25mm,black]}}}

\tikzset{measurement/.pic={\draw (55:0.5) arc (55:125:0.5); \draw (80:0.25) -- (80:0.6);}}
\tikzset{ground/.pic={\draw[thick] (-0.15,0) -- (0.15,0);\draw[thick] (-0.10,-0.05) -- (0.10,-0.05);\draw[thick] (-0.05,-0.1) -- (0.05,-0.1);}}
\tikzset{cross/.pic={\draw (-0.1,-0.1) -- (0.1,0.1);\draw (-0.1,0.1) -- (0.1,-0.1);}}
\tikzset{x/.pic={\draw (-0.1,-0.1) -- (0.1,0.1);\draw (-0.1,0.1) -- (0.1,-0.1);}}
\tikzset{c1/.pic={\fill (0, 0) circle [radius=0.075];}}
\tikzset{c0/.pic={\fill[white] (0, 0) circle [radius=0.075]; \draw (0, 0) circle [radius=0.075];}}
\tikzset{c01/.pic={
    \fill[white] (0, 0) circle [radius=0.075];
    \draw[fill] (0,0) -- (-135:0.075) arc (-135:45:0.075) -- cycle;
    \draw (0,0) circle [radius=0.075];
    }
}
\tikzset{c10/.pic={
    \fill[white] (0, 0) circle [radius=0.075];
    \draw[fill] (0,0) -- (225:0.075) arc (225:45:0.075) -- cycle;
    \draw (0,0) circle [radius=0.075];
    }
}
\definecolor{highlightyellow}{RGB}{251,247,25}
\tikzset{
    highlighted/.style={
        preaction={
            draw,highlightyellow,-, very thick,
            double=highlightyellow,
            double distance=2\pgflinewidth,
        }
    }
}

\makeatletter

\pgfkeys{/pgf/.cd,
  rectangle corner radius north west/.initial=10pt,
  rectangle corner radius north east/.initial=10pt,
  rectangle corner radius south west/.initial=10pt,
  rectangle corner radius south east/.initial=10pt
}
\newif\ifpgf@rectanglewrc@donecorner@
\def\pgf@rectanglewithroundedcorners@docorner#1#2#3#4#5{%
  \edef\pgf@marshal{%
    \noexpand\pgfintersectionofpaths
      {%
        \noexpand\pgfpathmoveto{\noexpand\pgfpoint{\the\pgf@xa}{\the\pgf@ya}}%
        \noexpand\pgfpathlineto{\noexpand\pgfpoint{\the\pgf@x}{\the\pgf@y}}%
      }%
      {%
        \noexpand\pgfpathmoveto{\noexpand\pgfpointadd
          {\noexpand\pgfpoint{\the\pgf@xc}{\the\pgf@yc}}%
          {\noexpand\pgfpoint{#1}{#2}}}%
        \noexpand\pgfpatharc{#3}{#4}{#5}%
      }%
    }%
  \pgf@process{\pgf@marshal\pgfpointintersectionsolution{1}}%
  \pgf@process{\pgftransforminvert\pgfpointtransformed{}}%
  \pgf@rectanglewrc@donecorner@true
}
\pgfdeclareshape{rectangle with rounded corners}
{
  \inheritsavedanchors[from=rectangle] % this is nearly a rectangle
  \inheritanchor[from=rectangle]{north}
  \inheritanchor[from=rectangle]{north west}
  \inheritanchor[from=rectangle]{north east}
  \inheritanchor[from=rectangle]{center}
  \inheritanchor[from=rectangle]{west}
  \inheritanchor[from=rectangle]{east}
  \inheritanchor[from=rectangle]{mid}
  \inheritanchor[from=rectangle]{mid west}
  \inheritanchor[from=rectangle]{mid east}
  \inheritanchor[from=rectangle]{base}
  \inheritanchor[from=rectangle]{base west}
  \inheritanchor[from=rectangle]{base east}
  \inheritanchor[from=rectangle]{south}
  \inheritanchor[from=rectangle]{south west}
  \inheritanchor[from=rectangle]{south east}

  \savedmacro\cornerradiusnw{%
    \edef\cornerradiusnw{\pgfkeysvalueof{/pgf/rectangle corner radius north west}}%
  }
  \savedmacro\cornerradiusne{%
    \edef\cornerradiusne{\pgfkeysvalueof{/pgf/rectangle corner radius north east}}%
  }
  \savedmacro\cornerradiussw{%
    \edef\cornerradiussw{\pgfkeysvalueof{/pgf/rectangle corner radius south west}}%
  }
  \savedmacro\cornerradiusse{%
    \edef\cornerradiusse{\pgfkeysvalueof{/pgf/rectangle corner radius south east}}%
  }

  \backgroundpath{%
    \northeast\advance\pgf@y-\cornerradiusne\relax
    \pgfpathmoveto{}%
    \pgfpatharc{0}{90}{\cornerradiusne}%
    \northeast\pgf@ya=\pgf@y\southwest\advance\pgf@x\cornerradiusnw\relax\pgf@y=\pgf@ya
    \pgfpathlineto{}%
    \pgfpatharc{90}{180}{\cornerradiusnw}%
    \southwest\advance\pgf@y\cornerradiussw\relax
    \pgfpathlineto{}%
    \pgfpatharc{180}{270}{\cornerradiussw}%
    \northeast\pgf@xa=\pgf@x\advance\pgf@xa-\cornerradiusse\southwest\pgf@x=\pgf@xa
    \pgfpathlineto{}%
    \pgfpatharc{270}{360}{\cornerradiusse}%
    \northeast\advance\pgf@y-\cornerradiusne\relax
    \pgfpathlineto{}%
    \pgfpathclose
  }

  \anchor{before north east}{\northeast\advance\pgf@y-\cornerradiusne}
  \anchor{after north east}{\northeast\advance\pgf@x-\cornerradiusne}
  \anchor{before north west}{\southwest\pgf@xa=\pgf@x\advance\pgf@xa\cornerradiusnw
    \northeast\pgf@x=\pgf@xa}
  \anchor{after north west}{\northeast\pgf@ya=\pgf@y\advance\pgf@ya-\cornerradiusnw
    \southwest\pgf@y=\pgf@ya}
  \anchor{before south west}{\southwest\advance\pgf@y\cornerradiussw}
  \anchor{after south west}{\southwest\advance\pgf@x\cornerradiussw}
  \anchor{before south east}{\northeast\pgf@xa=\pgf@x\advance\pgf@xa-\cornerradiusse
    \southwest\pgf@x=\pgf@xa}
  \anchor{after south east}{\southwest\pgf@ya=\pgf@y\advance\pgf@ya\cornerradiusse
    \northeast\pgf@y=\pgf@ya}

  \anchorborder{%
    \pgf@xb=\pgf@x% xb/yb is target
    \pgf@yb=\pgf@y%
    \southwest%
    \pgf@xa=\pgf@x% xa/ya is se
    \pgf@ya=\pgf@y%
    \northeast%
    \advance\pgf@x by-\pgf@xa%
    \advance\pgf@y by-\pgf@ya%
    \pgf@xc=.5\pgf@x% x/y is half width/height
    \pgf@yc=.5\pgf@y%
    \advance\pgf@xa by\pgf@xc% xa/ya becomes center
    \advance\pgf@ya by\pgf@yc%
    \edef\pgf@marshal{%
      \noexpand\pgfpointborderrectangle
      {\noexpand\pgfqpoint{\the\pgf@xb}{\the\pgf@yb}}
      {\noexpand\pgfqpoint{\the\pgf@xc}{\the\pgf@yc}}%
    }%
    \pgf@process{\pgf@marshal}%
    \advance\pgf@x by\pgf@xa%
    \advance\pgf@y by\pgf@ya%
    \pgfextract@process\borderpoint{}%
    \pgf@rectanglewrc@donecorner@false
    \southwest\pgf@xc=\pgf@x\pgf@yc=\pgf@y
    \advance\pgf@xc\cornerradiussw\relax\advance\pgf@yc\cornerradiussw\relax
    \borderpoint
    \ifdim\pgf@x<\pgf@xc\relax\ifdim\pgf@y<\pgf@yc\relax
      \pgf@rectanglewithroundedcorners@docorner{-\cornerradiussw}{0pt}{180}{270}{\cornerradiussw}%
    \fi\fi
    \ifpgf@rectanglewrc@donecorner@\else
      \southwest\pgf@yc=\pgf@y\relax\northeast\pgf@xc=\pgf@x\relax
      \advance\pgf@xc-\cornerradiusse\relax\advance\pgf@yc\cornerradiusse\relax
      \borderpoint
      \ifdim\pgf@x>\pgf@xc\relax\ifdim\pgf@y<\pgf@yc\relax
       \pgf@rectanglewithroundedcorners@docorner{0pt}{-\cornerradiusse}{270}{360}{\cornerradiusse}%
      \fi\fi
    \fi
    \ifpgf@rectanglewrc@donecorner@\else
      \northeast\pgf@xc=\pgf@x\relax\pgf@yc=\pgf@y\relax
      \advance\pgf@xc-\cornerradiusne\relax\advance\pgf@yc-\cornerradiusne\relax
      \borderpoint
      \ifdim\pgf@x>\pgf@xc\relax\ifdim\pgf@y>\pgf@yc\relax
       \pgf@rectanglewithroundedcorners@docorner{\cornerradiusne}{0pt}{0}{90}{\cornerradiusne}%
      \fi\fi
    \fi
    \ifpgf@rectanglewrc@donecorner@\else
      \northeast\pgf@yc=\pgf@y\relax\southwest\pgf@xc=\pgf@x\relax
      \advance\pgf@xc\cornerradiusnw\relax\advance\pgf@yc-\cornerradiusnw\relax
      \borderpoint
      \ifdim\pgf@x<\pgf@xc\relax\ifdim\pgf@y>\pgf@yc\relax
       \pgf@rectanglewithroundedcorners@docorner{0pt}{\cornerradiusnw}{90}{180}{\cornerradiusnw}%
      \fi\fi
    \fi
  }
}

\makeatother

\tikzset{prepare tall/.style={rectangle with rounded corners, rectangle corner radius north east=0pt, rectangle corner radius south east=0pt, rectangle corner radius north west=30pt, rectangle corner radius south west=30pt, color=black,draw=black,fill=white,thick,inner sep=1pt,minimum height=22mm,minimum width=12mm}}
\tikzset{measure tall/.style={rectangle with rounded corners, rectangle corner radius north west=0pt, rectangle corner radius south west=0pt, rectangle corner radius north east=30pt, rectangle corner radius south east=30pt, color=black,draw=black,fill=white,thick,inner sep=1pt,minimum height=22mm,minimum width=12mm}}